\newcommand{\CLASSINPUTtoptextmargin}{3.4 cm}
\newcommand{\CLASSINPUTbottomtextmargin}{1.6cm}
\documentclass[10pt,twocolumn]{IEEEtran}
\ifCLASSINFOpdf
\else
\fi
\usepackage{amsmath,amssymb,latexsym, cite}
\usepackage{bbm}
\usepackage{epsf, pgf, xcolor}
\usepackage{pstricks}
\usepackage{enumerate}
\usepackage{algorithm, algorithmic}
\usepackage{booktabs, ctable}
\usepackage{subfigure}
\usepackage{amsfonts}
\usepackage{bm}
\usepackage{epsfig}
\usepackage{array, tabularx}
\usepackage{pst-node}
\usepackage{cases}
\usepackage{etex}
\usepackage{pmat}
\def\argmax{\operatornamewithlimits{arg\,max}}

\newcommand{\degBC}{\overset{\textsc{~}}{\succ}}

\newrgbcolor{Sgray}{0.981 0.981 0.981}
\newrgbcolor{Rgray}{0.94 0.94 0.94}
\newrgbcolor{Dgray}{0.8 0.8 0.8}

\newcommand{\indicator}[1]{\mathbbm{1}_{\left[ {#1} \right] }}
\newcommand{\matbold}[1]{\ensuremath{\mathbf{#1}}}

\def\os2{\sigma_S^2}
\def\on2{\sigma_N^2}
\def\oz2{\sigma_Z^2}

\newtheorem{theorem}{Theorem}

\newtheorem{lemma}{Lemma}

\newtheorem{definition}{Definition}
\newtheorem{example}{Example}
\newtheorem{remark}{Remark}

\newtheorem{proposition}{Proposition}

\IEEEoverridecommandlockouts

\begin{document}

%\title{The Computation Capacity of a Simple Multi-Receiver Network Model}
\title{Polar Codes for Broadcast Channels}

%~\IEEEmembership{Student Member,~IEEE,},

\author{Naveen~Goela$^{\dagger}$, Emmanuel~Abbe$^{\sharp}$, and~Michael~Gastpar$^{\dagger}$%~\IEEEmembership{Member,~IEEE}
% <-this % stops a space
% This work was supported in part by the National Science Foundation under Grant CCF-0627024.
\thanks{This work was presented in part at the \emph{International Zurich Seminar on Communications}, Zurich, Switzerland on March 1, 2012, and submitted in part to the \emph{IEEE International Symposium on Information Theory} on January, 2013.}
%\thanks{Copyright~\copyright\,2012 IEEE. Personal use of this material is permitted. However,
%permission to use this material for any other purposes must be obtained from the
%IEEE by sending a request to pubs-permissions@ieee.org.}
\thanks{$^{\dagger}$N. Goela and M. C. Gastpar are with the Department of Electrical Engineering and Computer Science, University of California, Berkeley, Berkeley, CA 94720-1770 USA (e-mail: \{ngoela, gastpar\}@eecs.berkeley.edu) and also with the School of Computer and Communication Sciences, Ecole Polytechnique F\'ed\'erale (EPFL), Lausanne, Switzerland (e-mail: \{naveen.goela, michael.gastpar\}@epfl.ch).}% <-this % stops a space
\thanks{$^{\sharp}$E. Abbe was with the School of Computer and Communication Sciences, Ecole Polytechnique F\'ed\'erale (EPFL), Lausanne, Switzerland, and is currently with the School of Engineering and Applied Sciences, Princeton University, Princeton, NJ, 08544 USA (e-mail: \{eabbe@princeton.edu\}).}
%(e-mail: \{emmanuel.abbe\}@epfl.ch).}
%\thanks{$^{\dagger}$ N. Goela and M. Gastpar are with the School of Computer and Communication Sciences, Ecole Polytechnique %F\'ed\'erale (EPFL), Lausanne, Switzerland (e-mail: \{naveen.goela, michael.gastpar\}@epfl.ch). N. Goela and M. Gastpar are %also with the Department of Electrical Engineering and Computer Science, University of California, Berkeley, Berkeley, CA %94720-1770 USA (e-mail: \{ngoela, gastpar\}@eecs.berkeley.edu).}% <-this % stops a space
}

\maketitle

\begin{abstract}
%\boldmath
% from Ar\i kan's oint-to-point binary-input discrete, memoryless channels,
%Low-complexity capacity-achieving polar codes are introduced for classes of discrete memoryless broadcast channels (DM-BCs). For non-linear and linear $m$-user binary-output deterministic channels, polarization-based codes achieve optimal rates within the private-message capacity region. By contrast to random binning strategies requiring exponential complexity, polarization-based codes provide encoding and decoding complexities of $\mathcal{O}(n \log n)$ where $n$ is the code length, with an average block error probability of $\mathcal{O}(2^{-n^{\beta}})$ where $0 < \beta < \frac{1}{2}$. The proof characterizing the asymptotic behavior of the error probability is based on the chain rule for Kullback-Leibler distance. For noisy broadcast channels with $m = 2$ users, low-complexity polar codes are developed according to two key information-theoretic designs: i) Cover's superposition strategy; ii) Marton's coding scheme. Constraints on the auxiliary and input distributions for these codes ensure a proper alignment of polarization indices in the multi-user setting. The codes achieve capacity-optimal rate-tuples for classes of DM-BCs (e.g. stochastically degraded, semi-deterministic). Experiments for finite code lengths $n$ corroborate the theory.
%binary-output linear and non-linear
Polar codes are introduced for discrete memoryless broadcast channels. For $m$-user \emph{deterministic} broadcast channels, polarization is applied to map uniformly random message bits from $m$ independent messages to one codeword while satisfying broadcast constraints. The polarization-based codes achieve rates on the boundary of the private-message capacity region. For two-user \emph{noisy} broadcast channels, polar implementations are presented for two information-theoretic schemes: i) Cover's superposition codes; ii) Marton's codes. Due to the structure of polarization, constraints on the auxiliary and channel-input distributions are identified to ensure proper alignment of polarization indices in the multi-user setting. The codes achieve rates on the capacity boundary of a few classes of broadcast channels (e.g., binary-input stochastically degraded). The complexity of encoding and decoding is $\mathcal{O}(n \log n)$ where $n$ is the block length. In addition, polar code sequences obtain a stretched-exponential decay of $\mathcal{O}(2^{-n^{\beta}})$ of the average block error probability where $0 < \beta < \frac{1}{2}$.

\end{abstract}


% \fi
\IEEEpeerreviewmaketitle

\section{Introduction}

\IEEEPARstart{I}{troduced} by T. M. Cover in 1972, the broadcast problem consists of a single source transmitting $m$ independent private messages to $m$ receivers through a single discrete, memoryless, broadcast channel (DM-BC)~\cite{cover72}. The private-message capacity region is known if the channel structure is \emph{deterministic}, \emph{degraded}, \emph{less-noisy}, or \emph{more-capable}~\cite{elgamalkim2010}. For general classes of DM-BCs, there exist inner bounds such as Marton's inner bound~\cite{marton79} and outer bounds such as the Nair-El-Gamal outer bound~\cite{nairelgamal07}. One difficult aspect of the broadcast problem is to design an encoder which maps $m$ independent messages to a single codeword of symbols which are transmitted simultaneously to all receivers. Several codes relying on \emph{random binning}, \emph{superposition}, and \emph{Marton's strategy} have been analyzed in the literature (see e.g., the overview in~\cite{cover98}).
%~\cite[Chap. 5]{elgamalkim2010}

\subsection{Overview of Contributions}

The present paper focuses on low-complexity codes for broadcast channels based on polarization methods. Polar codes were invented originally by Ar\i kan and were shown to achieve the capacity of binary-input, symmetric, point-to-point channels with $\mathcal{O}(n\log n)$ encoding and decoding complexity where $n$ is the code length~\cite{arikan09}. In this paper, we obtain the following results.
\begin{itemize}
\item Polar codes for deterministic, linear and non-linear, binary-output, $m$-user DM-BCs (cf.~\cite{goelaIZS2012}). The capacity-achieving broadcast codes implement low-complexity \emph{random binning}, and are related to polar codes for other multi-user scenarios such as Slepian-Wolf distributed source coding~\cite{arikan10,arikan2012}, and multiple-access channel (MAC) coding~\cite{abbeIT2012}. For deterministic DM-BCs, the polar transform is applied to channel \emph{output} variables. Polarization is useful for shaping uniformly random message bits from $m$ independent messages into non-equiprobable codeword symbols in the presence of hard broadcast constraints. As discussed in Section~\ref{sec:RelatedWorkDeterministicChannels} and referenced in~\cite{aleksic05,colemanII05,braunstein07}, it is difficult to design low-complexity parity-check (LDPC) codes or belief propagation algorithms for the deterministic DM-BC due to multi-user broadcast constraints.
\item Polar codes for general two-user DM-BCs based on \emph{Cover's superposition coding} strategy. In the multi-user setting, constraints on the auxiliary and channel-input distributions are placed to ensure alignment of polarization indices. The achievable rates lie on the boundary of the capacity region for certain classes of DM-BCs such as binary-input stochastically degraded channels.
\item Polar codes for general two-user DM-BCs based on \emph{Marton's coding} strategy. In the multi-user setting, due to the structure of polarization, constraints on the auxiliary and channel-input distributions are identified to ensure alignment of polarization indices. The achievable rates lie on the boundary of the capacity region for certain classes of DM-BCs such as binary-input semi-deterministic channels.
\item For the above broadcast polar codes, the asymptotic decay of the average error probability under successive cancelation decoding at the broadcast receivers is established to be $\mathcal{O}(2^{-n^{\beta}})$ where $0 < \beta < \frac{1}{2}$. The error probability is analyzed by averaging over polar code ensembles. In addition, properties such as the chain rule of the Kullback-Leibler divergence between discrete probability measures are exploited.  %Deterministic encoders and decoders are also specified.
\end{itemize}

Throughout the paper, for different broadcast coding strategies, a systems-level block diagram of the communication channel and polar transforms is provided.

%Broadcast codes developed using polarization methods have the potential to go beyond

%Experiments for finite code lengths $n$ corroborate the theory.
%\item A communication systems-level approach is developed for the application of polar codes to multi-user %channels with future potential in multi-user \emph{networks}.
%via the probabilistic method. We now provide several details of the above contributions.
%For deterministic DM-BCs, the polar transform is applied to channel \emph{output} variables instead of input %variables.  In the broadcast setting, it is difficult to design conventional low-density parity-check (LDPC) %codes to satisfy the broadcast constraints. A polarization-based approach could be useful for other %constraint-satisfaction and interference-management problems as well.

%Furthermore, the coding scheme generalizes to the case of transmitting \emph{common} messages to subsets of $m$ broadcast receivers. It %is shown that \emph{joint} polarization of correlated random variables allows for private and common message rates which are optimal %under specific conditions even though the associated capacity-region is an open question in theory.
%Recently, polar alignment strategies were introduced for a limited class of many-to-one degraded additive interference %channels~\cite{appaiah2011}.

\subsection{Relation to Prior Work}\label{sec:RelatedPriorWork}

\subsubsection{Deterministic Broadcast Channels}\label{sec:RelatedWorkDeterministicChannels}

The deterministic broadcast channel has received considerable attention in the literature (e.g. due to related extensions such as secure broadcast, broadcasting with side information, and index coding~\cite{yossef2011,rouayheb10}). Several \emph{practical} codes have been designed. For example, the authors of \cite{aleksic05} propose sparse linear coset codes to emulate random binning and survey propagation to enforce broadcast channel constraints. In~\cite{colemanII05}, the authors propose enumerative source coding and Luby-Transform codes for deterministic DM-BCs specialized to interference-management scenarios. Additional research includes reinforced belief propagation with non-linear coding~\cite{braunstein07}. To our knowledge, polarization-based codes provide provable guarantees for achieving rates on the capacity-boundary in the general case.
%Functional and random coding dualities have been characterized for the following communication settings: %deterministic DM-BC, Slepian-Wolf distributed source compression, and deterministic MAC~\cite{vladimir06}. As a %result, there exist both similarities and differences for polar coding in multi-user settings.
\vspace{0.01in}
\subsubsection{Polar Codes for Multi-User Settings}

Subsequent to the derivation of channel polarization in~\cite{arikan09} and the refined rate of polarization in~\cite{arikantelatar09}, polarization methods have been extended to analyze multi-user information theory problems. In~\cite{abbeIT2012}, a joint polarization method is proposed for $m$-user MACs with connections to matroid theory. Polar codes were extended for several other multi-user settings: arbitrarily-permuted parallel channels~\cite{hof2010}, degraded relay channels~\cite{karzandIZS2012}, cooperative relaying~\cite{skoglund2012}, and wiretap channels~\cite{skoglundletters10,vardysecrecypolar2011,koyluoglu2012}. In addition, several binary multi-user communication scenarios including the Gelfand-Pinsker problem, and Wyner-Ziv problem were analyzed in~\cite[Chapter~4]{koradaphd09}. Polar codes for lossless and lossy source compression were investigated respectively in~\cite{arikan10} and~\cite{korada10}. In~\cite{arikan10}, source polarization was extended to the Slepian-Wolf problem involving distributed sources. The approach is based on an ``onion-peeling'' encoding of sources, whereas a joint encoding is proposed in~\cite{abbeITA11}. In~\cite{arikan2012}, a unified approach is provided for the Slepian-Wolf problem based on generalized monotone chain rules of entropy. To our knowledge, the design of polarization-based broadcast codes is relatively new.

\vspace{0.01in}
\subsubsection{Binary vs. $q$-ary Polarization}

The broadcast codes constructed in the present paper for DM-BCs are based on polarization for binary random variables. However, in extending to arbitrary alphabet sizes, a large body of prior work exists and has focused on generalized constructions and kernels~\cite{koradasasogluurbanke2010}, and generalized polarization for $q$-ary random variables and $q$-ary channels~\cite{sasoglu2009,moritanaka2010,sahebi2011,barg2013}. The reader is also referred to the monograph in~\cite{erenmonograph12} containing a clear overview of polarization methods.

%\cite{sasoglu2009} Arbitrary Discrete Memoryless Channels
%\cite{barg2013} Barg q-ary channel polarization
%\cite{sahebi2011} multi-level polarization
%\cite{moritanaka2010} q-ary channels and arbitrary kernels
%\cite{koradasasogluurbanke2010} exponent bounds and generalized constructions

\subsection{Notation}
%To index into a matrix array of random variables, the notation $Y(\mathcal{I}, \mathcal{J}) \triangleq \{Y(i,j) : i \in %\mathcal{I}, j \in \mathcal{J}\}$. The notation $Y([m], [n]) \triangleq \{Y(i,j) : i \in [m], j \in [n]\}$.
An index set $\{1, 2, \ldots, m\}$ is abbreviated as $[m]$. An $m \times n$ matrix array of random variables is comprised of variables $Y_i(j)$ where $i \in [m]$ represents the row and $j \in [n]$ the column. The notation $Y_i^{k:\ell} \triangleq \{Y_i(k), Y_i(k+1), \ldots, Y_i(\ell)\}$ for $k \leq \ell$. When clear by context, the term $Y_i^{n}$ represents $Y_i^{1:n}$. In addition, the notation for the random variable $Y_i(j)$ is used interchangeably with $Y_i^{j}$. The notation $f(n) = \mathcal{O}(g(n))$ means that there exists a constant $\kappa$ such that $f(n) \leq \kappa g(n)$ for sufficiently large $n$. For a set $\mathcal{S}$, $\operatorname{clo}(\mathcal{S})$ represents set closure, and $\operatorname{co}(\mathcal{S})$ the convex hull operation over set $\mathcal{S}$. Let $h_b(x) = -x\log_2(x) - (1-x)\log_2(1-x)$ denote the binary entropy function. Let $a * b \triangleq (1-a)b + a(1-b)$.

\input{./BlackwellChannel}

\section{Model}

%\subsection{Definitions}

\begin{definition}[Discrete, Memoryless Broadcast Channel]\label{def:DMBC} The discrete memoryless broadcast channel (DM-BC) with $m$ broadcast receivers consists of a discrete input alphabet $\mathcal{X}$, discrete output alphabets $\mathcal{Y}_i$ for $i \in [m]$, and a conditional distribution $P_{Y_1, Y_2, \ldots, Y_m|X}(y_1, y_2, \ldots, y_m| x)$ where $x \in \mathcal{X}$ and $y_i \in \mathcal{Y}_i$.
% where $X \in \mathcal{X}$ and $Y_i \in \mathcal{Y}_i$.
\end{definition}

\vspace{0.05in}
\begin{definition}[Private Messages]\label{def:SourceInformation} For a DM-BC with $m$ broadcast receivers, there exist $m$ private messages $\{W_i\}_{i \in [m]}$ such that each message $W_i$ is composed of $nR_i$ bits and $(W_1, W_2, \ldots, W_m)$ is uniformly distributed over $[2^{nR_1}] \times [2^{nR_2}] \times \cdot\cdot\cdot \times [2^{nR_m}]$.
\end{definition}

\vspace{0.05in}
\begin{definition}[Channel Encoding and Decoding]\label{def:ChannelEncoderDecoder} For the DM-BC with independent messages, let the vector of rates $\vec{R} \triangleq \left[\begin{array}{cccc} R_1 & R_2 & \ldots & R_m \end{array}\right]^{T}$. An $(\vec{R}, n)$ code for the DM-BC consists of one encoder
\begin{align}
x^{n}: [2^{nR_1}] \times [2^{nR_2}] \times \cdot\cdot\cdot \times [2^{nR_m}] \rightarrow \mathcal{X}^{n}, \notag
\end{align}
and $m$ decoders specified by $\hat{W}_i: \mathcal{Y}_{i}^{n} \rightarrow [2^{nR_i}]$ for $i \in [m]$. Based on received observations $\{Y_i(j)\}_{j \in [n]}$, each decoder outputs a decoded message $\hat{W}_i$.
\end{definition}

\vspace{0.05in}
\begin{definition}[Average Probability of Error]\label{def:AvgProbError} The average probability of error $P_e^{(n)}$ for a DM-BC code is defined to be the probability that the decoded message at all receivers is not equal to the transmitted message,
\begin{align} % \mathbf{Pr}
P_e^{(n)} & =  \mathbb{P}\left\{ \bigvee_{i=1}^{m} \hat{W}_i\left(\{Y_i(j)\}_{j \in [n]}\right) \neq W_i \right\}. \notag
%P_e^{(n)} & =  \mbox{Pr}\left\{ \bigvee_{i=1}^{m} \hat{W}_i(\{Y_{i,j}\}_{j \in [n]} \neq W_i \right\}. \notag
\end{align}
\end{definition}

\vspace{0.05in}
\begin{definition}[Private-Message Capacity Region]\label{def:AchievableRateRegionAndCapacity} If there exists a sequence of $(\vec{R}, n)$ codes with $P_e^{(n)} \rightarrow 0$, then the rates $\vec{R} \in \mathbb{R}_{+}^{m}$ are achievable. The private-message capacity region is the closure of the set of achievable rates.
\end{definition}
%\begin{definition}[Capacity Region]\label{def:CapacityDMBC}
%\end{definition}
%\input{./Figures/private_capacity_region_blackwell}

\section{Deterministic Broadcast Channels}
\vspace{0.1in}
\begin{definition}[Deterministic DM-BC]\label{def:DeterministicDMBC} Define $m$ deterministic functions $f_i(x): \mathcal{X} \rightarrow \mathcal{Y}_i$ for $i \in [m]$. The deterministic DM-BC with $m$ receivers is defined by the following conditional distribution
\begin{align}
P_{Y_1, Y_2, \ldots, Y_m|X}(y_1, y_2, \ldots, y_m| x) = \prod_{i=1}^{m} \indicator{y_i = f_i(x)}.
\end{align}
\end{definition}
\vspace{-0.2in}
\subsection{Capacity Region}

\begin{proposition}[Marton~\cite{marton77}, Pinsker~\cite{pinsker78}] The capacity region of the deterministic DM-BC includes those rate-tuples $\vec{R} \in \mathbb{R}_{+}^{m}$ in the region
\begin{align}
\mathfrak{C}_{DET\!-\!BC} & \triangleq \operatorname{co}\Bigl(\operatorname{clo}\Bigl( \bigcup_{X, \{Y_i\}_{i \in [m]}} \mathfrak{R}\bigl(X, \{Y_i\}_{i \in [m]}\bigl) \Bigl) \Bigl), \label{eqn:CapacityExpressionOfDBC}
\end{align}
where the polyhedral region $\mathfrak{R}(X, \{Y_i\}_{i \in [m]})$ is given by
\begin{align}
\mathfrak{R} \triangleq \Bigl\{\vec{R}~\Bigl|~ \sum_{i \in \mathcal{S}} R_i < H(\{Y_i\}_{i \in \mathcal{S}}), ~\forall \mathcal{S} \subseteq [m]\Bigl\}. \label{eqn:PolyhedronRateRegion}
\end{align}
The union in Eqn.~\eqref{eqn:CapacityExpressionOfDBC} is over all random variables $X, Y_1, Y_2, \ldots, Y_m$ with joint distribution induced by $P_{X}(x)$ and $Y_i = f_i(X)$.
\end{proposition}
\vspace{0.05in}
\begin{example}[Blackwell Channel]\label{ex:BlackwellChannel} In Figure~\ref{fig:BlackwellChannel}, the Blackwell channel is depicted with $\mathcal{X} = \{0,1,2\}$ and $\mathcal{Y}_i = \{0,1\}$. For any fixed distribution $P_X(x)$, it is seen that $P_{Y_1Y_2}(y_1, y_2)$ has \emph{zero mass} for the pair $(1,0)$. Let $\alpha \in [\frac{1}{2}, \frac{2}{3}]$. Due to the symmetry of this channel, the capacity region is the union of two regions,
\begin{align}
\{(R_1, R_2):~& R_1 \leq h_b(\alpha), R_2 \leq h_b(\frac{\alpha}{2}), \notag \\
& R_1 + R_2 \leq h_b(\alpha) + \alpha\}, \notag \\
\{(R_1, R_2):~& R_1 \leq h_b(\frac{\alpha}{2}), R_2 \leq h_b(\alpha), \notag \\
& R_1 + R_2 \leq h_b(\alpha) + \alpha\}, \notag
\end{align}
where the first region is achieved with input distribution $P_X(0) = P_X(1) = \frac{\alpha}{2}$, and the second region is achieved with $P_X(1) = P_X(2) = \frac{\alpha}{2}$~\cite[Lec.~9]{elgamalkim2010}. The sum rate is maximized for a uniform input distribution which yields a pentagonal achievable rate region: $R_1 \leq h_b(\frac{1}{3})$, $R_2 \leq h_b(\frac{1}{3})$, $R_1 + R_2 \leq \log_2 3$. Figure~\ref{fig:BlackwellChannel} illustrates the capacity region.
\end{example}
%\begin{remark} For a fixed independent and identically-distributed ($i.i.d.$) input distribution $P_X(x)$, the achievable rate region %of the DBC is a polyhedron in $\mathbb{R}_{+}^{m}$ given by Eqn.~\eqref{eqn:PolyhedronRateRegion}.
%\end{remark}

%\subsection{Simplifying Assumptions}\label{sec:Simplifying}

%We make the following simplifications to Definition~\ref{def:DBC}: (a) The number of receivers $m = 2$; (b) Polar code design is for achieving vertices of the polyhedron given in Eqn.~\eqref{eqn:PolyhedronRateRegion} for fixed $P_X(x)$; (c) The output alphabets $\mathcal{Y}_i$ are all $\mathbb{F}_2$. The first two simplifications are without loss of generality. In particular (a) is for compactness of notation and our coding scheme applies for arbitrary $m$, and for (b) it was shown that an arbitrary point in the achievable rate region of a DBC with $m$ receivers maps to a corresponding vertex point in the achievable region of a DBC with $(2m-1)$ receivers if rate-splitting is used~\cite{coleman05}. Item (c) is a mild restriction due to the algebraic construction of polar codes. It is possible to generalize the polar code construction to include output alphabets of prime, prime power, and different cardinalities.
%~\cite{arikantelatar09}.
\subsection{Main Result}
%\vspace{0.05in}
\begin{theorem}[Polar Code for Deterministic DM-BC]\label{thm:DetDMBC} Consider an $m$-user deterministic DM-BC with arbitrary discrete input alphabet $\mathcal{X}$, and binary output alphabets $Y_i \in \{0,1\}$. Fix input distribution $P_X(x)$ where $x \in \mathcal{X}$ and constant $0 < \beta < \frac{1}{2}$. Let $\pi: [m] \rightarrow [m]$ be a permutation on the index set of receivers. Let the vector
\begin{align}
\vec{R} & \triangleq \left[\begin{array}{cccc} R_{\pi(1)} & R_{\pi(2)} & \ldots & R_{\pi(m)} \end{array}\right]^{T}. \notag
\end{align}
There exists a sequence of polar broadcast codes over $n$ channel uses which achieves rates $\vec{R}$ where the rate for receiver $\pi(i) \in [m]$ is bounded as
%where $\vec{R} = (R_{\pi(1)}, R_{\pi(2)}, \ldots, R_{\pi(m)})$ is given by
%\begin{align}
%R_{\pi(i)} & < H(Y_{\pi(j)}|Y_{\pi(1)}, Y_{\pi(2)}, \ldots, Y_{\pi(i-1)}), \forall i \in [m]. \notag
%\end{align}
%\begin{align}
%\vec{R} & = \left[\begin{array}{c} R_{\pi(1)} \\ R_{\pi(2)} \\ \vdots \\ R_{\pi(m)} \end{array}\right], \notag
%\end{align}
\begin{align}
0 \leq R_{\pi(i)} < H\left(Y_{\pi(i)} | \{Y_{\pi(k)}\}_{k=1:i-1}\right). \notag
\end{align}
%\begin{align}
%R_{\pi(1)} & < H(Y_{\pi(1)}), \notag \\
%R_{\pi(2)} & < H(Y_{\pi(2)}|Y_{\pi(1)}), \notag \\
%R_{\pi(3)} & < H(Y_{\pi(3)}|Y_{\pi(1)}, Y_{\pi(2)}), \notag \\
%\vdots & ~~~~~~~ \vdots \notag \\
%R_{\pi(m)} & < H(Y_{\pi(m)}|Y_{\pi(1)}, Y_{\pi(2)}, \ldots, Y_{\pi(m-1)}). \notag
%\end{align}
The average error probability of this code sequence decays as $P_e^{(n)} = \mathcal{O}(2^{-n^{\beta}})$. The complexity of encoding and decoding is $\mathcal{O}(n \log n)$.
\end{theorem}
\vspace{0.15in}
\begin{remark} To prove the existence of \emph{low-complexity} broadcast codes, a successive randomized protocol is introduced in Section~\ref{sec:EncodingDecoding} which utilizes $o(n)$ bits of randomness at the encoder. A deterministic encoding protocol is also presented.
\end{remark}
\begin{remark} The achievable rates for a fixed input distribution $P_X(x)$ are the vertex points of the polyhedral rate region defined in~\eqref{eqn:PolyhedronRateRegion}. To achieve non-vertex points, the following coding strategies could be applied: time-sharing; rate-splitting for the deterministic DM-BC~\cite{coleman05}; polarization by Ar{\i}kan utilizing generalized chain rules of entropy~\cite{arikan2012}. For certain input distributions $P_X(x)$, as illustrated in Figure~\ref{fig:BlackwellChannel} for the Blackwell channel, a subset of the achievable vertex points lie on the capacity boundary.
\end{remark}
\begin{remark} Polarization of channels and sources extends to $q$-ary alphabets (see e.g.~\cite{sasoglu2009}). Similarly, it is entirely possible to extend Theorem~\ref{thm:DetDMBC} to include DM-BCs with $q$-ary output alphabets.
\end{remark}

\section{Overview of Polarization Method \\ For Deterministic DM-BCs}\label{sec:PolarizationTheorems}
%$\{X_{i,j}\}_{\begin{subarray}{c}
%        i \in [m], j \in [n]
%      \end{subarray}}$%
%
%$H(X_{1:m, 1:j-1})$%
%
%$H(X_{[m], [n]})$%
%
%$X(1:m, 1:j-1)$%%
%
%$\{X_i[j]\}_{i \in [m]}$
%Using the simplifications of Sec.~\ref{sec:Simplifying}, we assume $m = 2$ receivers.
%Before explaining the proof
%For the proof, the ordering of the receivers' rates in $\vec{R}$ is arbitrary due to symmetry and $\pi(\cdot)$ may be assumed to be the identity permutation.

For the proof of Theorem~\ref{thm:DetDMBC}, we utilize binary polarization theorems. By contrast to polarization for point-to-point channels, in the case of deterministic DM-BCs, the polar transform is applied to the \emph{output} random variables of the channel.

\subsection{Polar Transform}

Consider an input distribution $P_X(x)$ to the deterministic DM-BC. Over $n$ channel uses, the input random variables to the channel are given by
\begin{align}
X^{1:n} = \{X^{1}, X^{2}, \ldots, X^{n}\}, \notag
\end{align}
where $X^{j} \sim P_X$ are independent and identically distributed ($i.i.d.$) random variables. The channel output variables are given by $Y_i(j) = f_i(X(j))$ where $f_i(\cdot)$ are the deterministic functions to each broadcast receiver. Denote the random matrix of channel output variables by
\vspace{0.1in}
\begin{align}
\matbold{Y} & \triangleq \left[\begin{array}{ccccc} Y_1^{1} & Y_1^{2} & Y_1^{3} & \ldots & Y_1^{n} \\ Y_2^{1} & Y_2^{2} & Y_2^{3} & \ldots & Y_2^{n} \\ \vdots & \vdots & \vdots & \ldots & \vdots \\ Y_m^{1} & Y_m^{2} & Y_m^{3} & \ldots & Y_m^{n} \end{array}\right], \label{eqn:YBAR}
\end{align}
where $\matbold{Y} \in \mathbb{F}_2^{n \times n}$. For $n = 2^{\ell}$ and $\ell \geq 1$, the polar transform is defined as the following invertible linear transformation,
\begin{align}
\matbold{U} & = \matbold{Y}\matbold{G}_n \label{eqn:PolarTransform} \\
\mbox{where}~\matbold{G}_n & \triangleq \left[\begin{array}{cc} 1 & 0 \\ 1 & 1 \end{array}\right]^{\bigotimes \log_2 n} \matbold{B}_n. \notag
\end{align}
The matrix $\matbold{G}_n \in \mathbb{F}_2^{n \times n}$ is formed by multiplying a matrix of successive Kronecker matrix-products (denoted by $\bigotimes$) with a bit-reversal matrix $\matbold{B}_n$ introduced by Ar\i kan~\cite{arikan10}. The polarized random matrix $\matbold{U} \in \mathbb{F}_2^{n \times n}$ is indexed as
\vspace{0.1in}
\begin{align}
\matbold{U} & \triangleq \left[\begin{array}{ccccc} U_1^{1} & U_1^{2} & U_1^{3} & \ldots & U_1^{n} \\ U_2^{1} & U_2^{2} & U_2^{3} & \ldots & U_2^{n} \\ \vdots & \vdots & \vdots & \ldots & \vdots \\ U_m^{1} & U_m^{2} & U_m^{3} & \ldots & U_m^{n} \end{array}\right]. \label{eqn:UBARBAR}
\end{align}

\subsection{Joint Distribution of Polarized Variables}

Consider the channel output distribution $P_{Y_1 Y_2\cdot\cdot\cdot Y_m}$ of the deterministic DM-BC induced by input distribution $P_X(x)$. The $j$-th \emph{column} of the random matrix $\matbold{Y}$ is distributed as $(Y_1^{j}, Y_2^{j}, \cdot\cdot\cdot, Y_m^{j}) \sim P_{Y_1 Y_2\cdot\cdot\cdot Y_m}$. Due to the memoryless property of the channel, the joint distribution of all output variables is
%\begin{align}
%P_{Y_{(1,[n])}Y_{(2,[n])}\cdot\cdot\cdot Y_{(m,[n])}}\left(y_{(1,[n])}, y_{(2,[n])}\cdot\cdot\cdot y_{(m,[n])}\right) = \notag \\ %\prod_{j=1}^{n} P_{Y_1 Y_2\cdot\cdot\cdot Y_m}\left(y_{(1,j)}, y_{(2,j)}, \ldots, y_{(m,j)}\right). \label{eqn:YJOINT}
%\end{align}
\begin{align}
& P_{Y_1^{n}Y_2^{n}\cdot\cdot\cdot Y_m^{n}}\Bigl(y_1^{n}, y_2^{n},\cdot\cdot\cdot, y_m^{n}\Bigl) = \notag \\
& ~~~~~~~~~~~~~~~ \prod_{j=1}^{n} P_{Y_1Y_2\cdot\cdot\cdot Y_m}\Bigl(y_1^{j}, y_2^{j}, \ldots, y_m^{j}\Bigl). \label{eqn:YJOINT}
\end{align}
The joint distribution of the matrix variables in $\matbold{Y}$ is characterized easily due to the $i.i.d.$ structure. The polarized random matrix $\matbold{U}$ does \emph{not} have an $i.i.d.$ structure. However, one way to define the joint distribution of the variables in $\matbold{U}$ is via the polar transform equation~\eqref{eqn:PolarTransform}. An alternate representation is via a decomposition into conditional distributions as follows\footnote{The abbreviated notation of the form $P(a|b)$ which appears in~\eqref{eqn:UJOINT} indicates $P_{A|B}(a|b)$, i.e. the conditional probability $\mathbb{P}\{A = a| B = b\}$ where $A$ and $B$ are random variables.}.
%\footnote{The notation $P\Bigl(u_i(j) \Bigl| u_i^{1:j-1}, \{u_k^{1:n}\}_{k \in [1:i-1]}\Bigl)$ is shorthand for the %conditional probability $\mathbb{P}\Bigl\{ U_i(j) = u_i(j) \Bigl| U_i^{j-1} = u_i^{1:j-1}, \{U_k^{1:n}\}_{k \in [1:i-1]} %= \{u_k^{1:n}\}_{k \in [1:i-1]}\Bigl\}$}.
%\begin{align}
%& P_{U_{(1,[n])}U_{(2,[n])}\cdot\cdot\cdot U_{(m,[n])}}\left(u_{(1,[n])}, u_{(2,[n])}\cdot\cdot\cdot u_{(m,[n])}\right) = \notag \\
%& \prod_{i=1}^{m}\prod_{j=1}^{n} P_{U_{i,j} \Bigl| U_{(i,[j-1])}, U_{([i-1],[n])}}\left(u_{i,j} \Bigl| u_{(i,[j-1])}, u_{([i-1],[n])}\right).
%%& \prod_{i=1}^{m}\prod_{j=1}^{n} \mathbb{P}(U_{(i,j)} = u_{(i,j)} | U_{(i,[j-1])} = u_{(i,[j-1])}, U_{([i-1],[n])} = u_{([i-1],[n])}).
%\label{eqn:UJOINT}
%\end{align}
\begin{align}
& P_{U_1^{n}U_2^{n}\cdot\cdot\cdot U_m^{n}}\Bigl(u_1^{n}, u_2^{n},\cdot\cdot\cdot u_m^{n}\Bigl) = \notag \\
& ~~~~~ \prod_{i=1}^{m}\prod_{j=1}^{n} P\Bigl(u_i(j) \Bigl| u_i^{1:j-1}, \{u_k^{1:n}\}_{k \in [1:i-1]}\Bigl).
%& \prod_{i=1}^{m}\prod_{j=1}^{n} \mathbb{P}(U_{(i,j)} = u_{(i,j)} | U_{(i,[j-1])} = u_{(i,[j-1])}, U_{([i-1],[n])} = u_{([i-1],[n])}).
\label{eqn:UJOINT}
\end{align}
As derived by Ar\i kan in~\cite{arikan10} and summarized in Section~\ref{sec:EstimatingBhatt}, the conditional probabilities in~\eqref{eqn:UJOINT} and associated likelihoods may be computed using a dynamic programming method which ``divides-and-conquers'' the computations efficiently.

%In addition, for efficient encoding and decoding purposes of polar codes through dynamic programming, a decomposition of %the joint distribution of the variables in $\matbold{U}$ into conditional distributions is useful. The conditional %probabilities in~\eqref{eqn:UJOINT} may be computed efficiently and recursively as shown by Ar\i kan in~\cite{arikan10}.
%Consider the channel output distribution $P_{\{Y_i\}_{i=1}^{m}}(\{y_i\}_{i=1}^{m})$ of the deterministic DM-BC induced by input distribution $P_X(x)$. The $j$-th \emph{column} of the matrix array $Y$ is distributed as $Y_{[m], j}} \sim P_{Y_1 Y_2\cdot\cdot\cdot Y_m}$. Due to the memoryless property of the channel, the joint distribution of all output variables is
%\begin{align}
%P_{Y_{(1,[n])}Y_{(2,[n])}\cdot\cdot\cdot Y_{(m, [n])}}(y_{(1,[n])}, y_{(2,[n])}, \ldots, y_{(m,[n])}) = \notag \\ \prod_{j=1}^{n} P_{Y_1 Y_2\cdot\cdot\cdot Y_m}(y_{(1,j)}, y_{(2,j)}, \ldots, y_{(m,j)}). \label{eqn:YJOINT}
%\end{align}

\subsection{Polarization of Conditional Entropies}
%Let $u$ denote the binary vector $(u_0, u_1, \ldots, u_{N-1})$ and
%$x$ denote the binary (codeword) vector $(x_0, x_1, \ldots,
%x_{N-1})$. For any subset $\mathcal{I} \subseteq \{0,\dots, (N-1)\}$
%let $u_{\mathcal{I}} = (u_{i_0},\dots, u_{i_{|\mathcal{I}|-1}})$
%where $i_k\in\mathcal{I}$ and $i_{k}\leq i_{k+1}$.
%based on a kernel $G_2 = \left[\begin{IEEEeqnarraybox*}[\mysmallarraydecl] [c]{,c/c,} 1&0\\
%1&1 \end{IEEEeqnarraybox*}\right]$.
%As shown in~\cite{arikan10}, the conditional entropies associated to the random variables $U([m],[n])$ polarize to either %$0$ or $1$ as $n \rightarrow \infty$. The following theorem by Ar\i kan makes this statement precise.
\input{./PolarTransformBitIndices_NEW}
\begin{proposition}[Polarization~\cite{arikan10}]\label{thm:SrcPolarization}
Consider the pair of random matrices $(\matbold{Y}, \matbold{U})$ related through the polar transformation in~\eqref{eqn:PolarTransform}. For $i \in [m]$ and any $\epsilon \in (0,1)$, define the set of indices
\begin{align}
& \mathcal{A}_i^{(n)} \triangleq \Bigl\{ j \in [n]: \notag \\
& ~~~~~~~~~~H\Bigl(U_i(j) \Bigl| U_i^{1:j-1}, \{Y_k^{1:n}\}_{k \in [1:i-1]} \Bigl) \geq 1 - \epsilon \Bigl\}. \label{eqn:HighEntropyIndices}
\end{align}
Then in the limit as $n \rightarrow \infty$,
\begin{align}
\frac{1}{n}\Bigl| \mathcal{A}_i^{(n)} \Bigl| \rightarrow H(Y_i | Y_1Y_2\cdot\cdot\cdot Y_{i-1}). \label{eqn:SrcPolarization}
\end{align}
%\begin{align}
%\frac{1}{n}\Bigl| \bigl\{ j \in [n]: H(U_{1j}|U_{1,1}^{j-1}) & \geq 1 -
%\epsilon \bigl\} \Bigl| \rightarrow H(Y_1), \notag \\
%\Bigl| \bigl\{ j \in [n]: H(U_{2j}|U_{2,1}^{j-1}, \bar{Y}_1) & \geq 1 -
%\epsilon \bigl\} \Bigl| \rightarrow nH(Y_2|Y_1). \notag
%\end{align}
\end{proposition}

For sufficiently large $n$, Theorem~\ref{thm:SrcPolarization} establishes that there exist approximately $nH\left(Y_i | Y_1Y_2\cdot\cdot\cdot Y_{i-1} \right)$ indices per row $i \in [m]$ of random matrix $\matbold{U}$ for which the conditional entropy is close to $1$. The total number of indices in $\matbold{U}$ for which the conditional entropy terms polarize to $1$ is approximately $nH(Y_1Y_2\cdot\cdot\cdot Y_m)$. The polarization phenomenon is illustrated in Figure~\ref{fig:PolarTransformBitIndices}.

\begin{remark} Since the polar transform $\matbold{G}_n$ is invertible, $\{U_k^{1:n}\}_{k \in [1:i-1]}$ are in one-to-one correspondence with $\{Y_k^{1:n}\}_{k \in [1:i-1]}$. Therefore the conditional entropies $H\bigl(U_i(j) \bigl| U_i^{1:j-1}, \{U_k^{1:n}\}_{k \in [1:i-1]} \bigl)$ also polarize to $0$ or $1$.
\end{remark}
\vspace{0.1in}

%To code over the DBC, we would like to place uniformly random message bits into the indices %corresponding to these components.
%Theorem~\ref{thm:SrcPolarization} is a  each row $i \in [m]$ of $Y$ is processed successively.
\subsection{Rate of Polarization}

The Bhattacharyya parameter of random variables is closely related to the conditional entropy. The parameter is useful for characterizing the rate of polarization.
\begin{definition}[Bhattacharyya Parameter]\label{def:Bhattacharyya} Let $(T, V) \sim P_{T, V}$ where $T \in \{0,1\}$ and $V \in \mathcal{V}$ where $\mathcal{V}$ is an arbitrary discrete alphabet. The Bhattacharyya parameter $Z(T|V) \in [0,1]$ is defined
\begin{align}
Z(T|V) = 2 \sum_{v \in \mathcal{V}} P_{V}(v) \sqrt{P_{T|V}(0|v)P_{T|V}(1|v)}.
\end{align}
\end{definition}
%As shown in~\cite[Proposition 2]{arikan10}, $Z(T|V)^{2} \leq H(T|V)$ and $H(T|V) \leq \log(1 %+ Z(T|V))$. As one consequence of this result, if $Z(T|V) \geq 1 - \delta$, then $H(T|V) \geq %1 - 2\delta$ for $\delta > 0$. Using the Bhattacharyya parameter of random variables, %consider the following indices for $i \in [m]$,
As shown in Lemma~\ref{lemma:ClosenessOfHandZOne} of Appendix~\ref{sec:PolarCodingLemmas}, $Z(T|V) \rightarrow 1$ implies $H(T|V) \rightarrow 1$, and similarly $Z(T|V) \rightarrow 0$ implies $H(T|V) \rightarrow 0$ for $T$ a binary random variable. Based on the Bhattacharyya parameter, the following theorem specifies sets $\mathcal{M}_i^{(n)} \subset [n]$ that will be called \emph{message} sets.
\vspace{0.05in}
\begin{proposition}[Rate of Polarization]\label{thm:RateOfPolarization} Consider the pair of random matrices $(\matbold{Y}, \matbold{U})$ related through the polar transformation in~\eqref{eqn:PolarTransform}. Fix constants $0 < \beta < \frac{1}{2}$, $\tau > 0$, $i \in [m]$. Let $\delta_n = 2^{-n^{\beta}}$ be the rate of polarization. Define the set
\begin{align}
& \mathcal{M}^{(n)}_{i} \triangleq \Bigl\{j \in [n]: \notag \\
& ~~~~~~~~Z\Bigl(U_i(j) \Bigl| U_i^{1:j-1}, \{Y_k^{1:n}\}_{k \in [1:i-1]}\Bigl) \geq 1 - \delta_n \Bigl\}. \label{eqn:MessageIndices}
\end{align}
Then there exists an $N_o = N_o(\beta, \tau)$ such that
\begin{align}
\frac{1}{n}\Bigl|\mathcal{M}^{(n)}_{i}\Bigl| \geq H(Y_i | Y_1Y_2\cdot\cdot\cdot Y_{i-1}) - \tau, \label{eqn:MessageSetAchievesCapacity}
\end{align}
for all $n > N_o$.
\end{proposition}
The proposition is established via the Martingale Convergence Theorem by defining a super-martingale with respect to the Bhattacharyya parameters~\cite{arikan09}~\cite{arikan10}. The rate of polarization is characterized by Ar\i kan and Telatar in~\cite{arikantelatar09}.
%As derived in~\cite{arikantelatar09}, the rate of polarization may be selected as $\delta_n = %\frac{1}{n}2^{-n^{\beta}}$ for $0 < \beta < \frac{1}{2}$. As a result of %Theorem~\ref{thm:SrcPolarization} and the rate of polarization of the  Bhattacharyya source %parameter, for a fixed $\tau > 0$, there exists an $n$ large enough so that
%\begin{align}
%\frac{1}{n}\left|\mathcal{M}^{(n)}_{i}\right| \geq H\left(Y_i | Y_1^{i-1}\right) - \tau. %\notag
%\end{align}
\begin{remark} The message sets $\mathcal{M}_{i}^{(n)}$ are computed ``offline'' only once during a code construction phase. The sets do not depend on the realization of random variables. In the following Section~\ref{sec:EstimatingBhatt}, a Monte Carlo sampling approach for estimating Bhattacharyya parameters is reviewed. Other highly efficient algorithms are known in the literature for finding the message indices (see e.g. Tal and Vardy~\cite{talvardy2011}).
\end{remark}

\subsection{Estimating Bhattacharyya Parameters}\label{sec:EstimatingBhatt}

As shown in Lemma~\ref{lemma:BhattacharyyaMonteCarlo} in Appendix~\ref{sec:PolarCodingLemmas}, one way to estimate the Bhattacharyya parameter $Z(T|V)$ is to sample from the distribution $P_{T,V}(t,v)$ and evaluate $\mathbb{E}_{T,V} \sqrt{\varphi(T,V)}$. The function $\varphi(t,v)$ is defined based on likelihood ratios
\begin{align}
L(v) & \triangleq \frac{P_{T|V}(0|v)}{P_{T|V}(1|v)},  \notag \\
L^{-1}(v) & \triangleq \frac{P_{T|V}(1|v)}{P_{T|V}(0|v)}. \notag
\end{align}

Similarly, to determine the indices in the message sets $\mathcal{M}^{(n)}_{i}$ defined in Proposition~\ref{thm:RateOfPolarization}, the Bhattacharyya parameters $Z\left(U_i(j) \Bigl| U_{i}^{1:j-1}, \{Y_{k}^{1:n}\}_{k \in [i-1]}\right)$ must be estimated efficiently. For $n \geq 2$, define the likelihood ratio
\begin{align}
& L_n^{(i,j)}\left(u_i^{1:j-1}, \{y_k^{1:n}\}_{k \in [1:i-1]}\right) \triangleq \notag \\
& ~~~ \frac{\mathbb{P}\left(U_i(j) = 0 \Bigl| U_i^{1:j-1} = u_i^{1:j-1}, \{Y_k^{1:n} = y_k^{1:n}\}_{k \in [1:i-1]} \right)}{\mathbb{P}\left(U_i(j) = 1 \Bigl| U_i^{1:j-1} = u_i^{1:j-1}, \{Y_k^{1:n} = y_k^{1:n}\}_{k \in [1:i-1]} \right) }. \label{eqn:LikelihoodRatio}
\end{align}
The dynamic programming method given in~\cite{arikan10} allows for a recursive computation of the likelihood ratio. Define the following sub-problems
\begin{align}
\Xi_1 & = L_{\frac{n}{2}}^{(i,j)}\left(u_{i,o}^{1:2j-2} \oplus u_{i,e}^{1:2j-2}, \{y_k^{1:\frac{n}{2}}\}_{k \in [1:i-1]}\right), \notag \\
\Xi_2 & = L_{\frac{n}{2}}^{(i,j)}\left(u_{i,e}^{1:2j-2}, \{y_k^{\frac{n}{2}+1:n}\}_{k \in [1:i-1]}\right), \notag
\end{align}
where the notation $u_{i,o}^{1:2j-2}$ and $u_{i,e}^{1:2j-2}$ represents the odd and even indices respectively of the sequence $u_{i}^{1:2j-2}$. The recursive computation of the likelihoods is characterized by
\begin{align}
L_n^{(i,2j-1)}\left(u_i^{1:2j-2}, \{y_k^{1:n}\}_{k \in [1:i-1]}\right) & = \frac{ \Xi_1 \Xi_2 + 1 } {\Xi_1 + \Xi_2}. \notag \\
L_n^{(i,2j)}\left(u_i^{1:2j-1}, \{y_k^{1:n}\}_{k \in [1:i-1]}\right) & = \left(\Xi_1\right)^{\gamma} \Xi_2, \notag
\end{align}
where $\gamma = 1$ if $u_i(2j-1) = 0$ and $\gamma = -1$ if $u_i(2j-1) = 1$. In the above recursive computations, the base case is for sequences of length $n = 2$.

\section{Proof Of Theorem~\ref{thm:DetDMBC}}\label{sec:ProofOfMainTheorem}

The proof of Theorem~\ref{thm:DetDMBC} is based on binary polarization theorems as discussed in Section~\ref{sec:PolarizationTheorems}. The random coding arguments of C. E. Shannon prove the existence of capacity-achieving codes for point-to-point channels. Furthermore, random binning and joint-typicality arguments suffice to prove the existence of capacity-achieving codes for the deterministic DM-BC. However, it is shown in this section that there exist capacity-achieving \emph{polar codes} for the binary-output deterministic DM-BC.

\subsection{Broadcast Code Based on Polarization}
\label{sec:EncodingDecoding}

The ordering of the receivers' rates in $\vec{R}$ is arbitrary due to symmetry. Therefore, let $\pi(i) = i$ be the identity permutation which denotes the successive order in which the message bits are allocated for each receiver. The encoder must map $m$ independent messages $(W_1, W_2, \ldots, W_m)$ uniformly distributed over $[2^{nR_1}] \times [2^{nR_2}] \times \cdot\cdot\cdot \times [2^{nR_m}]$ to a codeword $x^{n} \in \mathcal{X}^{n}$. To construct a codeword for broadcasting $m$ independent messages, the following binary sequences are formed at the encoder: $u_1^{1:n}, u_2^{1:n}, \ldots, u_m^{1:n}$. To determine a particular bit $u_i(j)$ in the binary sequence $u_i^{1:n}$, if $j \in \mathcal{M}_{i}^{(n)}$, the bit is selected as a uniformly distributed message bit intended for receiver $i \in [m]$. As defined in~\eqref{eqn:MessageIndices} of Proposition~\ref{thm:RateOfPolarization}, the message set $\mathcal{M}_{i}^{(n)}$ represents those indices for bits transmitted to receiver $i$. The remaining \emph{non-message} indices in the binary sequence $u_i^{1:n}$ for each user $i \in [m]$ are computed either according to a deterministic or random mapping.

\vspace{0.1in}
\subsubsection{Deterministic Mapping} Consider a class of deterministic boolean functions indexed by $i \in [m]$ and $j \in [n]$:
\begin{align}
\psi^{(i,j)}: \{0,1\}^{n(\max\{0,i-1\}) + j-1} \rightarrow \{0,1\}. \label{eqn:BooleanMapsDetDMBC}
\end{align}
As an example, consider the deterministic boolean function based on the \emph{maximum a posteriori} polar coding rule.
\begin{align}
& \psi^{(i,j)}_{MAP}\left(u_i^{1:j-1}, \{y_k^{1:n}\}_{k \in [1:i-1]}\right) \triangleq \argmax_{u \in \{0,1\}} \Bigl\{\notag \\
& ~\mathbb{P}\left(U_i(j) = u \Bigl| U_i^{1:j-1} = u_i^{1:j-1}, \{Y_k^{1:n} = y_k^{1:n}\}_{k \in [1:i-1]} \right)\Bigl\}. \label{eqn:MAPRuleForDetBC}
\end{align}

\subsubsection{Random Mapping} Consider a class of random boolean functions indexed by $i \in [m]$ and $j \in [n]$:
\begin{align}
\Psi^{(i,j)}: \{0,1\}^{n(\max\{0,i-1\}) + j-1} \rightarrow \{0,1\}. \label{eqn:RandomMapsDetDMBC}
\end{align}
As an example, consider the random boolean function
\begin{align}
& \Psi^{(i,j)}_{RAND}\left(u_i^{1:j-1}, \{y_k^{1:n}\}_{k \in [1:i-1]}\right) \triangleq \notag \\
& ~~~~~~~~\begin{cases} 0, & \mbox{w.p. }~ \lambda_{0}\left(u_i^{1:j-1}, \{y_k^{1:n}\}_{k \in [1:i-1]}\right), \\ 1, & \mbox{w.p. }~1-\lambda_{0}\left(u_i^{1:j-1}, \{y_k^{1:n}\}_{k \in [1:i-1]}\right), \end{cases} \label{eqn:RandomMapForDetBC}
\end{align}
where
\begin{align}
& \lambda_{0}\left(u_i^{1:j-1}, \{y_k^{1:n}\}_{k \in [1:i-1]}\right) \triangleq \notag \\
& ~~\mathbb{P}\left(U_i(j) = 0 \Bigl| U_i^{1:j-1} = u_i^{1:j-1}, \{Y_k^{1:n} = y_k^{1:n}\}_{k \in [1:i-1]} \right). \notag
\end{align}
The random boolean function $\Psi^{(i,j)}_{RAND}$ may be thought of as a vector of Bernoulli random variables indexed by the input to the function. Each Bernoulli random variable of the vector has a fixed probability of being one or zero that is well-defined.
\vspace{0.1in}
\subsubsection{Mapping From Messages To Codeword}

The binary sequences $u_i^{1:n}$ for $i \in [m]$ are formed \emph{successively} bit by bit. If $j \in \mathcal{M}_{i}^{(n)}$, then the bit $u_i(j)$ is one message bit from the uniformly distributed message $W_i$ intended for user $i$. If $j \notin \mathcal{M}_{i}^{(n)}$, $u_i(j) = \psi^{(i,j)}_{MAP}\bigl(u_i^{1:j-1}, \{y_k^{1:n}\}_{k \in [1:i-1]}\bigl)$ in the case of a deterministic mapping, or $u_i(j) = \Psi^{(i,j)}_{RAND}\bigl(u_i^{1:j-1}, \{y_k^{1:n}\}_{k \in [1:i-1]}\bigl)$ in the case of a random mapping. The encoder then applies the \emph{inverse} polar transform for each sequence: $y_i^{1:n} = u_i^{1:n}\matbold{G}_n^{-1}$. The codeword $x^{n}$ is formed symbol-by-symbol as follows:
\begin{align}
x(j) \in \bigcap_{i=1}^{m} f_i^{-1}\left(  y_i(j)  \right). \notag
\end{align}
If the intersection set is empty, the encoder declares a block error. A block error only occurs at the encoder.
\vspace{0.1in}
\subsubsection{Decoding at Receivers}
If the encoder succeeds in transmitting a codeword $x^{n}$, each receiver obtains the sequence $y_i^{1:n}$ noiselessly and applies the polar transform $\matbold{G}_n$ to recover $u_i^{1:n}$ exactly. Since the message indices $\mathcal{M}_{i}^{(n)}$ are known to each receiver, the message bits in $u_i^{1:n}$ are decoded correctly by receiver $i$.

%\subsection{Bounding The Total Variation Distance Between Probability Measures}
\subsection{Total Variation Bound}

While the deterministic mapping $\psi^{(i,j)}_{MAP}$ performs well in practice, the average probability of error $P_e^{(n)}$ of the coding scheme is more difficult to analyze in theory. The random mapping $\Psi^{(i,j)}_{RAND}$ at the encoder is more amenable to analysis via the probabilistic method. Towards that goal, consider the following probability measure defined on the space of tuples of binary sequences\footnote{A related proof technique was provided for lossy source coding based on polarization in a different context~\cite{korada10}. In the present paper, a different proof is supplied that utilizes the chain rule for KL-divergence.}.
\begin{align}
& Q\Bigl(u_1^{n}, u_2^{n},\cdot\cdot\cdot, u_m^{n}\Bigl) \triangleq \notag \\
& ~~~~~~~~~ \prod_{i=1}^{m}\prod_{j=1}^{n} Q\Bigl(u_i(j) \Bigl| u_i^{1:j-1}, \{u_k^{1:n}\}_{k \in [1:i-1]}\Bigl).
%& \prod_{i=1}^{m}\prod_{j=1}^{n} \mathbb{P}(U_{(i,j)} = u_{(i,j)} | U_{(i,[j-1])} = u_{(i,[j-1])}, U_{([i-1],[n])} = u_{([i-1],[n])}).
\label{eqn:QUJOINT}
\end{align}
where the conditional probability measure
\begin{align}
& Q\left(u_i(j) \Bigl| u_i^{1:j-1}, \{u_k^{1:n}\}_{k \in [1:i-1]} \right) \triangleq \notag \\
& ~~~~~~ \begin{cases} \frac{1}{2}, & \mbox{\emph{if}}~j \in \mathcal{M}_i^{(n)}, \\ P\left( u_i(j) \Bigl| u_i^{1:j-1}, \{u_k^{1:n}\}_{k \in [1:i-1]}  \right), & \mbox{\emph{otherwise.}} \end{cases} \notag
\end{align}
The probability measure $Q$ defined in~\eqref{eqn:QUJOINT} is a perturbation of the joint probability measure $P$ defined in~\eqref{eqn:UJOINT} for the random variables $U_i(j)$. The only difference in definition between $P$ and $Q$ is due to those indices in message set $\mathcal{M}_i^{(n)}$. The following lemma provides a bound on the total variation distance between $P$ and $Q$.
\vspace{0.1in}
\begin{lemma}(\emph{Total Variation Bound})\label{lemma:TVBound}
Let probability measures $P$ and $Q$ be defined as in~\eqref{eqn:UJOINT} and~\eqref{eqn:QUJOINT} respectively. Let $0 < \beta < 1$. For sufficiently large $n$, the total variation distance between $P$ and $Q$ is bounded as
\begin{align}
\sum_{\{u_{k}^{1:n}\}_{k \in [m]}} \Bigl| P\bigl( \{u_{k}^{1:n}\}_{k \in [m]}\bigl) - Q\bigl(\{u_{k}^{1:n}\}_{k \in [m]}\bigl) \Bigl| \leq 2^{-n^{\beta}}. \notag
\end{align}
\end{lemma}
\begin{proof} See Section~\ref{sec:AppendixTVBound} of the Appendices. \end{proof}
\vspace{0.05in}
%\begin{remark}
%Lemma~\ref{lemma:TVBound} shows that the total variation distance is bounded as $\mathcal{O}(2^{-n^{\beta}})$ for a constant number of %broadcast receivers $m$. The distribution $Q$ approximates the distribution $P$ as $n \rightarrow \infty$, and the randomness is %extracted into uniformly distributed bits over message indices $\mathcal{M}_{i}^{(n)}$.
%\end{remark}

\subsection{Analysis of the Average Probability of Error}

For the $m$-user deterministic DM-BC, an error event occurs at the encoder if a codeword $x^{n}$ is unable to be constructed symbol by symbol according to the broadcast protocol described in Section~\ref{sec:EncodingDecoding}. Define the following set consisting of $m$-tuples of binary sequences,
\begin{align}
\mathcal{T} & \triangleq \Bigl\{ (y_1^{n}, y_2^{n}, \ldots, y_m^{n}): \exists j \in [n], \bigcap_{i=1}^{m} f_i^{-1}\left(  y_i(j) \right) = \emptyset \Bigl\}.
%\mathcal{\tilde{T}} & \triangleq \Bigl\{ (u_1^{n}, u_2^{n}, \ldots, u_m^{n}): u_1^{n} = y_1^{n}\matbold{G}_n, u_2^{n} = y_2^{n}\matbold{G}_n, \ldots, %u_m^{n} = y_m^{n}\matbold{G}_n~\mbox{and} \notag \\
%& ~~~~~ (y_1^{n}, y_2^{n}, \ldots, y_m^{n}) \in \mathcal{T} \Bigl\}.
\end{align}
The set $\mathcal{T}$ consists of those $m$-tuples of binary output sequences which are \emph{inconsistent} due to the properties of the deterministic channel. In addition, due to the one-to-one correspondence between sequences $u_i^{1:n}$ and $y_i^{1:n}$, denote by $\mathcal{\tilde{T}}$ the set of $m$-tuples $(u_1^{n}, u_2^{n}, \ldots, u_m^{n})$ that are inconsistent.

For the broadcast protocol, the rate $R_i = \frac{1}{n}\bigl|\mathcal{M}_i^{(n)}\bigl|$ for each receiver. Let the total sum rate for all broadcast receivers be $R_{\Sigma} = \sum_{i \in [m]} R_i$. If the encoder uses a fixed deterministic map $\psi^{(i,j)}$ in the broadcast protocol, the average probability of error is
\begin{align}
& P_e^{(n)}\left[\{\psi^{(i,j)}\}\right] = \frac{1}{2^{nR_{\Sigma}}} \sum_{\{u_{k}^{1:n}\}_{k \in [m]}} \Biggl[  \indicator{(u_1^{n}, u_2^{n}, \ldots, u_m^{n}) \in \mathcal{\tilde{T}} } \notag \\
& ~~ \cdot \prod_{\begin{subarray}{c} i \in [m] \\ j \in [n]: j \notin \mathcal{M}_i^{(n)} \end{subarray}} \indicator{\psi^{(i,j)}\left(u_i^{1:j-1}, \{y_k^{1:n}\}_{k \in [1:i-1]}\right) = u_i(j)} \Biggl]. \label{eqn:ProbOfErrorDetMAP}
\end{align}
In addition, if the random maps $\Psi^{(i,j)}$ are used at the encoder, the average probability of error is a random quantity given by
\begin{align}
& P_e^{(n)}\left[\{\Psi^{(i,j)}\}\right] = \frac{1}{2^{nR_{\Sigma}}} \sum_{\{u_{k}^{1:n}\}_{k \in [m]}} \Biggl[  \indicator{(u_1^{n}, u_2^{n}, \ldots, u_m^{n}) \in \mathcal{\tilde{T}} } \notag \\
& ~~ \cdot \prod_{\begin{subarray}{c} i \in [m] \\ j \in [n]: j \notin \mathcal{M}_i^{(n)} \end{subarray}} \indicator{\Psi^{(i,j)}\left(u_i^{1:j-1}, \{y_k^{1:n}\}_{k \in [1:i-1]}\right) = u_i(j)} \Biggl]. \label{eqn:ProbOfErrorRandomMAP}
\end{align}
Instead of characterizing $P_e^{(n)}$ directly for deterministic maps, the analysis of $P_e^{(n)}[\{\Psi^{(i,j)}\}]$ leads to the following lemma.

\begin{lemma}\label{theorem:ErrorProb} Consider the broadcast protocol of Section~\ref{sec:EncodingDecoding}. Let $R_i = \frac{1}{n}\bigl|\mathcal{M}_i^{(n)}\bigl|$ for $i \in [m]$ be the broadcast rates selected according to the criterion given in~\eqref{eqn:MessageIndices} in Proposition~\ref{thm:RateOfPolarization}. Then for $0 < \beta < 1$ and sufficiently large $n$,
\begin{align}
\mathbb{E}_{\{\Psi^{(i,j)}\}} \Bigl[ P_e^{(n)}[\{\Psi^{(i,j)}\}] \Bigl] < 2^{-n^{\beta}}. \notag
\end{align}
\end{lemma}
\begin{proof}
\begin{align}
& \mathbb{E}_{\{\Psi^{(i,j)}\}} \Bigl[ P_e^{(n)}[\{\Psi^{(i,j)}\}] \Bigl] \notag \\
& = \frac{1}{2^{nR_{\Sigma}}} \sum_{\{u_{k}^{1:n}\}_{k \in [m]}} \Biggl[  \indicator{(u_1^{n}, u_2^{n}, \ldots, u_m^{n}) \in \mathcal{\tilde{T}} } \cdot \notag \\
& \prod_{\begin{subarray}{c} i \in [m] \\ j \in [n]: j \notin \mathcal{M}_i^{(n)} \end{subarray}} \mathbb{P}\left\{ \Psi^{(i,j)}\left(u_i^{1:j-1}, \{y_k^{1:n}\}_{k \in [1:i-1]}\right) = u_i(j) \right\} \Biggl] \notag \\
& = \sum_{\{u_{k}^{1:n}\}_{k \in [m]} \in \tilde{\mathcal{T}}} Q\bigl(\{u_{k}^{1:n}\}_{k \in [m]}\bigl) \label{eqn:TheUseOfQ} \\
& = \sum_{\{u_{k}^{1:n}\}_{k \in [m]} \in \tilde{\mathcal{T}}} \Bigl| P\bigl( \{u_{k}^{1:n}\}_{k \in [m]}\bigl) - Q\bigl(\{u_{k}^{1:n}\}_{k \in [m]}\bigl) \Bigl| \label{eqn:PHasZeroMassOverBadSeq} \\
& \leq 2^{-n^{\beta}}. \label{eqn:UseOfLemma}
\end{align}
Step~\eqref{eqn:TheUseOfQ} follows since the probability measure $Q$ matches the desired calculation exactly. Step~\eqref{eqn:PHasZeroMassOverBadSeq} is due to the fact that the probability measure $P$ has \emph{zero mass} over $m$-tuples of binary sequences that are inconsistent. Step~\eqref{eqn:UseOfLemma} follows directly from Lemma~\ref{lemma:TVBound}. Lastly, since the expectation over random maps $\{\Psi^{(i,j)}\}$ of the average probability of error decays stretched-exponentially, there must exist a set of deterministic maps which exhibit the same behavior.
\end{proof}

\section{Noisy Broadcast Channels \\ Superposition Coding}

Coding for noisy broadcast channels is now considered using polarization methods. By contrast to the deterministic case, a decoding error event occurs at the receivers on account of the randomness due to noise. For the remaining sections, it is assumed that there exist $m = 2$ users in the DM-BC. The private-message capacity region for the DM-BC is unknown even for binary input, binary output two-user channels such as the skew-symmetric DM-BC. However, the private-message capacity region is known for specific classes.

\subsection{Special Classes of Noisy DM-BCs}\label{sec:SpecialClassesDMBCs}

\begin{definition}\label{def:PhysicallDegradedBC} The two-user \emph{physically degraded} DM-BC is a channel $P_{Y_1 Y_2|X}(y_1, y_2 | x)$ for which $X-Y_1-Y_2$ form a Markov chain, i.e. one of the receivers is statistically stronger than the other:
\begin{align}
P_{Y_1 Y_2 | X}(y_1, y_2 | x) & = P_{Y_1|X}(y_1| x) P_{Y_2|Y_1}(y_2|y_1). \label{eqn:PhysicallyDegraded}
\end{align}
\end{definition}

\begin{definition}\label{def:StochasticDegradedBC} A two-user DM-BC $P_{Y_1 Y_2|X}(y_1, y_2 | x)$ is \emph{stochastically degraded} if its conditional marginal distributions are the same as that of a physically degraded DM-BC, i.e., if there exists a distribution $\tilde{P}_{Y_2|Y_1}(y_2|y_1)$ such that
\begin{align}
P_{Y_2|X}(y_2|x) & = \sum_{y_1 \in \mathcal{Y}_1} P_{Y_1|X}(y_1|x) \tilde{P}_{Y_2|Y_1}(y_2|y_1). \label{eqn:StochasticDegraded}
\end{align}
If~\eqref{eqn:StochasticDegraded} holds for two conditional distributions $P_{Y_1|X}(y_1|x)$ and $P_{Y_2|X}(y_2|x)$ defined over the same input, then the property is denoted as follows: $P_{Y_1|X}(y_1|x) \degBC P_{Y_2|X}(y_2|x)$.
\end{definition}

\begin{definition}\label{def:LessNoisyBC} A two-user DM-BC $P_{Y_1 Y_2|X}(y_1, y_2 | x)$ for which $V-X-(Y_1,Y_2)$ forms a Markov chain is said to be \emph{less noisy} if
\begin{align}
\forall P_{VX}(v,x): I(V; Y_1) \geq I(V; Y_2). \label{eqn:LessNoisy}
\end{align}
\end{definition}

\begin{definition}\label{def:MoreCapableBC} A two-user DM-BC $P_{Y_1 Y_2|X}(y_1, y_2 | x)$ is said to be \emph{more capable} if
\begin{align}
\forall P_X(x): I(X; Y_1) \geq I(X; Y_2). \label{eqn:MoreCapable}
\end{align}
%If~\eqref{eqn:MoreCapable} holds for random variables $Y_1$, $Y_2$, $X$ with joint distribution $P_{Y_1 Y_2 X}(y_1, y_2, x) = P_{Y_1 %Y_2|X}(y_1, y_2 |x)P_X(x)$, then the property is denoted as follows: $P_{Y_1|X}(y_1|x) \overset{\textsc{mc}}{\succ} %P_{Y_2|X}(y_2|x)$.
\end{definition}

The following lemma relates the properties of the special classes of noisy broadcast channels. A more comprehensive treatment of special classes is given by C. Nair in~\cite{nair10}.
\begin{lemma}\label{lemma:SpecialClassesDMBCsProperties} Consider a two-user DM-BC $P_{Y_1 Y_2|X}(y_1, y_2 | x)$. Let $V-X-(Y_1,Y_2)$ form a Markov chain, $|\mathcal{V}| > 1$, and $P_{V}(v) > 0$. The following implications hold:
\begin{align}
& X-Y_1-Y_2 \notag \\
& ~~~~~~\Rightarrow P_{Y_1|X}(y_1|x) \degBC P_{Y_2|X}(y_2|x) \label{eqn:LemmaFirstI} \\
& ~~~~~~\Leftrightarrow \forall P_{X|V}(x|v): P_{Y_1|V}(y_1|v) \degBC P_{Y_2|V}(y_2|v) \label{eqn:LemmaSecondI} \\
& ~~~~~~\Rightarrow \forall P_{VX}(v,x): I(V; Y_1) \geq I(V; Y_2) \label{eqn:LemmaThirdI} \\
& ~~~~~~\Rightarrow \forall P_X(x): I(X; Y_1) \geq I(X; Y_2). \label{eqn:LemmaFourthI}
\end{align}
The converse statements for~\eqref{eqn:LemmaFirstI},~\eqref{eqn:LemmaThirdI}, and~\eqref{eqn:LemmaFourthI} do \emph{not} hold in general. Figure~\ref{fig:BroadcastChannelClasses} illustrates the different types of broadcast channels as a hierarchy.
\end{lemma}
\begin{proof} See Section~\ref{sec:AppendixLemmaSpecialClassesDMBCsProperties} of the Appendices.
\end{proof}
\input{./BroadcastChannelClasses}
%\begin{example}[Special Classes of Broadcast Channels]\label{ex:SpecialClassesDMBCs}  based on the implications established in %Lemma~\ref{lemma:SpecialClassesDMBCsProperties}. The private-message capacity region is achieved by superposition coding for classes %$I, II, III, IV$.
%\end{example}
%\begin{remark} It is known that for two-user DM-BCs, degradedness implies ``more-capability'', i.e. $P_{Y_1|X}(y_1|x) %\overset{\textsc{deg}}{\succ} P_{Y_2|X}(y_2|x) \Rightarrow P_{Y_1|X}(y_1|x) \overset{\textsc{mc}}{\succ} P_{Y_2|X}(y_2|x)$, although %the converse is not true. For the above classes of channels (and also for classes such as the \emph{less-noisy} DM-BCs), the %private-message capacity region is known.
%\end{remark}

%\subsection{Achievable Rate Regions}
\subsection{Cover's Inner Bound} Superposition coding involves one auxiliary random variable $V$ which conveys a ``cloud center'' or a coarse message decoded by both receivers~\cite{cover72}. One of the receivers then decodes an additional ``satellite codeword'' conveyed through $X$ containing a fine-grain message that is superimposed upon the coarse information.
\vspace{0.05in}
\begin{proposition}[Cover's Inner Bound] For any two-user DM-BC, the rates $(R_1, R_2) \in \mathbb{R}_{+}^{2}$ in the region $\mathfrak{R}(X, V, Y_1, Y_2)$ are achievable where
\begin{align}
\mathfrak{R}(X, V, Y_1, Y_2) \triangleq \Bigl\{R_1, R_2~\Bigl|~ R_1 & \leq I(X; Y_1|V), \notag \\
R_2 & \leq I(V; Y_2), \notag \\
R_1 + R_2 & \leq I(X; Y_1) \Bigl\}. \label{eqn:DegradedBCRateRegion}
\end{align}
and where random variables $X, V, Y_1, Y_2$ obey the Markov chain $V-X-(Y_1,Y_2)$.
%The region $\mathfrak{R}(X, V, Y_1, Y_2)$ is defined for all random variables $X, V, Y_1, Y_2$ with joint %distribution $P_{V}(v)P_{X|V}(x|v)P_{Y_1 Y_2|X}(y_1, y_2 | x)$.
\end{proposition}

\begin{remark} Cover's inner bound is applicable for \emph{any} broadcast channel. By symmetry, the following rate region is also achievable: $\{R_1, R_2~|~ R_2 \leq I(X; Y_2|V), R_1 \leq I(V; Y_1), R_1 + R_2 \leq I(X; Y_2)\}$ for random variables obeying the Markov chain $V-X-(Y_1, Y_2)$.
\end{remark}

\begin{remark} The inner bound is the capacity region for degraded, less-noisy, and more-capable DM-BCs (i.e. Class $I$ through Class $IV$ as shown in Figure~\ref{fig:BroadcastChannelClasses}). For the degraded and less-noisy special classes, the capacity region is simplified to $\{R_1, R_2~|~ R_1 \leq I(X; Y_1|V), R_2 \leq I(V; Y_2) \}$. To see this, note that $I(V;Y_2) \leq I(V;Y_1)$ which implies $I(V;Y_2) + I(X;Y_1|V) \leq I(V;Y_1) + I(X;Y_1|V) = I(X;Y_1)$. Therefore the sum-rate constraint $R_1 + R_2 \leq I(X;Y_1)$ of the rate-region in~\eqref{eqn:DegradedBCRateRegion} is automatically satisfied.
\end{remark}
%\begin{proposition}[Superposition Inner Bound] The single-letter capacity region of the two-user degraded DM-BC is defined by rate-tuples $(R_1, R_2) \in \mathbb{R}_{+}^{2}$ in the region
%\begin{align}
%\mathfrak{C}_{DEG\!-\!BC} & \triangleq \operatorname{co}\Bigl(\operatorname{clo}\Bigl( \bigcup_{X, U, Y_1, Y_2} \mathfrak{R}\bigl(X, U, Y_1, Y_2\bigl) \Bigl) \Bigl), \label{eqn:CapacityExpressionOfDegradedBC}
%\end{align}
%where the region $\mathfrak{R}(X, U, Y_1, Y_2)$ is defined as
%\begin{align}
%\mathfrak{R} \triangleq \Bigl\{R_1, R_2~\Bigl|~ & R_2 \leq I(U; Y_2), \notag \\
%& R_1 \leq I(X; Y_1|U) \Bigl\}. \label{eqn:DegradedBCRateRegion}
%\end{align}
%The union in~\eqref{eqn:CapacityExpressionOfDegradedBC} is over all random variables $X, U, Y_1, Y_2$ with joint distribution $P_{U}(u)P_{X|U}(x|u)P_{Y_1 Y_2|X}(y_1, y_2 | x)$ and cardinality bounds on the auxiliary random variable given by $|\mathcal{U}| \leq \min\{|\mathcal{X}|, |\mathcal{Y}_1|, |\mathcal{Y}_2|\}$.
%\end{proposition}

\vspace{0.15in}
\begin{example}[Binary Symmetric DM-BC]\label{ex:BinarySymmetricDMBC} The two-user binary symmetric DM-BC consists of a binary symmetric channel with flip probability $p_1$ denoted as \textsc{BSC}($p_1$) and a second channel \textsc{BSC}($p_2$). Assume that $p_1 < p_2 < \frac{1}{2}$ which implies stochastic degradation as defined in~\eqref{eqn:StochasticDegraded}. For $\alpha \in [0, \frac{1}{2}]$, Cover's superposition inner bound is the region,
\begin{align}
\Bigl\{R_1, R_2~\Bigl|~ & R_1 \leq h_b(\alpha * p_1) - h_b(p_1), \notag \\
& R_2 \leq 1 - h_b(\alpha * p_2)\Bigl\} \label{eqn:RateRegionBSCDMBC}
\end{align}
The above inner bound is determined by evaluating~\eqref{eqn:DegradedBCRateRegion} where $V$ is a Bernoulli random variable with $P_{V}(v) = \frac{1}{2}$, $X = V \oplus S$, and $S$ is a Bernoulli random variable with $P_{S}(1) = \alpha$. Figure~\ref{fig:BroadcastChannelTwoBSCs} plots this rectangular inner bound for two different values $\alpha = \frac{1}{10}$ and $\alpha = \frac{1}{4}$. The corner points of this rectangle given in~\eqref{eqn:RateRegionBSCDMBC} lie on the capacity boundary.
%The rate region~\eqref{eqn:RateRegionBSCDMBC} is also the capacity region and is determined by %evaluating~\eqref{eqn:DegradedBCRateRegion} where $U$ is a Bernoulli random variable with $\mathbb{P}\{U = 1\} = \frac{1}{2}$, $X = U %\oplus S$, and $S$ is a Bernoulli random variable with $\mathbb{P}\{S = 1\} = \alpha$.
\end{example}

\vspace{0.15in}
\begin{example}[DM-BC with \textsc{BEC}($\epsilon$) and \textsc{BSC}($p$)\cite{nair10}]\label{ex:DMBCWithBECAndBSC}
Consider a two-user DM-BC comprised of a \textsc{BSC}($p$) from $X$ to $Y_1$ and a \textsc{BEC}($\epsilon$) from $X$ to $Y_2$. Then it can be shown that the following cases hold:
\begin{itemize}
\item $0 < \epsilon \leq 2p$: $Y_1$ is degraded with respect to $Y_2$.
\item $2p < \epsilon \leq 4p(1-p)$: $Y_2$ is less noisy than $Y_1$ but $Y_1$ is not degraded with respect to $Y_2$.
\item $4p(1-p) < \epsilon \leq h_b(p)$: $Y_2$ is more capable than $Y_1$ but not less noisy.
\item $h_b(p) < \epsilon < 1$: The channel does not belong to the special classes.
\end{itemize}
The capacity region for all channel parameters for this example is achieved using superposition coding.
\end{example}

\subsection{Main Result}

\begin{theorem}[Polarization-Based Superposition Code]\label{thm:CoverCode} Consider any two-user DM-BC with binary input alphabet $\mathcal{X} = \{0,1\}$ and arbitrary output alphabets $\mathcal{Y}_1$, $\mathcal{Y}_2$. There exists a sequence of polar broadcast codes over $n$ channel uses which achieves the following rate region
\begin{align}
\mathfrak{R}(V, X, Y_1, Y_2) \triangleq \Bigl\{R_1, R_2~\Bigl|~ & R_1 \leq I(X; Y_1| V), \notag \\
& R_2 \leq I(V; Y_2) \Bigl\}, \label{eqn:MartonAchievable1}
\end{align}
where random variables $V, X, Y_1, Y_2$ have the following listed properties:
\begin{itemize}
\item $V$ is a binary random variable.
\item $P_{Y_1|V}(y_1|v) \degBC P_{Y_2|V}(y_2|v)$.
\item $V-X-(Y_1,Y_2)$ form a Markov chain.
%\item The joint distribution of all random variables is given by
%\begin{align}
%& P_{VXY_1Y_2}(v, x, y_1, y_2) = P_{VX}(v,x)P_{Y_1Y_2|X}(y_1, y_2 | x). \notag
%\end{align}
\end{itemize}
For $0 < \beta < \frac{1}{2}$, the average error probability of this code sequence decays as $P_e^{(n)} = \mathcal{O}(2^{-n^{\beta}})$. The complexity of encoding and decoding is $\mathcal{O}(n \log n)$.
\end{theorem}
\vspace{0.05in}
\input{./BroadcastChannelTwoBSCs}
\begin{remark} The requirement that auxiliary $V$ is a binary random variable is due to the use of binary polarization theorems in the proof. Indeed, the auxiliary $V$ may need to have a larger alphabet in the case of broadcast channels. An extension to $q$-ary random variables is entirely possible if $q$-ary polarization theorems are utilized.
\end{remark}
\begin{remark}
The requirement that $V-X-(Y_1,Y_2)$ holds is standard for superposition coding over noisy channels. However, the listed property $P_{Y_1|V}(y_1|v) \degBC P_{Y_2|V}(y_2|v)$ is due to the structure of polarization and is used in the proof to guarantee that polarization indices are \emph{aligned}. If both receivers are able to decode the coarse message carried by the auxiliary random variable $V$, the polarization indices for the coarse message must be nested for the two receivers' channels.
\end{remark}
\vspace{0.1in}
\begin{figure*}[t]
\begin{center}
%\psset{unit=0.50mm}
%\begin{pspicture}(-20,-8)(150,48)

%\small

% Generated with LaTeXDraw 2.0.0
% Tue Dec 04 19:03:58 CET 2012
% \usepackage[usenames,dvipsnames]{pstricks}
% \usepackage{epsfig}
% \usepackage{pst-grad} % For gradients
% \usepackage{pst-plot} % For axes
\scalebox{0.84} % Change this value to rescale the drawing.
{
\begin{pspicture}(0,-2.6)(16.915,1.6)

\rput(0, -0.62) {

\usefont{T1}{ptm}{m}{n}
\rput(8.297969,0.31){$X^n$}
\psframe[linewidth=0.04,dimen=outer](12.3,1.3)(9.1,-1.3)
\usefont{T1}{ptm}{m}{n}
\rput(10.649531,0.01){$P_{Y_1 Y_2 | X}(y_1, y_2 | x)$}
\psline[linewidth=0.03cm](3.6,-0.1)(5.6,-0.1)
\psline[linewidth=0.03cm,arrowsize=0.05291667cm 2.0,arrowlength=1.4,arrowinset=0.4]{->}(6.9,0.0)(9.1,0.0)
\psline[linewidth=0.03cm,arrowsize=0.05291667cm 2.0,arrowlength=1.4,arrowinset=0.4]{->}(12.3,1.0)(13.9,1.0)
\psline[linewidth=0.03cm,arrowsize=0.05291667cm 2.0,arrowlength=1.4,arrowinset=0.4]{->}(12.3,-1.1)(13.9,-1.1)
\psframe[linewidth=0.04,dimen=outer](15.1,-0.7)(13.9,-1.4)
\usefont{T1}{ptm}{m}{n}
\rput(14.509375,-1.09){$\mathcal{D}_2$}
\usefont{T1}{ptm}{m}{n}
\rput(12.9975,1.31){$Y_1^n$}
\usefont{T1}{ptm}{m}{n}
\rput(13.076875,-0.79){$Y_2^{n}$}
\usefont{T1}{ptm}{m}{n}
\rput(0.52875,1.11){$W_1$}
\usefont{T1}{ptm}{m}{n}
\rput(16.236563,1.31){$\hat{W}_1$}
\psline[linewidth=0.03cm,arrowsize=0.05291667cm 2.0,arrowlength=1.4,arrowinset=0.4]{->}(15.1,-1.1)(16.9,-1.1)
\usefont{T1}{ptm}{m}{n}
\rput(16.236563,-0.79){$\hat{W}_2$}
\usefont{T1}{ptm}{m}{n}
\rput(4.2028127,-0.89){$\matbold{G}_n$}
\psline[linewidth=0.03cm,arrowsize=0.05291667cm 2.0,arrowlength=1.4,arrowinset=0.4]{->}(0.0,-0.9)(1.4,-0.9)
\usefont{T1}{ptm}{m}{n}
\rput(0.52875,-0.59){$W_2$}
\psline[linewidth=0.03cm](4.8,-0.9)(5.6,-0.9)
\psline[linewidth=0.03cm](5.6,-0.9)(5.6,-0.1)
\usefont{T1}{ptm}{m}{n}
\rput(5.1876564,-0.59){$V^n$}
\psframe[linewidth=0.04,linestyle=solid,framearc=0.2,dimen=outer](7.5,1.6)(1.1,-1.6)
\psframe[linewidth=0.04,dimen=outer](4.8,-0.6)(3.6,-1.2)
\psframe[linewidth=0.04,dimen=outer](15.1,1.4)(13.9,0.7)
\usefont{T1}{ptm}{m}{n}
\rput(14.4948435,1.01){$\mathcal{D}_1$}
\psline[linewidth=0.03cm,arrowsize=0.05291667cm 2.0,arrowlength=1.4,arrowinset=0.4]{->}(15.1,1.0)(16.9,1.0)
\psframe[linewidth=0.04,linestyle=solid,framearc=0.2,dimen=outer](15.4,1.6)(13.6,0.5)
\psframe[linewidth=0.04,linestyle=solid,framearc=0.2,dimen=outer](15.4,-0.5)(13.6,-1.6)
\psframe[linewidth=0.04,dimen=outer](2.6,-0.6)(1.4,-1.2)
\psline[linewidth=0.03cm,arrowsize=0.05291667cm 2.0,arrowlength=1.4,arrowinset=0.4]{->}(2.6,-0.9)(3.6,-0.9)
\usefont{T1}{ptm}{m}{n}
\rput(3.0782812,-0.59){$U_2^n$}
\usefont{T1}{ptm}{m}{n}
\rput(1.98875,-0.89){$\mathcal{E}_2$}
\psframe[linewidth=0.04,dimen=outer](4.2,1.1)(3.0,0.5)
\usefont{T1}{ptm}{m}{n}
\rput(3.5742188,0.81){$\mathcal{E}_1$}
\psframe[linewidth=0.04,dimen=outer](6.4,1.1)(5.2,0.5)
\usefont{T1}{ptm}{m}{n}
\rput(5.8028126,0.81){$\matbold{G}_n$}
\psline[linewidth=0.03cm](6.9,0.0)(6.9,0.8)
\psline[linewidth=0.03cm](6.4,0.8)(6.9,0.8)
\usefont{T1}{ptm}{m}{n}
\rput(4.6782813,1.11){$U_1^n$}
\psline[linewidth=0.03cm,arrowsize=0.05291667cm 2.0,arrowlength=1.4,arrowinset=0.4]{->}(4.2,0.8)(5.2,0.8)
\psline[linewidth=0.03cm,arrowsize=0.05291667cm 2.0,arrowlength=1.4,arrowinset=0.4]{->}(0.0,0.8)(3.0,0.8)
\psline[linewidth=0.03cm,arrowsize=0.05291667cm 2.0,arrowlength=1.4,arrowinset=0.4]{->}(3.6,-0.1)(3.6,0.5)

}

\end{pspicture}
}

\end{center}\vspace{-0.05in}
\caption{ Block diagram of a polarization-based superposition code for a two-user noisy broadcast channel. } \label{fig:CoverCoding}
\end{figure*}
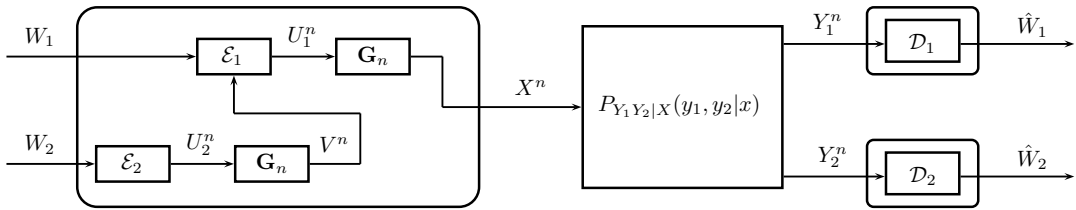
\section{Proof of Theorem~\ref{thm:CoverCode}}\label{section:ProofOfCoverPolarCode}

The block diagram for polarization-based superposition coding is given in Figure~\ref{fig:CoverCoding}. Similar to random codes in Shannon theory, polarization-based codes rely on $n$-length $i.i.d.$ statistics of random variables; however, a specific polarization structure based on the chain rule of entropy allows for efficient encoding and decoding. The key idea of Cover's inner bound is to superimpose two messages of information onto one codeword.

\subsection{Polar Transform}\label{sec:PolarTransformSuperposition}

Consider the $i.i.d.$ sequence of random variables $(V^{j}, X^{j}, Y_1^{j}, Y_2^{j}) \sim P_{V}(v)P_{X|V}(x|v)P_{Y_1Y_2|X}(y_1, y_2|x)$ where the index $j \in [n]$. Let the $n$-length sequence of auxiliary and input variables $(V^{j}, X^{j})$ be organized into the random matrix
\vspace{0.1in}
\begin{align}
\matbold{\Omega} & \triangleq \left[\begin{array}{ccccc} X^{1} & X^{2} & X^{3} & \ldots & X^{n} \\ V^{1} & V^{2} & V^{3} & \ldots & V^{n} \end{array}\right]. \label{eqn:AMatSuperposition}
\end{align}
Applying the polar transform to $\matbold{\Omega}$ results in the random matrix $\matbold{U} \triangleq \matbold{\Omega}\matbold{G}_n$. Let the random variables in the random matrix $\matbold{U}$ be indexed as follows:
\begin{align}
\matbold{U} & = \left[\begin{array}{ccccc} U_1^{1} & U_1^{2} & U_1^{3} & \ldots & U_1^{n} \\ U_2^{1} & U_2^{2} & U_2^{3} & \ldots & U_2^{n} \end{array}\right]. \label{eqn:UMatSuperposition}
\end{align}
The above definitions are consistent with the block diagram given in Figure~\ref{fig:CoverCoding} (and noting that $\matbold{G}_n = \matbold{G}_n^{-1}$). The polar transform extracts the randomness of $\matbold{\Omega}$. In the transformed domain, the joint distribution of the random variables in $\matbold{U}$ is given by
\begin{align}
& P_{U_1^{n}U_2^{n}}\bigl(u_1^{n}, u_2^{n}\bigl) \triangleq P_{X^{n}V^{n}}\bigl(u_1^{n}\matbold{G}_n, u_2^{n}\matbold{G}_n\bigl). \label{eqn:U1U2JointSuperposition}
\end{align}
For polar coding purposes, the joint distribution is decomposed as follows,
\begin{align}
& P_{U_1^{n}U_2^{n}}\bigl(u_1^{n}, u_2^{n}\bigl) = P_{U_2^{n}}(u_2^{n})P_{U_1^{n}|U_2^{n}}\bigl(u_1^{n}\bigl|u_2^{n}\bigl) \notag \\
& = \prod_{j=1}^{n} P\bigl(u_2(j) \bigl| u_2^{1:j-1}\bigl) P\bigl(u_1(j) \bigl| u_1^{1:j-1}, u_2^{n}\bigl).
\label{eqn:U1U2DecomposedDistributionSuperposition}
\end{align}
The conditional distributions may be computed efficiently using recursive protocols as already mentioned. The polarized variables in $\matbold{U}$ are \emph{not} \emph{i.i.d.} random variables.

%The joint distribution of the variables $(X^{n}, V^{n})$ is given by
%\begin{align}
%P_{X^{n}V^{n}}\Bigl(x^{n}, v^{n}\Bigl) & = \prod_{j = 1}^{n} P_{XV}\Bigl(x(j), v(j)\Bigl). \label{eqn:XVDistributionSuperposition}
%\end{align}
%The joint distribution of $(U_1^{n}, U_2^{n})$ is given by
%\begin{align}
%P_{U_1^{n}U_2^{n}}\Bigl(u_1^{n}, u_2^{n}\Bigl) & = P_{X^{n}V^{n}}\Bigl(u_1^{n}\matbold{G}_n, u_2^{n}\matbold{G}_n\Bigl). \label{eqn:U1U2DistributionSuperposition}
%\end{align}

\subsection{Polarization Theorems Revisited}

%The following polarization sets are defined for superposition coding over the two-user DM-BC.
\begin{definition}[Polarization Sets for Superposition Coding]\label{def:PolarizationSetsSuperposition} Let $V^{n}, X^{n}, Y_1^{n}, Y_2^{n}$ be the sequence of random variables as introduced in Section~\ref{sec:PolarTransformSuperposition}. In addition, let $U_1^{n} = X^{n}\matbold{G}_n$ and $U_2^{n} = V^{n}\matbold{G}_n$. Let $\delta_n = 2^{-n^{\beta}}$ for $0 < \beta < \frac{1}{2}$. The following \emph{polarization sets} are defined:
\begin{align}
\mathcal{H}^{(n)}_{X|V} & \triangleq \Bigl\{j \in [n]: Z\left(U_1(j) \Bigl| U_1^{1:j-1}, V^{n}\right) \geq 1 - \delta_n \Bigl\}, \notag \\
\mathcal{L}^{(n)}_{X|VY_1} & \triangleq \Bigl\{j \in [n]: Z\left(U_1(j) \Bigl| U_1^{1:j-1}, V^{n}, Y_1^{n}\right) \leq \delta_n \Bigl\}, \notag \\
\mathcal{L}^{(n)}_{V|Y_1} & \triangleq \Bigl\{j \in [n]: Z\left(U_2(j) \Bigl| U_2^{1:j-1}, Y_1^{n}\right) \leq \delta_n \Bigl\}. \notag \\
\mathcal{H}^{(n)}_{V} & \triangleq \Bigl\{j \in [n]: Z\left(U_2(j) \Bigl| U_2^{1:j-1}\right) \geq 1 - \delta_n \Bigl\}, \notag \\
\mathcal{L}^{(n)}_{V|Y_2} & \triangleq \Bigl\{j \in [n]: Z\left(U_2(j) \Bigl| U_2^{1:j-1}, Y_2^{n}\right) \leq \delta_n \Bigl\}. \notag
\end{align}
\end{definition}

\begin{definition}[Message Sets for Superposition Coding]\label{def:MessageSetsSuperposition} In terms of the polarization sets given in Definition~\ref{def:PolarizationSetsSuperposition}, the following \emph{message sets} are defined:
\begin{align}
\mathcal{M}^{(n)}_{1v} & \triangleq \mathcal{H}^{(n)}_{V} \cap \mathcal{L}^{(n)}_{V|Y_1}, \label{eqn:MessageIndicesSuperpositionM1V} \\
\mathcal{M}^{(n)}_{1} & \triangleq \mathcal{H}^{(n)}_{X|V} \cap \mathcal{L}^{(n)}_{X|VY_1}. \label{eqn:MessageIndicesSuperpositionM1} \\
\mathcal{M}^{(n)}_2 & \triangleq \mathcal{H}^{(n)}_{V} \cap \mathcal{L}^{(n)}_{V|Y_2}. \label{eqn:MessageIndicesSuperpositionM2}
\end{align}
\end{definition}
%given in~\eqref{eqn:AMatSuperposition} and~\eqref{eqn:UMatSuperposition}.
%The following theorem identifies the polarization indices necessary for broadcast superposition coding.
\vspace{0.05in}
\begin{proposition}[Polarization]\label{thm:RateOfPolarizationSuperposition} Consider the polarization sets given in Definition~\ref{def:PolarizationSetsSuperposition} and the message sets given in Definition~\ref{def:MessageSetsSuperposition} with parameter $\delta_n = 2^{-n^{\beta}}$ for $0 < \beta < \frac{1}{2}$. Fix a constant $\tau > 0$. Then there exists an $N_o = N_o(\beta, \tau)$ such that
\begin{align}
\frac{1}{n}\left|\mathcal{M}^{(n)}_{1}\right| & \geq \Bigl(H(X|V) - H(X|V, Y_1)\Bigl) - \tau, \label{eqn:MessageSetAchievesCapacitySuperpositionM1} \\
\frac{1}{n}\left|\mathcal{M}^{(n)}_{2}\right| & \geq \Bigl(H(V) - H(V|Y_2)\Bigl) - \tau, \label{eqn:MessageSetAchievesCapacitySuperpositionM2}
\end{align}
for all $n > N_o$.
\end{proposition}
\vspace{0.1in}
\begin{lemma}\label{lemma:NestedSetsSuperpostionCoding} Consider the message sets defined in Definition~\ref{def:MessageSetsSuperposition}. If the property $P_{Y_1|V}(y_1|v) \degBC P_{Y_2|V}(y_2|v)$ holds for conditional distributions $P_{Y_1|V}(y_1|v)$ and $P_{Y_2|V}(y_2|v)$, then the Bhattacharyya parameters
\begin{align}
Z\left(U_2(j) \Bigl| U_2^{1:j-1}, Y_1^{n}\right) & \leq Z\left(U_2(j) \Bigl| U_2^{1:j-1}, Y_2^{n}\right) \notag
\end{align}
for all $j \in [n]$. As a result,
\begin{align}
\mathcal{M}^{(n)}_{2} \subseteq \mathcal{M}^{(n)}_{1v}. \notag
\end{align}
\end{lemma}
\begin{proof} The proof follows from Lemma~\ref{lemma:DegradationBhatt} and repeated application of Lemma~\ref{lemma:SuccessiveDegradationBhatt} in Appendix~\ref{sec:PolarCodingLemmas}.
\end{proof}
%\begin{align}
%\mathcal{B}^{(n)}_{V|Y_2} \subset \mathcal{A}^{(n)}_{V|Y_1}. \notag
%\end{align}
%Lemma~\ref{lemma:NestedSetsSuperpostionCoding} ensures that \emph{both} receivers are able to decode the message %$W_2$ using polar decoding even though $W_2$ is intended only for the second receiver.

\subsection{Broadcast Encoding Based on Polarization}
\label{sec:EncodingDecodingSuperposition}

The polarization theorems of the previous section are useful for defining a multi-user communication system as diagrammed in Figure~\ref{fig:CoverCoding}. The broadcast encoder must map two independent messages $(W_1, W_2)$ uniformly distributed over $[2^{nR_1}] \times [2^{nR_2}]$ to a codeword $x^{n} \in \mathcal{X}^{n}$ in such a way that the decoding at each separate receiver is successful. The achievable rates for a particular block length $n$ are
\begin{align}
R_1 & = \frac{1}{n}\left|\mathcal{M}^{(n)}_{1}\right|, \notag \\
R_2 & = \frac{1}{n}\left|\mathcal{M}^{(n)}_{2}\right|. \notag
\end{align}

To construct a codeword, the encoder first produces two binary sequences $u_1^{n} \in \{0,1\}^{n}$ and $u_2^{n} \in \{0,1\}^{n}$. To determine $u_1(j)$ for $j \in \mathcal{M}^{(n)}_{1}$, the bit is selected as a uniformly distributed message bit intended for the first receiver. To determine $u_2(j)$ for $j \in \mathcal{M}^{(n)}_{2}$, the bit is selected as a uniformly distributed message bit intended for the second receiver. The remaining \emph{non-message} indices of $u_1^{n}$ and $u_2^{n}$ are computed according to deterministic or random functions which are \emph{shared} between the encoder and decoder.

\vspace{0.1in}
\subsubsection{Deterministic Mapping} Consider the following deterministic boolean functions indexed by $j \in [n]$:
\begin{align}
\psi_1^{(j)}: \{0,1\}^{n+j-1} \rightarrow \{0,1\}, \label{eqn:BooleanMapsDetSuperposition1} \\
\psi_2^{(j)}: \{0,1\}^{j-1} \rightarrow \{0,1\}.  \label{eqn:BooleanMapsDetSuperposition2}
\end{align}
As an example, consider the deterministic boolean functions based on the \emph{maximum a posteriori} polar coding rule.
\begin{align}
& \psi^{(j)}_{1}\left(u_1^{1:j-1}, v^{n}\right) \triangleq \notag \\
& ~~ \argmax_{u \in \{0,1\}} \Bigl\{ \mathbb{P}\left(U_1(j) = u \Bigl| U_1^{1:j-1} = u_1^{1:j-1}, V^{n} = v^{n} \right)\Bigl\}. \label{eqn:DetMapSpecificSuperposition1} \\
& \psi^{(j)}_{2}\left(u_2^{1:j-1}\right) \triangleq \notag \\
& ~~ \argmax_{u \in \{0,1\}} \Bigl\{ \mathbb{P}\left(U_2(j) = u \Bigl| U_2^{1:j-1} = u_2^{1:j-1} \right)\Bigl\}. \label{eqn:DetMapSpecificSuperposition2}
\end{align}

\subsubsection{Random Mapping} Consider the following class of random boolean functions indexed by $j \in [n]$:
\begin{align}
\Psi_1^{(j)}: \{0,1\}^{n+j-1} \rightarrow \{0,1\}, \label{eqn:BooleanMapsRANDSuperposition1} \\
\Psi_2^{(j)}: \{0,1\}^{j-1} \rightarrow \{0,1\}. \label{eqn:BooleanMapsRANDSuperposition2}
\end{align}
As an example, consider the random boolean functions
\begin{align}
\Psi_1^{(j)}\bigl(u_1^{1:j-1}, v^{n} \bigl) & \triangleq \begin{cases} 0, & \mbox{w.p. }~ \lambda_{0}\bigl(u_1^{1:j-1}, v^{n}\bigl), \\ 1, & \mbox{w.p. }~1-\lambda_{0}\bigl(u_1^{1:j-1}, v^{n}\bigl), \end{cases} \label{eqn:RANDMapSpecificSuperposition1} \\
\Psi_2^{(j)}\bigl(u_2^{1:j-1} \bigl) & \triangleq \begin{cases} 0, & \mbox{w.p. }~ \lambda_{0}\bigl(u_2^{1:j-1} \bigl), \\ 1, & \mbox{w.p. }~1-\lambda_{0}\bigl(u_2^{1:j-1} \bigl), \end{cases} \label{eqn:RANDMapSpecificSuperposition2}
\end{align}
where
\begin{align}
\lambda_{0}\bigl(u_2^{1:j-1} \bigl) & \triangleq \mathbb{P}\bigl(U_2(j) = 0 \bigl| U_2^{1:j-1} = u_2^{1:j-1}\bigl). \notag \\
\lambda_{0}\bigl(u_1^{1:j-1}, v^{n} \bigl) & \triangleq \notag \\
& \!\!\!\! \mathbb{P}\bigl(U_1(j) = 0 \bigl| U_1^{1:j-1} = u_1^{1:j-1}, V^{n} = v^{n}\bigl). \notag
\end{align}
The random boolean functions $\Psi_1^{(j)}$ and $\Psi_2^{(j)}$ may be thought of as a vector of independent Bernoulli random variables indexed by the input to the function. Each Bernoulli random variable of the vector is zero or one with a fixed probability.
\vspace{0.05in}
\subsubsection{Protocol} The encoder constructs the sequence $u_2^{n}$ first using the message bits $W_2$ and either~\eqref{eqn:DetMapSpecificSuperposition2} or~\eqref{eqn:RANDMapSpecificSuperposition2}. Next, the sequence $v^{n} = u_2^{n}\matbold{G}_n$ is created. Finally, the sequence $u_1^{n}$ is constructed using the message bits $W_1$, the sequence $v^{n}$, and either the deterministic maps defined in~\eqref{eqn:DetMapSpecificSuperposition1} or the randomized maps in~\eqref{eqn:RANDMapSpecificSuperposition1}. The transmitted codeword is $x^{n} = u_1^{n}\matbold{G}_n$.

\subsection{Broadcast Decoding Based on Polarization}
\label{sec:DecodingSuperposition}

\subsubsection{Decoding At First Receiver}

Decoder $\mathcal{D}_1$ decodes the binary sequence $\hat{u}_2^{n}$ first using its observations $y_1^{n}$. It then reconstructs $\hat{v}^{n} = \hat{u}_2^{n}\matbold{G}_n$. Using the sequence $\hat{v}^{n}$ and observations $y_1^{n}$, the decoder reconstructs $\hat{u}_1^{n}$. The message $W_1$ is located at the indices $j \in \mathcal{M}^{(n)}_{1}$ in the sequence $\hat{u}_1^{n}$. More precisely, define the following deterministic polar decoding functions:
\begin{align}
& \xi^{(j)}_{v}\left(u_2^{1:j-1}, y_1^{n}\right) \triangleq \notag \\
& ~~\argmax_{u \in \{0,1\}} \Bigl\{ \mathbb{P}\left(U_2(j) = u \Bigl| U_2^{1:j-1} = u_2^{1:j-1}, Y_1^{n} = y_1^{n} \right)\Bigl\}. \label{eqn:DecodingSuperposition1a} \\
& \xi^{(j)}_{u_1}\left(u_1^{1:j-1}, v^{n}, y_1^{n}\right) \triangleq \argmax_{u \in \{0,1\}} \Bigl\{ \notag \\
& ~~\mathbb{P}\left(U_1(j) = u \Bigl| U_1^{1:j-1} = u_1^{1:j-1}, V^{n} = v^{n}, Y_1^{n} = y_1^{n} \right)\Bigl\}. \label{eqn:DecodingSuperposition1b}
\end{align}
The decoder $\mathcal{D}_1$ reconstructs $\hat{u}_2^{n}$ bit-by-bit successively as follows using the \emph{identical} shared random mapping $\Psi_2^{(j)}$ (or possibly the identical shared mapping $\psi_2^{(j)}$) used at the encoder:
\begin{align}
\hat{u}_2(j) & = \begin{cases} \xi^{(j)}_{v}\left(\hat{u}_2^{1:j-1}, y_1^{n}\right), & \mbox{\emph{if}}~ j \in \mathcal{M}^{(n)}_{2}, \\ \Psi_2^{(j)}\left(\hat{u}_2^{1:j-1} \right), & \mbox{\emph{otherwise.}}~\end{cases} \label{eqn:Decoder1ForV}
\end{align}
If Lemma~\ref{lemma:NestedSetsSuperpostionCoding} holds, note that $\mathcal{M}^{(n)}_{2} \subseteq \mathcal{M}^{(n)}_{1v}$. With $\hat{u}_2^{n}$, decoder $\mathcal{D}_1$ reconstructs $\hat{v}^{n} = \hat{u}_2^{n}\matbold{G}_n$. Then the sequence $\hat{u}_1^{n}$ is constructed bit-by-bit successively as follows using the \emph{identical} shared random mapping $\Psi_1^{(j)}$ (or possibly the identical shared mapping $\psi_1^{(j)}$) used at the encoder:
\begin{align}
\hat{u}_1(j) & = \begin{cases} \xi^{(j)}_{u_1}\left(\hat{u}_1^{1:j-1}, \hat{v}^{n}, y_1^{n}\right), & \mbox{\emph{if}}~ j \in \mathcal{M}^{(n)}_{1}, \\ \Psi_1^{(j)}\left(\hat{u}_1^{1:j-1}, \hat{v}^{n} \right), & \mbox{\emph{otherwise.}}~\end{cases} \label{eqn:Decoder1ForU1}
\end{align}

\vspace{0.1in}
\subsubsection{Decoding At Second Receiver}

The decoder $\mathcal{D}_2$ decodes the binary sequence $\hat{u}_2^{n}$ using observations $y_2^{n}$. The message $W_2$ is located at the indices $j \in \mathcal{M}^{(n)}_{2}$ of the sequence $\hat{u}_2^{n}$. More precisely, define the following polar decoding functions
\begin{align}
& \xi^{(j)}_{v}\left(u_2^{1:j-1}, y_2^{n}\right) \triangleq \notag \\
& ~~\argmax_{u \in \{0,1\}} \Bigl\{ \mathbb{P}\left(U_2(j) = u \Bigl| U_2^{1:j-1} = u_2^{1:j-1}, Y_2^{n} = y_2^{n} \right)\Bigl\}. \label{eqn:DecodingSuperposition2}
\end{align}
The decoder $\mathcal{D}_2$ reconstructs $\hat{u}_2^{n}$ bit-by-bit successively as follows using the \emph{identical} shared random mapping $\Psi_2^{(j)}$ (or possibly the identical shared mapping $\psi_2^{(j)}$) used at the encoder:
\begin{align}
\hat{u}_2(j) & = \begin{cases} \xi^{(j)}_{v}\left(\hat{u}_2^{1:j-1}, y_2^{n}\right), & \mbox{\emph{if}}~ j \in \mathcal{M}^{(n)}_{2}, \\ \Psi_2^{(j)}\left(\hat{u}_2^{1:j-1} \right), & \mbox{\emph{otherwise.}}~\end{cases} \label{eqn:Decoder2ForV}
\end{align}

\begin{remark} The encoder and decoders execute \emph{the same} protocol for reconstructing bits at the \emph{non-message} indices. This is achieved by applying the same deterministic maps $\psi_1^{(j)}$ and $\psi_2^{(j)}$ or randomized maps $\Psi_1^{(j)}$ and $\Psi_2^{(j)}$.
\end{remark}
%A broadcast  error occurs if either $\mathcal{D}_1$ or $\mathcal{D}_2$ makes an error on the message bits.

\subsection{Total Variation Bound}

To analyze the average probability of error $P_e^{(n)}$ via the probabilistic method, it is assumed that both the encoder and decoder share the \emph{randomized} mappings $\Psi_1^{(j)}$ and $\Psi_2^{(j)}$. Define the following probability measure on the space of tuples of binary sequences.
\begin{align}
& Q\bigl(u_1^{n}, u_2^{n}\bigl) \triangleq Q\bigl(u_2^{n}\bigl)Q\bigl(u_1^{n}\bigl|u_2^{n}\bigl) \notag \\
& = \prod_{j=1}^{n} Q\Bigl(u_2(j) \Bigl| u_2^{1:j-1}\Bigl) Q\Bigl(u_1(j) \Bigl| u_1^{1:j-1}, u_2^{n}\Bigl).
%& \prod_{i=1}^{m}\prod_{j=1}^{n} \mathbb{P}(U_{(i,j)} = u_{(i,j)} | U_{(i,[j-1])} = u_{(i,[j-1])}, U_{([i-1],[n])} = u_{([i-1],[n])}).
\label{eqn:QUJOINTSuperposition}
\end{align}
In~\eqref{eqn:QUJOINTSuperposition}, the conditional probability measures are defined as
\begin{align}
& Q\Bigl(u_2(j) \Bigl| u_2^{1:j-1}\Bigl) \triangleq \notag \\
& ~~~~~~ \begin{cases} \frac{1}{2}, & \mbox{\emph{if}}~j \in \mathcal{M}_2^{(n)}, \\ P\left( u_2(j) \Bigl| u_2^{1:j-1}\right), & \mbox{\emph{otherwise.}} \end{cases} \notag \\
& Q\Bigl(u_1(j) \Bigl| u_1^{1:j-1}, u_2^{n}\Bigl) \triangleq \notag \\
& ~~~~~~ \begin{cases} \frac{1}{2}, & \mbox{\emph{if}}~j \in \mathcal{M}_1^{(n)}, \\ P\left( u_1(j) \Bigl| u_1^{1:j-1}, u_2^{n}\right), & \mbox{\emph{otherwise.}} \end{cases} \notag
\end{align}
The probability measure $Q$ defined in~\eqref{eqn:QUJOINTSuperposition} is a perturbation of the joint probability measure $P_{U_1^{n}U_2^{n}}(u_1^{n}, u_2^{n})$ in~\eqref{eqn:U1U2DecomposedDistributionSuperposition}. The only difference in definition between $P$ and $Q$ is due to those indices in message sets $\mathcal{M}_1^{(n)}$ and $\mathcal{M}_2^{(n)}$. The following lemma provides a bound on the total variation distance between $P$ and $Q$. The lemma establishes the fact that inserting uniformly distributed message bits in the proper indices $\mathcal{M}_1^{(n)}$ and $\mathcal{M}_2^{(n)}$ at the encoder \emph{does not} perturb the statistics of the $n$-length random variables too much.
\vspace{0.1in}
\begin{lemma}(\emph{Total Variation Bound})\label{lemma:TVBoundSuperposition}
Let probability measures $P$ and $Q$ be defined as in~\eqref{eqn:U1U2DecomposedDistributionSuperposition} and~\eqref{eqn:QUJOINTSuperposition} respectively. Let $0 < \beta < 1$. For sufficiently large $n$, the total variation distance between $P$ and $Q$ is bounded as
\begin{align}
\sum_{\begin{subarray}{c} u_1^{n} \in \{0,1\}^{n} \\ u_2^{n} \in \{0,1\}^{n} \end{subarray}} \Bigl| P_{U_1^{n}U_2^{n}}\bigl( u_1^{n}, u_2^{n} \bigl) - Q\bigl( u_1^{n}, u_2^{n} \bigl) \Bigl| \leq 2^{-n^{\beta}}. \notag
\end{align}
\end{lemma}
\begin{proof} See Section~\ref{sec:AppendixSuperposition} of the Appendices. \end{proof}
\vspace{0.05in}
%\begin{remark}
%Lemma~\ref{lemma:TVBound} shows that the total variation distance is bounded as $\mathcal{O}(2^{-n^{\beta}})$ for a constant number of %broadcast receivers $m$. The distribution $Q$ approximates the distribution $P$ as $n \rightarrow \infty$, and the randomness is %extracted into uniformly distributed bits over message indices $\mathcal{M}_{i}^{(n)}$.
%\end{remark}

\subsection{Error Sequences}

The decoding protocols for $\mathcal{D}_1$ and $\mathcal{D}_2$ were established in Section~\ref{sec:DecodingSuperposition}. To analyze the probability of error of successive cancelation (SC) decoding, consider the sequences $u_1^{n}$ and $u_2^{n}$ formed at the encoder, and the resulting observations $y_1^{n}$ and $y_2^{n}$ received by the decoders. It is convenient to group the sequences together and consider all tuples $(u_1^{n}, u_2^{n}, y_1^{n}, y_2^{n})$.

Decoder $\mathcal{D}_1$ makes an SC decoding error on the $j$-th bit for the following tuples:
\begin{align}
\mathcal{T}_{1v}^{j} & \triangleq \Bigl\{\bigl(u_1^{n}, u_2^{n}, y_1^{n}, y_2^{n}\bigl): \notag \\
& ~~ P_{U_2^{j}\bigl|U_2^{1:j-1}Y_1^{n}}\bigl(u_2(j)\bigl|u_2^{1:j-1},y_1^{n}\bigl) \leq \notag \\
& ~~ P_{U_2^{j}\bigl|U_2^{1:j-1}Y_1^{n}}\bigl(u_2(j) \oplus 1 \bigl| u_2^{1:j-1}, y_1^{n}\bigl)\Bigl\}, \notag \\
\mathcal{T}_{1}^{j} & \triangleq \Bigl\{\bigl(u_1^{n}, u_2^{n}, y_1^{n}, y_2^{n}\bigl): \notag \\
& ~~ P_{U_1^{j}\bigl|U_1^{1:j-1}V^{n}Y_1^{n}}\bigl(u_1(j)\bigl|u_1^{1:j-1}, u_2^{n}\matbold{G}_n, y_1^{n}\bigl) \leq \notag \\
& ~~ P_{U_1^{j}|U_1^{1:j-1}V^{n}Y_1^{n}}\bigl(u_1(j) \oplus 1 \bigl|u_1^{1:j-1}, u_2^{n}\matbold{G}_n, y_1^{n}\bigl)\Bigl\}.
\end{align}
The set $\mathcal{T}_{1v}^{j}$ represents those tuples causing an error at $\mathcal{D}_1$ in the case $u_2(j)$ is inconsistent with respect to observations $y_1^{n}$ and the decoding rule. The set $\mathcal{T}_{1}^{j}$ represents those tuples causing an error at $\mathcal{D}_1$ in the case $u_1(j)$ is inconsistent with respect to $v^{n} = u_2^{n}\matbold{G}_n$, observations $y_1^{n}$, and the decoding rule. Similarly, decoder $\mathcal{D}_2$ makes an SC decoding error on the $j$-th bit for the following tuples:
\begin{align}
\mathcal{T}_2^{j} & \triangleq \Bigl\{\bigl(u_1^{n}, u_2^{n}, y_1^{n}, y_2^{n}\bigl): P_{U_2\bigl|U_2^{1:j-1}Y_2^{n}}\bigl(u_2\bigl|u_2^{1:j-1},y_2^{n}\bigl) \leq \notag \\
& ~~ P_{U_2\bigl|U_2^{1:j-1}Y_2^{n}}\bigl(u_2 \oplus 1 \bigl| u_2^{1:j-1}, y_2^{n}\bigl)\Bigl\}. \notag
\end{align}
The set $\mathcal{T}_2^{j}$ represents those tuples causing an error at $\mathcal{D}_2$ in the case $u_2(j)$ is inconsistent with respect to observations $y_2^{n}$ and the decoding rule. Since both decoders $\mathcal{D}_1$ and $\mathcal{D}_2$ only declare errors for those indices in the message sets, the set of tuples causing an error is
\begin{align}
\mathcal{T}_{1v} & \triangleq \bigcup_{j \in \mathcal{M}_2^{(n)} \subseteq \mathcal{M}_{1v}^{(n)}} \mathcal{T}_{1v}^{j}, \label{eqn:Total1vErrorTuple} \\
\mathcal{T}_{1} & \triangleq \bigcup_{j \in \mathcal{M}_1^{(n)}} \mathcal{T}_{1}^{j}, \label{eqn:Total1ErrorTuple} \\
\mathcal{T}_{2} & \triangleq \bigcup_{j \in \mathcal{M}_2^{(n)}} \mathcal{T}_{2}^{j}. \label{eqn:Total2ErrorTuple}
\end{align}
The complete set of tuples causing a broadcast error is
\begin{align}
\mathcal{T} \triangleq \mathcal{T}_{1v} \cup \mathcal{T}_{1} \cup \mathcal{T}_{2}. \label{eqn:TotalErrorTupleSuperposition}
\end{align}
The goal is to show that the probability of choosing tuples of error sequences in the set $\mathcal{T}$ is small under the distribution induced by the broadcast code.

\subsection{Average Error Probability}

Denote the total sum rate of the broadcast protocol as $R_{\Sigma} = R_1 + R_2$. Consider first the use of fixed deterministic maps $\psi_1^{(j)}$ and $\psi_2^{(j)}$ shared between the encoder and decoders. Then the probability of error of broadcasting the two messages at rates $R_1$ and $R_2$ is given by
\begin{align}
& P_e^{(n)}\left[\{\psi_1^{(j)}, \psi_2^{(j)}\}\right] = \notag \\
& ~~ \sum_{\{u_1^{n}, u_2^{n}, y_1^{n}, y_2^{n}\} \in \mathcal{T}} \Biggl[ P_{Y_1^{n}Y_2^{n}\bigl|U_1^{n}U_2^{n}}\bigl(y_1^{n}, y_2^{n}\bigl|u_1^{n}, u_2^{n}\bigl) \notag \\
& ~~ \cdot \frac{1}{2^{nR_2}} \prod_{j \in [n]: j \notin \mathcal{M}_2^{(n)}} \indicator{\psi_2^{(j)}\left(u_2^{1:j-1}\right) = u_2(j)} \notag \\
& ~~ \cdot \frac{1}{2^{nR_1}} \prod_{j \in [n]: j \notin \mathcal{M}_1^{(n)}} \indicator{\psi_1^{(j)}\left(u_1^{1:j-1}, u_2^{n}\matbold{G}_n \right) = u_1(j)}\Biggl]. \notag
\end{align}

If the encoder and decoders share randomized maps $\Psi_1^{(j)}$ and $\Psi_{2}^{(j)}$, then the average probability of error is a random quantity determined as follows
\begin{align}
& P_e^{(n)}\left[\{\Psi_1^{(j)}, \Psi_2^{(j)}\}\right] = \notag \\
& ~~ \sum_{\{u_1^{n}, u_2^{n}, y_1^{n}, y_2^{n}\} \in \mathcal{T}} \Biggl[ P_{Y_1^{n}Y_2^{n}\bigl|U_1^{n}U_2^{n}}\bigl(y_1^{n}, y_2^{n}\bigl|u_1^{n}, u_2^{n}\bigl) \notag \\
& ~~ \cdot \frac{1}{2^{nR_2}} \prod_{j \in [n]: j \notin \mathcal{M}_2^{(n)}} \indicator{\Psi_2^{(j)}\left(u_2^{1:j-1}\right) = u_2(j)} \notag \\
& ~~ \cdot \frac{1}{2^{nR_1}} \prod_{j \in [n]: j \notin \mathcal{M}_1^{(n)}} \indicator{\Psi_1^{(j)}\left(u_1^{1:j-1}, u_2^{n}\matbold{G}_n \right) = u_1(j)}\Biggl]. \notag
\end{align}
By averaging over the randomness in the encoders and decoders, the expected block error probability $P_e^{(n)}[\{\Psi_1^{(j)}, \Psi_2^{(j)}\}]$ is upper bounded in the following lemma.
\vspace{0.05in}
\begin{lemma}\label{theorem:ErrorProbSuperposition} Consider the polarization-based superposition code described in Section~\ref{sec:EncodingDecodingSuperposition} and Section~\ref{sec:DecodingSuperposition}. Let $R_1$ and $R_2$ be the broadcast rates selected according to the Bhattacharyya criterion given in Proposition~\ref{thm:RateOfPolarizationSuperposition}. Then for $0 < \beta < 1$ and sufficiently large $n$,
\begin{align}
\mathbb{E}_{\{\Psi_1^{(j)}, \Psi_2^{(j)}\}} \Bigl[ P_e^{(n)}[\{\Psi_1^{(j)}, \Psi_2^{(j)}\}] \Bigl] < 2^{-n^{\beta}}. \notag
\end{align}
\end{lemma}
\begin{proof} See Section~\ref{sec:AppendixSuperposition} of the Appendices.
\end{proof}
If the average probability of error decays to zero in expectation over the random maps $\{\Psi_1^{(j)}\}$ and $\{\Psi_2^{(j)}\}$, then there must exist at least one fixed set of maps for which $P_e^{(n)} \rightarrow 0$.
%Practically, we can also evaluate the fixed maps $\{\psi_1^{(j)}\}$ and $\{\psi_2^{(j)}\}$ in experiments.

%\begin{align}
%& \mathbb{E}_{\{\Psi^{(i,j)}\}} \Bigl[ P_e^{(n)}[\{\Psi^{(i,j)}\}] \Bigl] \notag \\
%& = \frac{1}{2^{nR_{\Sigma}}} \sum_{\{u_{k}^{1:n}\}_{k \in [m]}} \Biggl[  \indicator{(u_1^{n}, u_2^{n}, \ldots, u_m^{n}) \in \mathcal{\tilde{T}} } \cdot \notag \\
%& \prod_{\begin{subarray}{c} i \in [m] \\ j \in [n]: j \notin \mathcal{M}_i^{(n)} \end{subarray}} \mathbb{P}\left\{ \Psi^{(i,j)}\left(u_i^{1:j-1}, \{y_k^{1:n}\}_{k \in [1:i-1]}\right) = u_i(j) \right\} \Biggl] \notag \\
%& = \sum_{\{u_{k}^{1:n}\}_{k \in [m]} \in \tilde{\mathcal{T}}} Q\bigl(\{u_{k}^{1:n}\}_{k \in [m]}\bigl) \label{eqn:TheUseOfQ} \\
%& = \sum_{\{u_{k}^{1:n}\}_{k \in [m]} \in \tilde{\mathcal{T}}} \Bigl| P\bigl( \{u_{k}^{1:n}\}_{k \in [m]}\bigl) - Q\bigl(\{u_{k}^{1:n}\}_{k \in [m]}\bigl) \Bigl| \label{eqn:PHasZeroMassOverBadSeq} \\
%& \leq 2^{-n^{\beta}}. \label{eqn:UseOfLemma}
%\end{align}
%Step~\eqref{eqn:TheUseOfQ} follows since the probability measure $Q$ matches the desired calculation exactly. Step~\eqref{eqn:PHasZeroMassOverBadSeq} is due to the fact that the probability measure $P$ has zero mass over $m$-tuples of binary sequences that are inconsistent. Step~\eqref{eqn:UseOfLemma} follows directly from Lemma~\ref{lemma:TVBound}. Lastly, since the expectation over random maps $\{\Psi^{(i,j)}\}$ of the average probability of error decays stretched-exponentially, there must exist a set of deterministic maps which exhibit the same behavior.

\section{Noisy Broadcast Channels \\ Marton's Coding Scheme}

\subsection{Marton's Inner Bound} For general noisy broadcast channels, Marton's inner bound involves two correlated auxiliary random variables $V_1$ and $V_2$~\cite{marton79}. The intuition behind the coding strategy is to identify two ``virtual'' channels, one from $V_1$ to $Y_1$, and the other from $V_2$ to $Y_2$. Somewhat surprisingly, although the broadcast messages are independent, the auxiliary random variables $V_1$ and $V_2$ may be correlated to increase rates to both receivers. While there exist generalizations of Marton's strategy, the basic version of the inner bound is presented in this section\footnote{In addition, it is difficult even to evaluate Marton's inner bound for general channels due to the need for proper cardinality bounds on the auxiliaries~\cite{gohari2012}. These issues lie outside the scope of the present paper.}.
\vspace{0.05in}
\begin{proposition}[Marton's Inner Bound] For any two-user DM-BC, the rates $(R_1, R_2) \in \mathbb{R}_{+}^{2}$ in the pentagonal region $\mathfrak{R}(X, V_1, V_2, Y_1, Y_2)$ are achievable where
\begin{align}
\mathfrak{R}(X, V_1, V_2, & Y_1, Y_2) \triangleq \notag \\
\Bigl\{R_1, R_2~\Bigl|~ R_1 & \leq I(V_1; Y_1), \notag \\
R_2 & \leq I(V_2; Y_2), \notag \\
R_1 + R_2 & \leq I(V_1; Y_1) + I(V_2; Y_2) - I(V_1; V_2) \Bigl\}. \label{eqn:MartonRateRegion}
\end{align}
and where $X, V_1, V_2, Y_1, Y_2$ have a joint distribution given by $P_{V_1V_2}(v_1,v_2)P_{X|V_1V_2}(x|v_1, v_2)P_{Y_1Y_2|X}(y_1, y_2|x)$.
\end{proposition}
\vspace{0.1in}
\begin{remark} It can be shown that for Marton's inner bound there is no loss of generality if $P_{X|V_1V_2}(x|v_1,v_2) = \indicator{x = \phi(v_1, v_2)}$ where $\phi(v_1, v_2)$ is a deterministic function~\cite[Section~8.3]{elgamalkim2010}. Thus, by allowing a larger alphabet size for the auxiliaries, $X$ may be a deterministic function of auxiliaries $(V_1, V_2)$. Marton's inner bound is tight for the class of \emph{semi-deterministic} DM-BCs for which one of the outputs is a deterministic function of the input.
\end{remark}
\vspace{0.1in}
\begin{figure*}[t]
\begin{center}
%\psset{unit=0.50mm}
%\begin{pspicture}(-20,-8)(150,48)

%\small

% Generated with LaTeXDraw 2.0.0
% Tue Dec 04 18:04:38 CET 2012
% \usepackage[usenames,dvipsnames]{pstricks}
% \usepackage{epsfig}
% \usepackage{pst-grad} % For gradients
% \usepackage{pst-plot} % For axes
\scalebox{0.84} % Change this value to rescale the drawing.
{
\begin{pspicture}(0,-2.6)(16.615,1.6)

\rput(0, -0.62) {

\definecolor{color3775}{rgb}{0.0,0.0,0.2}
\definecolor{color3781}{rgb}{0.0,0.2,0.2}
\psframe[linewidth=0.04,dimen=outer](4.7,1.2)(3.5,0.6)
\usefont{T1}{ptm}{m}{n}
\rput(4.1028123,0.91){$\matbold{G}_n$}
\usefont{T1}{ptm}{m}{n}
\rput(8.7579685,0.31){$X^n$}
\psframe[linewidth=0.04,dimen=outer](12.4,1.3)(9.2,-1.3)
\usefont{T1}{ptm}{m}{n}
\rput(10.749531,0.01){$P_{Y_1 Y_2 | X}(y_1, y_2 | x)$}
\psline[linewidth=0.03cm,arrowsize=0.05291667cm 2.0,arrowlength=1.4,arrowinset=0.4]{->}(0.0,0.9)(1.5,0.9)
\psline[linewidth=0.03cm,arrowsize=0.05291667cm 2.0,arrowlength=1.4,arrowinset=0.4]{->}(8.2,0.0)(9.2,0.0)
\psline[linewidth=0.03cm,arrowsize=0.05291667cm 2.0,arrowlength=1.4,arrowinset=0.4]{->}(12.4,1.0)(14.0,1.0)
\psline[linewidth=0.03cm,arrowsize=0.05291667cm 2.0,arrowlength=1.4,arrowinset=0.4]{->}(12.4,-1.1)(14.0,-1.1)
\psframe[linewidth=0.04,dimen=outer](15.2,-0.7)(14.0,-1.4)
\usefont{T1}{ptm}{m}{n}
\rput(14.609375,-1.09){$\mathcal{D}_2$}
\usefont{T1}{ptm}{m}{n}
\rput(13.0975,1.31){$Y_1^n$}
\usefont{T1}{ptm}{m}{n}
\rput(13.176875,-0.79){$Y_2^n$}
\usefont{T1}{ptm}{m}{n}
\rput(0.42875,1.21){$W_1$}
\usefont{T1}{ptm}{m}{n}
\rput(16.036562,1.31){$\hat{W}_1$}
\psline[linewidth=0.03cm,arrowsize=0.05291667cm 2.0,arrowlength=1.4,arrowinset=0.4]{->}(15.2,-1.1)(16.6,-1.1)
\usefont{T1}{ptm}{m}{n}
\rput(16.036562,-0.79){$\hat{W}_2$}
\usefont{T1}{ptm}{m}{n}
\rput(4.1028123,-0.89){$\matbold{G}_n$}
\psline[linewidth=0.03cm,arrowsize=0.05291667cm 2.0,arrowlength=1.4,arrowinset=0.4]{->}(0.0,-0.9)(1.5,-0.9)
\usefont{T1}{ptm}{m}{n}
\rput(0.42875,-0.59){$W_2$}
\psframe[linewidth=0.04,dimen=outer](8.2,0.4)(5.8,-0.4)
\usefont{T1}{ptm}{m}{n}
\rput(6.970156,0.01){$x = \phi(v_1, v_2)$}
\psline[linewidth=0.03cm](4.7,0.9)(5.5,0.9)
\psline[linewidth=0.03cm](5.5,0.9)(5.5,0.2)
\psline[linewidth=0.03cm,arrowsize=0.05291667cm 2.0,arrowlength=1.4,arrowinset=0.4]{->}(5.5,0.2)(5.8,0.2)
\psline[linewidth=0.03cm](4.7,-0.9)(5.5,-0.9)
\psline[linewidth=0.03cm](5.5,-0.9)(5.5,-0.2)
\psline[linewidth=0.03cm,arrowsize=0.05291667cm 2.0,arrowlength=1.4,arrowinset=0.4]{->}(5.5,-0.2)(5.8,-0.2)
\usefont{T1}{ptm}{m}{n}
\rput(5.0776563,1.21){$V_1^n$}
\usefont{T1}{ptm}{m}{n}
\rput(5.0776563,-0.59){$V_2^n$}
%\psframe[linewidth=0.04,linecolor=color3775,linestyle=dashed,framearc=0.2,dash=0.16cm %0.16cm,dimen=outer](8.3,1.6)(1.0,-1.6)
\psframe[linewidth=0.04,linecolor=color3775,linestyle=solid,framearc=0.2,dimen=outer](8.4,1.6)(1.0,-1.6)
\psframe[linewidth=0.04,dimen=outer](4.7,-0.6)(3.5,-1.2)
\psframe[linewidth=0.04,dimen=outer](15.2,1.4)(14.0,0.7)
\usefont{T1}{ptm}{m}{n}
\rput(14.594844,1.01){$\mathcal{D}_1$}
\psline[linewidth=0.03cm,arrowsize=0.05291667cm 2.0,arrowlength=1.4,arrowinset=0.4]{->}(15.2,1.0)(16.6,1.0)
%\psframe[linewidth=0.04,linecolor=color3781,linestyle=dashed,framearc=0.2,dash=0.16cm %0.16cm,dimen=outer](15.5,1.6)(13.7,0.5)
\psframe[linewidth=0.04,linecolor=color3781,linestyle=solid,framearc=0.2,dimen=outer](15.5,1.6)(13.7,0.5)
\psframe[linewidth=0.04,linecolor=color3781,linestyle=solid,framearc=0.2,dimen=outer](15.5,-0.5)(13.7,-1.6)
\psframe[linewidth=0.04,dimen=outer](2.7,-0.6)(1.5,-1.2)
\psframe[linewidth=0.04,dimen=outer](2.7,1.2)(1.5,0.6)
\psline[linewidth=0.03cm,arrowsize=0.05291667cm 2.0,arrowlength=1.4,arrowinset=0.4]{->}(2.7,0.9)(3.5,0.9)
\psline[linewidth=0.03cm,arrowsize=0.05291667cm 2.0,arrowlength=1.4,arrowinset=0.4]{->}(2.7,-0.9)(3.5,-0.9)
\usefont{T1}{ptm}{m}{n}
\rput(3.0882812,1.21){$U_1^n$}
\usefont{T1}{ptm}{m}{n}
\rput(3.0882812,-0.59){$U_2^n$}
\usefont{T1}{ptm}{m}{n}
\rput(2.0742188,0.91){$\mathcal{E}_1$}
\usefont{T1}{ptm}{m}{n}
\rput(2.08875,-0.89){$\mathcal{E}_2$}
\psline[linewidth=0.02cm,linestyle=dashed,dash=0.16cm 0.16cm](2.2,0.2)(5.5,0.2)
\psline[linewidth=0.02cm,linestyle=dashed,dash=0.16cm 0.16cm,arrowsize=0.05291667cm 2.0,arrowlength=1.4,arrowinset=0.4]{->}(2.2,0.2)(2.2,-0.6)
\psline[linewidth=0.02cm,linestyle=dashed,dash=0.16cm 0.16cm](2.0,-0.2)(5.5,-0.2)
\psline[linewidth=0.02cm,linestyle=dashed,dash=0.16cm 0.16cm,arrowsize=0.05291667cm 2.0,arrowlength=1.4,arrowinset=0.4]{<-}(2.0,0.6)(2.0,-0.2)

}

\end{pspicture}
}

\end{center}\vspace{-0.05in}
\caption{ Block diagram of a polarization-based Marton code for a two-user noisy broadcast channel. } \label{fig:MartonCoding}
\end{figure*}
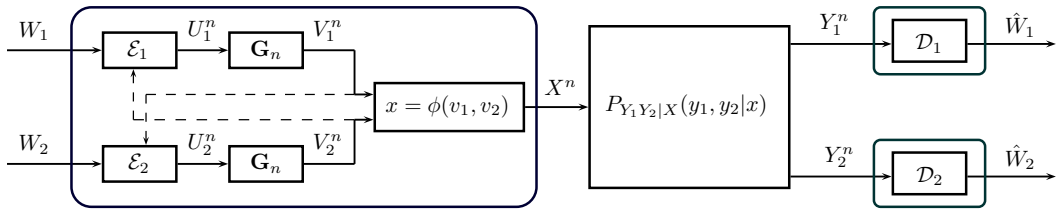
\subsection{Main Result}

\begin{theorem}[Polarization-Based Marton Code]\label{thm:MartonCode} Consider any two-user DM-BC with arbitrary input and output alphabets. There exist sequences of polar broadcast codes over $n$ channel uses which achieve the following rate region
\begin{align}
\mathfrak{R}(V_1, V_2, X, Y_1, Y_2) \triangleq \notag \\
\Bigl\{R_1, R_2~\Bigl|~ & R_1 \leq I(V_1; Y_1), \notag \\
& R_2 \leq I(V_2; Y_2) - I(V_1; V_2) \Bigl\}, \label{eqn:MartonAchievable1}
\end{align}
where random variables $V_1, V_2, X, Y_1, Y_2$ have the following listed properties:
\begin{itemize}
\item $V_1$ and $V_2$ are binary random variables.
\item $P_{Y_2|V_2}(y_2|v_2) \degBC P_{V_1|V_2}(v_1|v_2)$.
\item For a deterministic function $\phi: \{0,1\}^{2} \rightarrow \mathcal{X}$, the joint distribution of all random variables is given by
\begin{align}
& P_{V_1V_2XY_1Y_2}(v_1, v_2, x, y_1, y_2) = \notag \\
& ~~~~~~P_{V_1V_2}(v_1, v_2)\indicator{x = \phi(v_1, v_2)}P_{Y_1Y_2|X}(y_1, y_2 | x). \notag
\end{align}
\end{itemize}
For $0 < \beta < \frac{1}{2}$, the average error probability of this code sequence decays as $P_e^{(n)} = \mathcal{O}(2^{-n^{\beta}})$. The complexity of encoding and decoding is $\mathcal{O}(n \log n)$.
\end{theorem}
\vspace{0.05in}
\begin{remark} The listed property $P_{Y_2|V_2}(y_2|v_2) \degBC P_{V_1|V_2}(v_1|v_2)$ is necessary in the proof due to polarization-based codes requiring an \emph{alignment} of polarization indices. The property is a \emph{natural} restriction since it also implies that $I(Y_2; V_2) > I(V_1; V_2)$ so that $R_2 > 0$. However, certain joint distributions on random variables are not permitted using the analysis of polarization presented here. It is not clear whether a different approach obviates the need for an alignment of indices.
\end{remark}
\begin{remark} By symmetry, the rate tuple $(R_1, R_2) = (I(V_1; Y_1) - I(V_1; V_2), I(V_2, Y_2))$ is achievable with low-complexity codes under similar constraints on the joint distribution of $V_1, V_2, X, Y_1, Y_2$. The rate tuple is a corner point of the pentagonal rate region of Marton's inner bound given in~\eqref{eqn:MartonRateRegion}.
\end{remark}

\section{Proof of Theorem~\ref{thm:MartonCode}}\label{section:ProofOfMartonPolarCode}

The block diagram for polarization-based Marton coding is given in Figure~\ref{fig:MartonCoding}. Marton's strategy differs form Cover's superposition coding with the presence of two auxiliaries and the function $\phi(v_1, v_2)$ which forms the codeword symbol-by-symbol. The polar transform is applied to each $n$-length $i.i.d.$ sequence of auxiliary random variables.

\subsection{Polar Transform}\label{sec:PolarTransformMarton}

Consider the $i.i.d.$ sequence of random variables $(V_1^{j}, V_2^{j}, X^{j}, Y_1^{j}, Y_2^{j}) \sim P_{V_1V_2}(v_1,v_2)P_{X|V_1V_2}(x|v_1,v_2)P_{Y_1Y_2|X}(y_1, y_2|x)$ where the index $j \in [n]$. For the particular coding strategy analyzed in this section, $P_{X|V_1V_2}(x|v_1,v_2) = \indicator{x = \phi(v_1, v_2)}$. Let the $n$-length sequence of auxiliary variables $(V_1^{j}, V_2^{j})$ be organized into the random matrix
\vspace{0.1in}
\begin{align}
\matbold{\Omega} & \triangleq \left[\begin{array}{ccccc} V_1^{1} & V_1^{2} & V_1^{3} & \ldots & V_1^{n} \\ V_2^{1} & V_2^{2} & V_2^{3} & \ldots & V_2^{n} \end{array}\right]. \label{eqn:AMatMarton}
\end{align}
Applying the polar transform to $\matbold{\Omega}$ results in the random matrix $\matbold{U} \triangleq \matbold{\Omega}\matbold{G}_n$. Index the random variables of $\matbold{U}$ as follows:
\begin{align}
\matbold{U} & = \left[\begin{array}{ccccc} U_1^{1} & U_1^{2} & U_1^{3} & \ldots & U_1^{n} \\ U_2^{1} & U_2^{2} & U_2^{3} & \ldots & U_2^{n} \end{array}\right]. \label{eqn:UMatMarton}
\end{align}
The above definitions are consistent with the block diagram given in Figure~\ref{fig:MartonCoding} (and noting that $\matbold{G}_n = \matbold{G}_n^{-1}$). The polar transform extracts the randomness of $\matbold{\Omega}$. In the transformed domain, the joint distribution of the variables in $\matbold{U}$ is given by
\begin{align}
& P_{U_1^{n}U_2^{n}}\bigl(u_1^{n}, u_2^{n}\bigl) \triangleq P_{V_1^{n}V_2^{n}}\bigl(u_1^{n}\matbold{G}_n, u_2^{n}\matbold{G}_n\bigl). \label{eqn:U1U2JointMarton}
\end{align}
However, for polar coding purposes, the joint distribution is decomposed as follows,
\begin{align}
& P_{U_1^{n}U_2^{n}}\bigl(u_1^{n}, u_2^{n}\bigl) = P_{U_1^{n}}(u_1^{n})P_{U_2^{n}|U_1^{n}}\bigl(u_2^{n}\bigl|u_1^{n}\bigl) \notag \\
& = \prod_{j=1}^{n} P\bigl(u_1(j) \bigl| u_1^{1:j-1}\bigl)P\bigl(u_2(j) \bigl| u_2^{1:j-1}, u_1^{n}\bigl).
\label{eqn:U1U2DecomposedDistributionMarton}
\end{align}
The above conditional distributions may be computed efficiently using recursive protocols. The polarized random variables of $\matbold{U}$ do \emph{not} have an $i.i.d.$ distribution.

%The random variables $U_1^{n}$, $U_2^{n}$, $V_1^{n}$, $V_2^{n}$, $X^{n}$, $Y_1^{n}$, and $Y_2^{n}$ have a joint probability distribution given by,
%\begin{align}
%& P_{U_1^{n}U_2^{n}V_1^{n}V_2^{n}Y_1^{n}Y_2^{n}}(u_1^{n}, u_2^{n}, v_1^{n}, v_2^{n}, x_1^{n}, y_1^{n}, y_2^{n}) = \notag \\
%& ~~\indicator{u_1^{n} = v_1^{n}\matbold{G}_n}\indicator{u_2^{n} = v_2^{n}\matbold{G}_n} \prod_{i=1}^{n} P_{V_1V_2}\Bigl(v_1(i), v_2(i)\Bigl) \notag \\
%& ~~~\cdot \prod_{i=1}^{n} \indicator{x(i) = \phi(v_1(i), v_2(i))} \prod_{i=1}^{n} P_{Y_1Y_2|X}\Bigl(y_1(i), y_2(i) \Bigl| x(i)\Bigl). \notag
%\end{align}

\subsection{Effective Channel}

Marton's achievable strategy establishes virtual channels for the two receivers via the function $\phi(v_1, v_2)$. The virtual channel is given by
\begin{align}
P_{Y_1Y_2|V_1V_2}^{\phi}\Bigl(y_1, y_2 \Bigl| v_1, v_2\Bigl) & \triangleq P_{Y_1Y_2|X}\Bigl(y_1, y_2 \Bigl| \phi\bigl(v_1, v_2\bigl)\Bigl). \notag
\end{align}
Due to the memoryless property of the DM-BC, the effective channel between auxiliaries and channel outputs is given by
\begin{align}
& P^{\phi}_{Y_1^{n}Y_2^{n}|V_1^{n}V_2^{n}}\Bigl(y_1^{n}, y_2^{n} \Bigl| v_1^{n}, v_2^{n}\Bigl) \triangleq \notag \\
& ~~~~~~~~~~~~ \prod_{i=1}^{n}P_{Y_1Y_2|X}\Bigl(y_1(i), y_2(i) \Bigl| \phi\bigl(v_1(i), v_2(i)\bigl)\Bigl). \notag
\end{align}
The polarization-based Marton code establishes a different effective channel between polar-transformed auxiliaries and the channel outputs. The effective polarized channel is
\begin{align}
& P^{\phi}_{Y_1^{n}Y_2^{n}|U_1^{n}U_2^{n}}\Bigl(y_1^{n}, y_2^{n} \Bigl| u_1^{n}, u_2^{n}\Bigl) \triangleq \notag \\ & ~~~~~~~~~~~~~~ P^{\phi}_{Y_1^{n}Y_2^{n}|V_1^{n}V_2^{n}}\Bigl(y_1^{n}, y_2^{n} \Bigl| u_1^{n}\matbold{G}_n, u_2^{n}\matbold{G}_n\Bigl). \label{eqn:EffectivePolarizedChannelMarton}
\end{align}
%The key distribution of random variables utilized for analysis of the block decoding error probability is given by
%\begin{align}
%& P_{U_1^{n}U_2^{n}Y_1^{n}Y_2^{n}}\Bigl(u_1^{n}, u_2^{n}, y_1^{n}, y_2^{n}\Bigl) = \notag \\
%& ~~~~ P_{U_1^{n}U_2^{n}}\Bigl(u_1^{n}, u_2^{n}\Bigl) P^{\phi}_{Y_1^{n}Y_2^{n}|U_1^{n}U_2^{n}}\Bigl(y_1^{n}, y_2^{n} \Bigl| %u_1^{n}, u_2^{n}\Bigl). \label{eqn:DistributionOfInterestMarton}
%\end{align}
\subsection{Polarization Theorems Revisited}
\vspace{0.1in}
\begin{figure*}[t]
\begin{center}
%\psset{unit=0.50mm}
%\begin{pspicture}(-20,-8)(150,48)

%\small

% Generated with LaTeXDraw 2.0.0
% Mon Jan 14 14:21:29 CET 2013
% \usepackage[usenames,dvipsnames]{pstricks}
% \usepackage{epsfig}
% \usepackage{pst-grad} % For gradients
% \usepackage{pst-plot} % For axes
\scalebox{0.692} % Change this value to rescale the drawing.
{
\begin{pspicture}(0,-4.8039064)(13.038438,4.8039064)
\definecolor{color3184b}{rgb}{0.8431372549019608,0.403921568627451,0.403921568627451}
\definecolor{color3184}{rgb}{0.6,0.0,0.0}
\definecolor{color3217}{rgb}{0.0,0.2,0.2}
\definecolor{color3221b}{rgb}{0.0,0.4,0.0}
\definecolor{color3230b}{rgb}{0.6,0.6,0.6}
\definecolor{color308b}{rgb}{0.8,0.8,0.8}
\psframe[linewidth=0.02,dimen=outer,fillstyle=solid,fillcolor=black](1.6,1.9539063)(0.8,1.1539062)
\psframe[linewidth=0.06,linecolor=color3184,dimen=outer,fillstyle=solid,fillcolor=color3184b](11.21,2.7339063)(10.39,-2.4660938)
\psframe[linewidth=0.06,linecolor=color3184,dimen=outer,fillstyle=solid,fillcolor=color3184b](8.81,2.7339063)(7.99,-2.4660938)
\psframe[linewidth=0.06,linecolor=color3184,dimen=outer,fillstyle=solid,fillcolor=color3184b](7.21,2.7339063)(6.39,-2.4660938)
\psframe[linewidth=0.06,linecolor=color3184,dimen=outer,fillstyle=solid,fillcolor=color3184b](5.61,2.7339063)(4.79,-2.4660938)
\psline[linewidth=0.035cm,linestyle=dashed,dash=0.16cm 0.16cm,arrowsize=0.05291667cm 2.0,arrowlength=1.4,arrowinset=0.4]{<->}(0.02,2.3339062)(12.82,2.3339062)
\usefont{T1}{ptm}{m}{n}
\rput(12.908125,2.6239061){$n$}
\psframe[linewidth=0.02,dimen=outer,fillstyle=solid](0.8,1.9539063)(0.0,1.1539062)
\psframe[linewidth=0.02,dimen=outer,fillstyle=solid,fillcolor=black](3.2,1.9539063)(2.4,1.1539062)
\psframe[linewidth=0.02,dimen=outer,fillstyle=solid](2.4,1.9539063)(1.6,1.1539062)
\psframe[linewidth=0.02,dimen=outer,fillstyle=solid](4.0,1.9539063)(3.2,1.1539062)
\psframe[linewidth=0.02,dimen=outer,fillstyle=solid,fillcolor=color308b](4.8,1.9539063)(4.0,1.1539062)
\psframe[linewidth=0.02,dimen=outer,fillstyle=solid](5.6,1.9539063)(4.8,1.1539062)
\psframe[linewidth=0.02,dimen=outer,fillstyle=solid,fillcolor=color308b](6.4,1.9539063)(5.6,1.1539062)
\psframe[linewidth=0.02,dimen=outer,fillstyle=solid](8.0,1.9539063)(7.2,1.1539062)
\psframe[linewidth=0.02,dimen=outer,fillstyle=solid](9.6,1.9539063)(8.8,1.1539062)
\psframe[linewidth=0.02,dimen=outer,fillstyle=solid,fillcolor=black](10.4,1.9539063)(9.6,1.1539062)
\psframe[linewidth=0.02,dimen=outer,fillstyle=solid](11.2,1.9539063)(10.4,1.1539062)
\psframe[linewidth=0.02,dimen=outer,fillstyle=solid](12.0,1.9539063)(11.2,1.1539062)
\psframe[linewidth=0.02,dimen=outer,fillstyle=solid,fillcolor=black](12.8,1.9539063)(12.0,1.1539062)
\psframe[linewidth=0.02,dimen=outer,fillstyle=solid,fillcolor=black](1.6,-0.8460938)(0.8,-1.6460937)
\psframe[linewidth=0.02,dimen=outer,fillstyle=solid](0.8,-0.8460938)(0.0,-1.6460937)
\psframe[linewidth=0.02,dimen=outer,fillstyle=solid,fillcolor=black](3.2,-0.8460938)(2.4,-1.6460937)
\psframe[linewidth=0.02,dimen=outer,fillstyle=solid](2.4,-0.8460938)(1.6,-1.6460937)
\psframe[linewidth=0.02,dimen=outer,fillstyle=solid,fillcolor=black](6.4,-0.8460938)(5.6,-1.6460937)
\psframe[linewidth=0.02,dimen=outer,fillstyle=solid,fillcolor=black](7.2,-0.8460938)(6.4,-1.6460937)
\psframe[linewidth=0.02,dimen=outer,fillstyle=solid](8.0,-0.8460938)(7.2,-1.6460937)
\psframe[linewidth=0.02,dimen=outer,fillstyle=solid,fillcolor=black](8.8,-0.8460938)(8.0,-1.6460937)
\psframe[linewidth=0.02,dimen=outer,fillstyle=solid](9.6,-0.8460938)(8.8,-1.6460937)
\psframe[linewidth=0.02,dimen=outer,fillstyle=solid,fillcolor=black](10.4,-0.8460938)(9.6,-1.6460937)
\psframe[linewidth=0.02,dimen=outer,fillstyle=solid](12.0,-0.8460938)(11.2,-1.6460937)
\psframe[linewidth=0.02,dimen=outer,fillstyle=solid,fillcolor=black](12.8,-0.8460938)(12.0,-1.6460937)
\psline[linewidth=0.04cm,linecolor=color3217,fillcolor=black,dotsize=0.07055555cm 2.0]{-*}(9.22,3.5339062)(9.22,1.5339062)
\usefont{T1}{ptm}{m}{n}
\rput(9.446094,3.8239062){$Z\left(U_2(j) \Bigl| U_2^{1:j-1}, V_1^{n}\right) \geq 1 - \delta_n$}
\psline[linewidth=0.04cm,linecolor=color3217,fillcolor=color3221b,dotsize=0.07055555cm 2.0]{-*}(4.42,4.333906)(4.42,1.5339062)
\usefont{T1}{ptm}{m}{n}
\rput(4.713281,4.623906){$\delta_n < Z\left(U_2(j) \Bigl| U_2^{1:j-1}, V_1^{n}\right) < 1 - \delta_n$}
\usefont{T1}{ptm}{m}{n}
\rput(1.1960938,3.8239062){$Z\left(U_2(j) \Bigl| U_2^{1:j-1}, V_1^{n}\right) \leq \delta_n$}
\psframe[linewidth=0.02,dimen=outer,fillstyle=solid](7.2,1.9539063)(6.4,1.1539062)
\psframe[linewidth=0.02,dimen=outer,fillstyle=solid](8.8,1.9539063)(8.0,1.1539062)
\psframe[linewidth=0.02,dimen=outer,fillstyle=solid,fillcolor=color3230b](4.0,-0.8460938)(3.2,-1.6460937)
\psline[linewidth=0.035cm,linestyle=dashed,dash=0.16cm 0.16cm,arrowsize=0.05291667cm 2.0,arrowlength=1.4,arrowinset=0.4]{<->}(0.02,-0.46609375)(12.82,-0.46609375)
\usefont{T1}{ptm}{m}{n}
\rput(12.908125,-0.17609376){$n$}
\psframe[linewidth=0.02,dimen=outer,fillstyle=solid,fillcolor=black](5.6,-0.8460938)(4.8,-1.6460937)
\psframe[linewidth=0.02,dimen=outer,fillstyle=solid,fillcolor=black](11.2,-0.8460938)(10.4,-1.6460937)
\psline[linewidth=0.04cm,linecolor=color3217,fillcolor=color3221b,dotsize=0.07055555cm 2.0]{-*}(1.22,3.5339062)(1.22,1.5339062)
\psline[linewidth=0.04cm,linecolor=color3217,fillcolor=color3221b,dotsize=0.07055555cm 2.0]{*-}(1.22,-1.2660937)(1.22,-3.2660937)
\psline[linewidth=0.04cm,linecolor=color3217,fillcolor=black,dotsize=0.07055555cm 2.0]{*-}(11.62,-1.2660937)(11.62,-3.2660937)
\psline[linewidth=0.04cm,linecolor=color3217,fillcolor=color3221b,dotsize=0.07055555cm 2.0]{*-}(3.62,-1.2660937)(3.62,-4.066094)
\usefont{T1}{ptm}{m}{n}
\rput(3.9232812,-4.5760937){$\delta_n < Z\left(U_2(j) \Bigl| U_2^{1:j-1}, Y_2^{n}\right) < 1 - \delta_n$}
\usefont{T1}{ptm}{m}{n}
\rput(1.2060938,-3.7760937){$Z\left(U_2(j) \Bigl| U_2^{1:j-1}, Y_2^{n}\right) \leq \delta_n$}
\usefont{T1}{ptm}{m}{n}
\rput(11.456094,-3.7760937){$Z\left(U_2(j) \Bigl| U_2^{1:j-1}, Y_2^{n}\right) \geq 1 - \delta_n$}
\psframe[linewidth=0.02,dimen=outer,fillstyle=solid,fillcolor=color3230b](4.8,-0.8460938)(4.0,-1.6460937)
\end{pspicture}
}

\end{center}\vspace{-0.05in}
\caption{The alignment of polarization indices for Marton coding over noisy broadcast channels with respect to the second receiver. The message set $\mathcal{M}^{(n)}_{2}$ is highlighted by the vertical red rectangles. At finite code length $n$, exact alignment is not possible due to partially-polarized indices pictured in gray.} \label{fig:PolarTransformBitIndicesMarton}
\end{figure*}
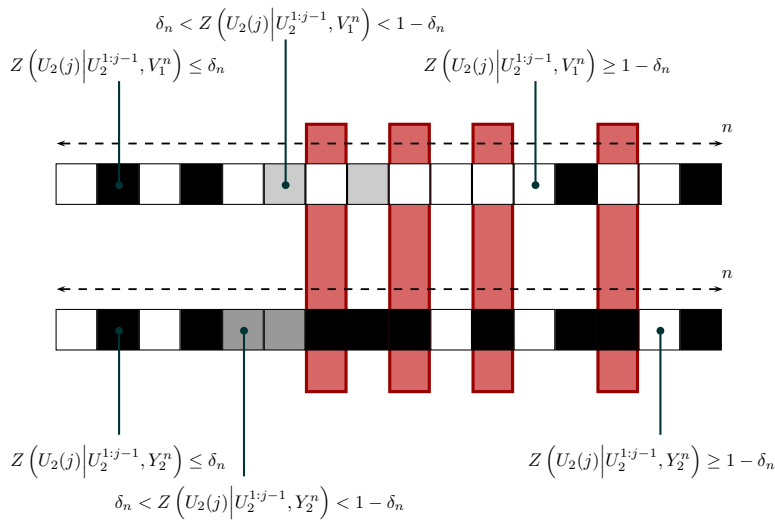

\begin{definition}[Polarization Sets for Marton Coding]\label{def:PolarizationSetsMarton} Let $V_1^{n}, V_2^{n}, X^{n}, Y_1^{n}, Y_2^{n}$ be the sequence of random variables as introduced in Section~\ref{sec:PolarTransformMarton}. In addition, let $U_1^{n} = V_1^{n}\matbold{G}_n$ and $U_2^{n} = V_2^{n}\matbold{G}_n$. Let $\delta_n = 2^{-n^{\beta}}$ for $0 < \beta < \frac{1}{2}$. The following \emph{polarization sets} are defined:
\begin{align}
\mathcal{H}^{(n)}_{V_1} & \triangleq \Bigl\{j \in [n]: Z\left(U_1(j) \Bigl| U_1^{1:j-1}\right) \geq 1 - \delta_n \Bigl\}, \notag \\
\mathcal{L}^{(n)}_{V_1|Y_1} & \triangleq \Bigl\{j \in [n]: Z\left(U_1(j) \Bigl| U_1^{1:j-1}, Y_1^{n}\right) \leq \delta_n \Bigl\}, \notag \\
\mathcal{H}^{(n)}_{V_2|V_1} & \triangleq \Bigl\{j \in [n]: Z\left(U_2(j) \Bigl| U_2^{1:j-1}, V_1^{n}\right) \geq 1 - \delta_n \Bigl\}, \notag \\
\mathcal{L}^{(n)}_{V_2|V_1} & \triangleq \Bigl\{j \in [n]: Z\left(U_2(j) \Bigl| U_2^{1:j-1}, V_1^{n}\right) \leq \delta_n \Bigl\}, \notag \\
\mathcal{H}^{(n)}_{V_2|Y_2} & \triangleq \Bigl\{j \in [n]: Z\left(U_2(j) \Bigl| U_2^{1:j-1}, Y_2^{n}\right) \geq 1 - \delta_n \Bigl\}, \notag \\
\mathcal{L}^{(n)}_{V_2|Y_2} & \triangleq \Bigl\{j \in [n]: Z\left(U_2(j) \Bigl| U_2^{1:j-1}, Y_2^{n}\right) \leq \delta_n \Bigl\}. \notag
\end{align}
\end{definition}

\begin{definition}[Message Sets for Marton Coding]\label{def:MessageSetsMarton} In terms of the polarization sets given in Definition~\ref{def:PolarizationSetsMarton}, the following \emph{message sets} are defined:
\begin{align}
\mathcal{M}^{(n)}_{1} & \triangleq \mathcal{H}^{(n)}_{V_1} \cap \mathcal{L}^{(n)}_{V_1|Y_1}, \label{eqn:MessageIndicesMartonM1} \\
\mathcal{M}^{(n)}_2 & \triangleq \mathcal{H}^{(n)}_{V_2|V_1} \cap \mathcal{L}^{(n)}_{V_2|Y_2}. \label{eqn:MessageIndicesMartonM2}
\end{align}
\end{definition}

\vspace{0.05in}
\begin{proposition}[Polarization]\label{thm:RateOfPolarizationMarton} Consider the polarization sets given in Definition~\ref{def:PolarizationSetsMarton} and the message sets given in Definition~\ref{def:MessageSetsMarton} with parameter $\delta_n = 2^{-n^{\beta}}$ for $0 < \beta < \frac{1}{2}$. Fix a constant $\tau > 0$. Then there exists an $N_o = N_o(\beta, \tau)$ such that
\begin{align}
& \frac{1}{n}\left|\mathcal{M}^{(n)}_{1}\right| \geq \Bigl(H(V_1) - H(V_1|Y_1)\Bigl) - \tau, \label{eqn:MessageSetAchievesCapacityMartonM1} \\
& \frac{1}{n}\left|\mathcal{M}^{(n)}_{2}\right| \geq \Bigl(H(V_2|V_1) - H(V_2|Y_2)\Bigl) - \tau, \label{eqn:MessageSetAchievesCapacityMartonM2}
\end{align}
for all $n > N_o$.
\end{proposition}
\vspace{0.1in}
\begin{lemma}\label{lemma:NestedSetsMartonCoding} Consider the polarization sets defined in Proposition~\ref{thm:RateOfPolarizationMarton}. If the property $P_{Y_2|V_2}(y_2|v_2) \degBC P_{V_1|V_2}(v_1|v_2)$ holds for conditional distributions $P_{Y_2|V_2}(y_2|v_2)$ and $P_{V_1|V_2}(v_1|v_2)$, then $I(V_2;Y_2) > I(V_1;V_2)$ and the Bhattacharyya parameters
\begin{align}
Z\left(U_2(j) \Bigl| U_2^{1:j-1}, Y_2^{n}\right) & \leq Z\left(U_2(j) \Bigl| U_2^{1:j-1}, V_1^{n}\right) \notag
\end{align}
for all $j \in [n]$. As a result,
\begin{align}
\mathcal{L}^{(n)}_{V_2|V_1} \subseteq \mathcal{L}^{(n)}_{V_2|Y_2},  \notag \\
\mathcal{H}^{(n)}_{V_2|Y_2} \subseteq \mathcal{H}^{(n)}_{V_2|V_1}. \notag
\end{align}
\end{lemma}
\begin{proof} The proof follows from Lemma~\ref{lemma:DegradationBhatt} and repeated application of Lemma~\ref{lemma:SuccessiveDegradationBhatt} in Appendix~\ref{sec:PolarCodingLemmas}.
\end{proof}

\begin{remark}
The alignment of polarization indices characterized by Lemma~\ref{lemma:NestedSetsMartonCoding} is diagrammed in Figure~\ref{fig:PolarTransformBitIndicesMarton}. The alignment ensures the existence of polarization indices in the set $\mathcal{M}^{(n)}_{2}$ for the message $W_2$ to have a positive rate $R_2 > 0$. The indices in $\mathcal{M}^{(n)}_{2}$ represent those bits freely set at the broadcast encoder and simultaneously those bits that may be decoded by $\mathcal{D}_2$ given its observations.
\end{remark}

\subsection{Partially-Polarized Indices}\label{sec:PartialPolarIndices}

As shown in Figure~\ref{fig:PolarTransformBitIndicesMarton}, for the Marton coding scheme, exact alignment of polarization indices is not possible. However, the alignment holds for all but $o(n)$ indices. The sets of partially-polarized indices shown in Figure~\ref{fig:PolarTransformBitIndicesMarton} are defined as follows.
\begin{definition}[Sets of Partially-Polarized Indices]
\begin{align}
\Delta_1 & \triangleq [n] ~ \backslash ~ \bigl( \mathcal{H}^{(n)}_{V_2|V_1} \cup \mathcal{L}^{(n)}_{V_2|V_1} \bigl), \label{eqn:MartonDelta1} \\
\Delta_2 & \triangleq [n] ~ \backslash ~ \bigl( \mathcal{H}^{(n)}_{V_2|Y_2} \cup \mathcal{L}^{(n)}_{V_2|Y_2} \bigl). \label{eqn:MartonDelta2}
\end{align}
\end{definition}
As implied by Ar\i kan's polarization theorems, the number of partially-polarized indices is negligible asymptotically as $n \rightarrow \infty$. For an arbitrarily small $\eta > 0$,
\begin{align}
\frac{\bigl| \Delta_1 \cup \Delta_2 \bigl|}{n} \leq \eta, \label{eqn:EtaPartialPolarizationMarton}
%\frac{\bigl| \Delta_2 \bigl|}{n} \leq \eta, \label{eqn:EtaPartialPolarizationMarton2}
\end{align}
for all $n$ sufficiently large enough. As will be discussed, providing these $o(n)$ bits as ``genie-given'' bits to the decoders results in a rate penalty; however, the rate penalty is negligible for sufficiently large code lengths.

\subsection{Broadcast Encoding Based on Polarization}
\label{sec:EncodingDecodingMarton}

As diagrammed in Figure~\ref{fig:MartonCoding}, the broadcast encoder must map two independent messages $(W_1, W_2)$ uniformly distributed over $[2^{nR_1}] \times [2^{nR_2}]$ to a codeword $x^{n} \in \mathcal{X}^{n}$ in such a way that the decoding at each separate receiver is successful. The achievable rates for a particular block length $n$ are
\begin{align}
R_1 & = \frac{1}{n}\left|\mathcal{M}^{(n)}_{1}\right|, \notag \\
R_2 & = \frac{1}{n}\left|\mathcal{M}^{(n)}_{2}\right|. \notag
\end{align}

To construct a codeword, the encoder first produces two binary sequences $u_1^{n} \in \{0,1\}^{n}$ and $u_2^{n} \in \{0,1\}^{n}$. To determine $u_1(j)$ for $j \in \mathcal{M}^{(n)}_{1}$, the bit is selected as a uniformly distributed message bit intended for the first receiver. To determine $u_2(j)$ for $j \in \mathcal{M}^{(n)}_{2}$, the bit is selected as a uniformly distributed message bit intended for the second receiver. The remaining \emph{non-message} indices of $u_1^{n}$ and $u_2^{n}$ are decided \emph{randomly} according to the proper statistics as will be described in this section. The transmitted codeword is formed symbol-by-symbol via the $\phi$ function,
\begin{align}
\forall j \in [n]: x(j) = \phi\bigl(v_1(j), v_2(j)\bigl) \notag
\end{align}
where $v_1^{n} = u_1^{n}\matbold{G}_n$ and $v_2^{n} = u_2^{n}\matbold{G}_n$. A valid codeword sequence is always guaranteed to be formed unlike in the case of coding for deterministic broadcast channels.

%\subsubsection{Deterministic Mapping} Consider the following deterministic boolean functions indexed by $j \in [n]$:
%\begin{align}
%\psi_1^{(j)}: \{0,1\}^{j-1} \rightarrow \{0,1\}.  \label{eqn:BooleanMapsDetMarton1} \\
%\psi_2^{(j)}: \{0,1\}^{n+j-1} \rightarrow \{0,1\}, \label{eqn:BooleanMapsDetMarton2}
%\end{align}
%As an example, consider the deterministic boolean functions based on the \emph{maximum a posteriori} polar coding rule.
%\begin{align}
%& \psi^{(j)}_{1}\left(u_1^{1:j-1}\right) \triangleq \notag \\
%& ~~ \argmax_{u \in \{0,1\}} \Bigl\{ \mathbb{P}\left(U_1(j) = u \Bigl| U_1^{1:j-1} = u_1^{1:j-1} \right)\Bigl\}. \label{eqn:DetMapSpecificMarton1} \\
%& \psi^{(j)}_{2}\left(u_2^{1:j-1}, v_1^{n}\right) \triangleq \notag \\
%& ~~ \argmax_{u \in \{0,1\}} \Bigl\{ \mathbb{P}\left(U_2(j) = u \Bigl| U_2^{1:j-1} = u_2^{1:j-1}, V_1^{n} = v_1^{n} \right)\Bigl\}. \label{eqn:DetMapSpecificMarton2}
%\end{align}
\vspace{0.15in}
\subsubsection{Random Mapping} To fill in the non-message indices, we define the following random mappings. Consider the following class of random boolean functions where $j \in [n]$:
\begin{align}
\Psi_1^{(j)}: \{0,1\}^{j-1} \rightarrow \{0,1\}, \label{eqn:BooleanMapsRANDMarton1} \\
\Psi_2^{(j)}: \{0,1\}^{n+j-1} \rightarrow \{0,1\}, \label{eqn:BooleanMapsRANDMarton2} \\
\Gamma: [n] \rightarrow \{0,1\}. \label{eqn:BooleanMapsRANDMartonBernoulli}
\end{align}
More concretely, we consider the following specific random boolean functions based on the statistics derived from polarization methods:
\begin{align}
\Psi_1^{(j)}\left(u_1^{1:j-1} \right) & \triangleq \begin{cases} 0, & \mbox{\emph{w.p.} }~ \lambda_{0}\left(u_1^{1:j-1} \right), \\ 1, & \mbox{\emph{w.p.} }~1-\lambda_{0}\left(u_1^{1:j-1} \right), \end{cases} \label{eqn:RANDMapSpecificMarton1} \\
\Psi_2^{(j)}\left(u_2^{1:j-1}, v_1^{n} \right) & \triangleq \begin{cases} 0, & \mbox{\emph{w.p.} }~ \lambda_{0}\left(u_2^{1:j-1}, v_1^{n}\right), \\ 1, & \mbox{\emph{w.p.} }~1-\lambda_{0}\left(u_2^{1:j-1}, v_1^{n}\right) \end{cases} \label{eqn:RANDMapSpecificMarton2} \\
\Gamma(j) & \triangleq \begin{cases} 0, & \mbox{\emph{w.p.} }~ \frac{1}{2}, \\ 1, & \mbox{\emph{w.p.} }~\frac{1}{2}, \end{cases} \label{eqn:RANDMapSpecificMartonBernoulli}
\end{align}
where
\begin{align}
\lambda_{0}\left(u_1^{1:j-1} \right) & \triangleq \mathbb{P}\left(U_1(j) = 0 \Bigl| U_1^{1:j-1} = u_1^{1:j-1}\right). \notag \\
\lambda_{0}\left(u_2^{1:j-1}, v_1^{n} \right) & \triangleq \notag \\
& \!\!\!\! \mathbb{P}\left(U_2(j) = 0 \Bigl| U_2^{1:j-1} = u_2^{1:j-1}, V_1^{n} = v_1^{n}\right). \notag
\end{align}
For a fixed $j \in [n]$, the random boolean functions $\Psi_1^{(j)}$, $\Psi_2^{(j)}$ may be thought of as a vector of independent Bernoulli random variables indexed by the input to the function. Each Bernoulli random variable of the vector is zero or one with a fixed well-defined probability that is efficiently computable. The random boolean function $\Gamma$ may be thought of as an $n$-length vector of Bernoulli$(\frac{1}{2})$ random variables.

\vspace{0.15in}
\subsubsection{Encoding Protocol} The broadcast encoder constructs the sequence $u_1^{n}$ bit-by-bit successively,
\begin{align}
u_1(j) & = \begin{cases} W_1~\mbox{message bit}, & \mbox{\emph{if}}~ j \in \mathcal{M}^{(n)}_{1}, \\ \Psi_1^{(j)}\bigl(u_1^{1:j-1} \bigl), & \mbox{\emph{otherwise.}}~\end{cases} \label{eqn:EncoderMartonRule1}
\end{align}
The encoder then computes the sequence $v_1^{n} = u_1^{n}\matbold{G}_n$. To generate $v_2^{n}$, the encoder constructs the sequence $u_2^{n}$ (given $v_1^{n}$) as follows,
\begin{align}
u_2(j) & = \begin{cases} W_2~\mbox{message bit}, & \mbox{\emph{if}}~ j \in \mathcal{M}^{(n)}_{2}, \\ \Gamma(j), & \mbox{\emph{if}}~ j \in \mathcal{H}^{(n)}_{V_2|V_1} ~\backslash~ \mathcal{M}^{(n)}_{2}, \\ \Psi_2^{(j)}\bigl(u_2^{1:j-1}, v_1^{n} \bigl), & \mbox{\emph{otherwise.}}~\end{cases} \label{eqn:EncoderMartonRule2}
\end{align}
Then the sequence $v_2^{n} = u_2^{n}\matbold{G}_n$. The randomness in the above encoding protocol over non-message indices ensures that the pair of sequences $(u_1^{n}, u_2^{n})$ has the correct statistics as if drawn from the joint distribution of $(U_1^{n}, U_2^{n})$. In the last step, the encoder transmits a codeword $x^{n}$ formed symbol-by-symbol: $x(j) = \phi\bigl(v_1(j), v_2(j)\bigl)$ for all $j \in [n]$. For $j \in \Delta_2$, where $\Delta_2$ is the set of partially-polarized indices defined in~\eqref{eqn:MartonDelta2}, the encoder records the realization of $u_2(j)$. These indices will be provided to the second receiver's decoder $\mathcal{D}_2$ as ``genie-given'' bits.

\subsection{Broadcast Decoding Based on Polarization}
\label{sec:DecodingMarton}

\subsubsection{Decoding At First Receiver}

Decoder $\mathcal{D}_1$ decodes the binary sequence $\hat{u}_1^{n}$ using its observations $y_1^{n}$. The message $W_1$ is located at the indices $j \in \mathcal{M}^{(n)}_{1}$ in the sequence $\hat{u}_1^{n}$. More precisely, we define the following deterministic polar decoding function for the $j$-th bit:
\begin{align}
& \xi^{(j)}_{u_1}\left(u_1^{1:j-1}, y_1^{n}\right) \triangleq \argmax_{u \in \{0,1\}} \Bigl\{ \notag \\
& ~~\mathbb{P}\left(U_1(j) = u \Bigl| U_1^{1:j-1} = u_1^{1:j-1}, Y_1^{n} = y_1^{n} \right)\Bigl\}. \label{eqn:DecodingMarton1}
\end{align}
Decoder $\mathcal{D}_1$ reconstructs $\hat{u}_1^{n}$ bit-by-bit successively as follows using the identical random mapping $\Psi_1^{(j)}$ at the encoder:
\begin{align}
\hat{u}_1(j) & = \begin{cases} \xi^{(j)}_{u_1}\bigl(\hat{u}_1^{1:j-1}, y_1^{n}\bigl), & \mbox{\emph{if}}~ j \in \mathcal{M}^{(n)}_{1}, \\ \Psi_1^{(j)}\bigl(\hat{u}_1^{1:j-1} \bigl), & \mbox{\emph{otherwise.}}~\end{cases} \label{eqn:DecoderMartonRule1}
\end{align}
Given that all previous bits $\hat{u}_1^{1:j-1}$ have been decoded correctly, decoder $\mathcal{D}_1$ makes a mistake on the $j$-th bit $\hat{u}_1(j)$ only if $j \in \mathcal{M}^{(n)}_{1}$. For the remaining indices, the decoder produces the same bit produced at the encoder due to the shared random maps.

\vspace{0.1in}
\subsubsection{Decoding At Second Receiver}

The decoder $\mathcal{D}_2$ decodes the binary sequence $\hat{u}_2^{n}$ using observations $y_2^{n}$. The message $W_2$ is located at the indices $j \in \mathcal{M}^{(n)}_{2}$ of the sequence $\hat{u}_2^{n}$. Define the following deterministic polar decoding functions
\begin{align}
& \xi^{(j)}_{u_2}\left(u_2^{1:j-1}, y_2^{n}\right) \triangleq \notag \\
& ~~\argmax_{u \in \{0,1\}} \Bigl\{ \mathbb{P}\left(U_2(j) = u \Bigl| U_2^{1:j-1} = u_2^{1:j-1}, Y_2^{n} = y_2^{n} \right)\Bigl\}. \label{eqn:DecodingMarton2}
\end{align}
Decoder $\mathcal{D}_2$ reconstructs $\hat{u}_2^{n}$ bit-by-bit successively as follows using the \emph{identical} shared random mapping $\Gamma$ used at the encoder. Including all but $o(n)$ of the indices,
\begin{align}
\hat{u}_2(j) & = \begin{cases} \xi^{(j)}_{u_2}\left(\hat{u}_2^{1:j-1}, y_2^{n}\right), & \mbox{\emph{if}}~ j \in \mathcal{L}^{(n)}_{V_2|Y_2}, \\ \Gamma(j), & \mbox{\emph{if}}~ j \in \mathcal{H}^{(n)}_{V_2|Y_2}.~\end{cases} \label{eqn:DecoderMartonRule2}
\end{align}
For those indices $j \in \Delta_2$ where $\Delta_2$ is the set of partially-polarized indices defined in~\eqref{eqn:MartonDelta2}, the decoder $\mathcal{D}_2$ is provided with ``genie-given'' bits from the encoder. Thus, all bits are decoded, and $\mathcal{D}_2$ only makes a successive cancelation error for those indices $j \in \mathcal{L}^{(n)}_{V_2|Y_2}$. Communicating the genie-given bits from the encoder to decoder results in a rate penalty. However, since the number of genie-given bits scales asymptotically as $o(n)$, the rate penalty can be made arbitrarily small.
\vspace{0.15in}
\begin{remark}
It is notable that decoder $\mathcal{D}_2$ reconstructs $\hat{u}_2^{n}$ using only the observations $y_2^{n}$. At the encoder, the sequence $u_2^{n}$ was generated with the realization of a sequence $v_1^{n}$ as given in~\eqref{eqn:EncoderMartonRule2}. However, decoder $\mathcal{D}_2$ does \emph{not} reconstruct the sequence $\hat{v}_1^{n}$. From this operational perspective, Marton's scheme differs crucially from Cover's superposition scheme because there does not exist the notion of a ``stronger'' receiver which reconstructs all the sequences decoded at the ``weaker'' receiver.
\end{remark}

\subsection{Total Variation Bound}

To analyze the average probability of error $P_e^{(n)}$, it is assumed that both the encoder and decoder share the \emph{randomized} mappings $\Psi_1^{(j)}$, $\Psi_2^{(j)}$, and $\Gamma$ (where $\Psi_2^{(j)}$ is not utilized at decoder $\mathcal{D}_2$). Define the following probability measure on the space of tuples of binary sequences.
\begin{align}
& Q\bigl(u_1^{n}, u_2^{n}\bigl) \triangleq Q\bigl(u_1^{n}\bigl)Q\bigl(u_1^{n}\bigl|u_2^{n}\bigl) \notag \\
& = \prod_{j=1}^{n} Q\bigl(u_1(j) \bigl| u_1^{1:j-1}\bigl) Q\bigl(u_2(j) \bigl| u_2^{1:j-1}, u_1^{n}\bigl), \label{eqn:QUJOINTMarton}
\end{align}
where the conditional probability measures are defined as
\begin{align}
& Q\Bigl(u_1(j) \Bigl| u_1^{1:j-1}\Bigl) \triangleq \notag \\
& ~~~~~~ \begin{cases} \frac{1}{2}, & \mbox{\emph{if}}~j \in \mathcal{M}_1^{(n)}, \\ P\left( u_1(j) \Bigl| u_1^{1:j-1}\right), & \mbox{\emph{otherwise.}} \end{cases} \notag \\
& Q\Bigl(u_2(j) \Bigl| u_2^{1:j-1}, u_1^{n}\Bigl) \triangleq \notag \\
& ~~~~~~ \begin{cases} \frac{1}{2}, & \mbox{\emph{if}}~j \in \mathcal{H}^{(n)}_{V_2|V_1}, \\ P\left( u_2(j) \Bigl| u_2^{1:j-1}, u_1^{n}\right), & \mbox{\emph{otherwise.}} \end{cases} \notag
\end{align}
The probability measure $Q$ defined in~\eqref{eqn:QUJOINTMarton} is a perturbation of the joint probability measure $P_{U_1^{n}U_2^{n}}(u_1^{n}, u_2^{n})$ in~\eqref{eqn:U1U2DecomposedDistributionMarton}. The only difference in definition between $P$ and $Q$ is due to those indices in message sets $\mathcal{M}_1^{(n)}$ and $\mathcal{H}^{(n)}_{V_2|V_1}$ (note: $\mathcal{M}_2^{(n)} \subseteq \mathcal{H}^{(n)}_{V_2|V_1}$). The following lemma provides a bound on the total variation distance between $P$ and $Q$. The lemma establishes the fact that inserting uniformly distributed message bits in the proper indices $\mathcal{M}_1^{(n)}$ and $\mathcal{M}_2^{(n)}$ (or the entire set $\mathcal{H}^{(n)}_{V_2|V_1}$) at the encoder \emph{does not} perturb the statistics of the $n$-length random variables too much.
\vspace{0.1in}
\begin{lemma}(\emph{Total Variation Bound})\label{lemma:TVBoundMarton}
Let probability measures $P$ and $Q$ be defined as in~\eqref{eqn:U1U2DecomposedDistributionMarton} and~\eqref{eqn:QUJOINTMarton} respectively. Let $0 < \beta < 1$. For sufficiently large $n$, the total variation distance between $P$ and $Q$ is bounded as
\begin{align}
\sum_{\begin{subarray}{c} u_1^{n} \in \{0,1\}^{n} \\ u_2^{n} \in \{0,1\}^{n} \end{subarray}} \Bigl| P_{U_1^{n}U_2^{n}}\bigl( u_1^{n}, u_2^{n} \bigl) - Q\bigl( u_1^{n}, u_2^{n} \bigl) \Bigl| \leq 2^{-n^{\beta}}. \notag
\end{align}
\end{lemma}
\begin{proof} \emph{Omitted.} The proof follows via the chain rule for KL-divergence and is identical to the previous proofs of Lemma~\ref{lemma:TVBound} and Lemma~\ref{lemma:TVBoundSuperposition}. \end{proof}

\subsection{Error Sequences}

The decoding protocols for $\mathcal{D}_1$ and $\mathcal{D}_2$ were established in Section~\ref{sec:DecodingMarton}. To analyze the probability of error of successive cancelation (SC) decoding, consider the sequences $u_1^{n}$ and $u_2^{n}$ formed at the encoder, and the resulting observations $y_1^{n}$ and $y_2^{n}$ received by the decoders. The effective polarized channel $P^{\phi}_{Y_1^{n}Y_2^{n}|U_1^{n}U_2^{n}}\bigl(y_1^{n}, y_2^{n} \bigl| u_1^{n}, u_2^{n}\bigl)$ was defined in~\eqref{eqn:EffectivePolarizedChannelMarton} for a fixed $\phi$ function. It is convenient to group the sequences together and consider all tuples $(u_1^{n}, u_2^{n}, y_1^{n}, y_2^{n})$.

Decoder $\mathcal{D}_1$ makes an SC decoding error on the $j$-th bit for the following tuples:
\begin{align}
\mathcal{T}_{1}^{j} & \triangleq \Bigl\{\bigl(u_1^{n}, u_2^{n}, y_1^{n}, y_2^{n}\bigl): \notag \\
& ~~ P_{U_1^{j}\bigl|U_1^{1:j-1}Y_1^{n}}\bigl(u_1(j)\bigl|u_1^{1:j-1}, y_1^{n}\bigl) \leq \notag \\
& ~~ P_{U_1^{j}|U_1^{1:j-1}Y_1^{n}}\bigl(u_1(j) \oplus 1 \bigl|u_1^{1:j-1}, y_1^{n}\bigl)\Bigl\}.
\end{align}
The set $\mathcal{T}_{1}^{j}$ represents those tuples causing an error at $\mathcal{D}_1$ in the case $u_1(j)$ is inconsistent with respect to observations $y_1^{n}$ and the decoding rule. Similarly, decoder $\mathcal{D}_2$ makes an SC decoding error on the $j$-th bit for the following tuples:
\begin{align}
\mathcal{T}_2^{j} & \triangleq \Bigl\{\bigl(u_1^{n}, u_2^{n}, y_1^{n}, y_2^{n}\bigl):  \notag \\
& ~~ P_{U_2\bigl|U_2^{1:j-1}Y_2^{n}}\bigl(u_2\bigl|u_2^{1:j-1},y_2^{n}\bigl) \leq \notag \\
& ~~ P_{U_2\bigl|U_2^{1:j-1}Y_2^{n}}\bigl(u_2 \oplus 1 \bigl| u_2^{1:j-1}, y_2^{n}\bigl)\Bigl\}. \notag
\end{align}
The set $\mathcal{T}_2^{j}$ represents those tuples causing an error at $\mathcal{D}_2$ in the case $u_2(j)$ is inconsistent with respect to observations $y_2^{n}$ and the decoding rule. The set of tuples causing an error is
\begin{align}
\mathcal{T}_{1} & \triangleq \bigcup_{j \in \mathcal{M}_1^{(n)}} \mathcal{T}_{1}^{j}, \label{eqn:Total1ErrorTupleMarton} \\
\mathcal{T}_{2} & \triangleq \bigcup_{j \in \mathcal{L}^{(n)}_{V_2|V_1}} \mathcal{T}_{2}^{j}, \label{eqn:Total2ErrorTupleMarton} \\
\mathcal{T} & \triangleq \mathcal{T}_{1} \cup \mathcal{T}_{2}. \label{eqn:TotalErrorTupleSuperposition}
\end{align}
The goal is to show that the probability of choosing tuples of error sequences in the set $\mathcal{T}$ is small under the distribution induced by the broadcast code.

\subsection{Average Error Probability}

If the encoder and decoders share randomized maps $\Psi_1^{(j)}$, $\Psi_{2}^{(j)}$, and $\Gamma$, then the average probability of error is a random quantity determined as follows
\begin{align}
& P_e^{(n)}\left[\{\Psi_1^{(j)}, \Psi_2^{(j)}, \Gamma\}\right] = \notag \\
& ~~ \sum_{\{u_1^{n}, u_2^{n}, y_1^{n}, y_2^{n}\} \in \mathcal{T}} \Biggl[ P^{\phi}_{Y_1^{n}Y_2^{n}\bigl|U_1^{n}U_2^{n}}\bigl(y_1^{n}, y_2^{n}\bigl|u_1^{n}, u_2^{n}\bigl) \notag \\
& ~~ \cdot \frac{1}{2^{nR_1}} \prod_{j \in [n]: j \notin \mathcal{M}_1^{(n)}} \indicator{\Psi_1^{(j)}\left(u_1^{1:j-1}\right) = u_1(j)} \notag \\
& ~~ \cdot \frac{1}{2^{nR_2}} \prod_{j \in \mathcal{H}^{(n)}_{V_2|V_1} \backslash \mathcal{M}_2^{(n)}} \indicator{\Gamma(j) = u_2(j)} \notag \\
& ~~ \cdot \prod_{j \in [n]: j \notin \mathcal{H}^{(n)}_{V_2|V_1}} \indicator{\Psi_2^{(j)}\left(u_2^{1:j-1}, u_1^{n}\matbold{G}_n \right) = u_2(j)}\Biggl]. \notag
\end{align}
By averaging over the randomness in the encoders and decoders, the expected block error probability $P_e^{(n)}[\{\Psi_1^{(j)}, \Psi_2^{(j)}\}, \Gamma]$ is upper bounded in the following lemma.
\vspace{0.05in}
\begin{lemma}\label{theorem:ErrorProbMarton} Consider the polarization-based Marton code described in Section~\ref{sec:EncodingDecodingMarton} and Section~\ref{sec:DecodingMarton}. Let $R_1$ and $R_2$ be the broadcast rates selected according to the Bhattacharyya criterion given in Proposition~\ref{thm:RateOfPolarizationMarton}. Then for $0 < \beta < 1$ and sufficiently large $n$,
\begin{align}
\mathbb{E}_{\{\Psi_1^{(j)}, \Psi_2^{(j)}, \Gamma\}} \Bigl[ P_e^{(n)}[\{\Psi_1^{(j)}, \Psi_2^{(j)}, \Gamma\}] \Bigl] < 2^{-n^{\beta}}. \notag
\end{align}
\end{lemma}
\begin{proof} See Section~\ref{sec:AppendixMarton} of the Appendices.
\end{proof}
\vspace{0.1in}
If the average probability of block error decays to zero in expectation over the random maps $\{\Psi_1^{(j)}\}$, $\{\Psi_2^{(j)}\}$, and $\Gamma$, then there must exist at least one fixed set of maps for which $P_e^{(n)} \rightarrow 0$. Hence, polar codes for Marton's inner bound exist under suitable restrictions on distributions and they achieve reliable transmission according to the advertised rates (except for a small set of $o(n)$ polarization indices as is discussed next).

\subsection{Rate Penalty Due to Partial Polarization}

Lemma~\ref{theorem:ErrorProbMarton} is true assuming that decoder $\mathcal{D}_2$ obtains ``genie-given'' bits for the set of indices $\Delta_2$ defined in~\eqref{eqn:MartonDelta2}. The set $\Delta_2$ represents those indices that are partially-polarized and which cause a slight misalignment of polarization indices in the Marton scheme. Fortunately, the set $\Delta_2$ contains a vanishing fraction of indices: $\frac{1}{n}\bigl| \Delta_2 \bigl| \leq \eta$ for $\eta > 0$ arbitrarily small and $n$ sufficiently large. Therefore, a two-phase strategy suffices for sending the ``genie-given'' bits. In the first phase of communication, the encoder sends several $n$-length blocks while decoder $\mathcal{D}_2$ waits to decode. After accumulating several blocks of output sequences, the encoder transmits all the known bits in the set $\Delta_2$ for all the first-phase transmissions. The encoder and decoder can use any reliable point-to-point polar code with non-vanishing rate for communication. Having received the ``genie-aided'' bits in the second-phase, the second receiver then decodes all the first-phase blocks. The number of blocks sent in the first-phase is $\mathcal{O}(\frac{1}{\eta})$. The rate penalty is $\mathcal{O}(\eta)$ where $\eta$ can be made arbitrarily small. A similar argument was provided in~\cite{korada10} for designing polar codes for the Gelfand-Pinsker problem.

\vspace{-0.025in}
\section{Conclusion}

Coding for broadcast channels is fundamental to our understanding of communication systems. Broadcast codes based on polarization methods achieve rates on the capacity boundary for several classes of DM-BCs. In the case of $m$-user deterministic DM-BCs, polarization of random variables from the channel output provides the ability to extract uniformly random message bits while maintaining broadcast constraints at the encoder. As referenced in the literature, maintaining  multi-user constraints for the DM-BC is a difficult task for traditional belief propagation algorithms and LDPC codes.

\vspace{0.1in}
For two-user noisy DM-BCs, polar codes were designed based on Marton's coding strategy and Cover's superposition strategy. Constraints on auxiliary and input distributions were placed in both cases to ensure alignment of polarization indices in the multi-user setting. The asymptotic behavior of the average error probability was shown to be $P_e^{(n)} = \mathcal{O}(2^{-n^{\beta}})$ with an encoding and decoding complexity of $\mathcal{O}(n \log n)$. The next step is to supplement the theory with experimental evidence of the error-correcting capability of polar codes over simulated channels for finite code lengths. The results demonstrate that polar codes have a potential for use in several \emph{network communication} scenarios.

\appendices

\section{Polar Coding Lemmas}\label{sec:PolarCodingLemmas}

The following lemmas provide a basis for proving polar coding theorems. A subset of the lemmas were proven in different contexts, e.g., channel vs. source coding, and contain citations to references.
\begin{lemma}\label{lemma:ConditionalEntropyAndKLDistance}
Consider two random variables $X \in \{0,1\}$ and $Y \in \mathcal{Y}$ with joint distribution $P_{X,Y}(x,y)$. Let $Q(x|y) = \frac{1}{2}$ denote a uniform conditional distribution for $x \in \{0,1\}$ and $y \in \mathcal{Y}$. Then the following identity holds.
\begin{align}
D\left(P_{X|Y}(x|y) \Bigl\| Q(x|y)\right) & = 1 - H(X|Y). \label{eqn:EntropyAndKL}
\end{align}
\end{lemma}
\begin{proof} The identity follows from standard definitions of entropy and Kullback-Leibler distance.
\begin{align}
& H(X|Y) \notag \\
& = \sum_{y \in \mathcal{Y}} P_{Y}(y)\sum_{x \in \{0,1\}} P_{X|Y}(x|y) \log_2 \frac{1}{P_{X|Y}(x|y)} \notag \\
& = \sum_{y \in \mathcal{Y}} P_{Y}(y) \sum_{x \in \{0,1\}} P_{X|Y}(x|y) \log_2 \frac{1}{Q(x|y)} \notag \\
& ~~~ - \sum_{y \in \mathcal{Y}} P_{Y}(y) \sum_{x \in \{0,1\}} P_{X|Y}(x|y) \log_2 \frac{P_{X|Y}(x|y)}{Q(x|y)} \notag \\
& = \sum_{y \in \mathcal{Y}} P_{Y}(y) \left[ 1 - \sum_{x \in \{0,1\}} P_{X|Y}(x|y) \log_2 \frac{P_{X|Y}(x|y)}{Q(x|y)} \right] \notag \\
& = 1 - D\left(P_{X|Y}(x|y) \Bigl\| Q(x|y)\right). \notag
\end{align}
\end{proof}
%Define a modified likelihood function $\tilde{L}(v)$ and inverse likelihood function $\tilde{L}^{-1}(v)$ as %follows. For $v \in \mathcal{V}$ such that $P_{T|V}(1|v) = 0$, let $\tilde{L}(v) = 0$. Similarly for $v \in %\mathcal{V}$ such that $P_{T|V}(0|v) = 0$, let $\tilde{L}^{-1}(v) = 0$. For all other $v \in \mathcal{V}$, %define
\begin{lemma}[Estimating The Bhattacharyya Parameter]\label{lemma:BhattacharyyaMonteCarlo} Let $(T, V) \sim P_{T, V}(t,v)$ where $T \in \{0,1\}$ and $V \in \mathcal{V}$ where $\mathcal{V}$ is an arbitrary discrete alphabet. Define a likelihood function $L(v)$ and inverse likelihood function $L^{-1}(v)$ as follows.
\begin{align}
L(v) \triangleq \frac{P_{T|V}(0|v)}{P_{T|V}(1|v)}, ~~ L^{-1}(v) \triangleq \frac{P_{T|V}(1|v)}{P_{T|V}(0|v)} \notag
\end{align}
To account for degenerate cases in which $P_{T|V}(t|v) = 0$, define the following function,
\begin{subnumcases}{\varphi(t,v) \triangleq }
0 & $\mbox{\emph{if}}~\indicator{P_{T|V}(t|v) = 0}$ \notag \\
L(v) & $\mbox{\emph{if}}~\indicator{P_{T|V}(t|v) > 0}$~\mbox{and}~$\indicator{t=1}$ \notag \\
L^{-1}(v) & $\mbox{\emph{if}}~\indicator{P_{T|V}(t|v) > 0}$~\mbox{and}~$\indicator{t=0}$ \notag
\end{subnumcases}
In order to estimate $Z(T|V) \in [0,1]$, it is convenient to sample from $P_{TV}(t,v)$ and express $Z(T|V)$ as an expectation over random variables $T$ and $V$,
\begin{align}
%Z(T|V) & = \mathbb{E}_{T,V}\left[ \sqrt{ \tilde{L}^{-1}(V)\indicator{T = 0} + \tilde{L}(V)\indicator{T = 1} } %\right]. \label{eqn:BhattEstimate1}
Z(T|V) & = \mathbb{E}_{T,V} \sqrt{\varphi(T,V)}. \label{eqn:BhattEstimate1}
\end{align}
\end{lemma}
\begin{proof} The following forms of the Bhattacharyya parameter are equivalent.
\begin{align}
Z(T|V) & \triangleq 2 \sum_{v \in \mathcal{V}} P_{V}(v) \sqrt{P_{T|V}(0|v)P_{T|V}(1|v)} \notag \\
& = 2 \sum_{v \in \mathcal{V}} \sqrt{P_{TV}(0,v)P_{TV}(1,v)} \notag \\
& = \sum_{v \in \mathcal{V}} P_{V}(v) \sum_{t \in \{0,1\}} \sqrt{P_{T|V}(t|v)(1 - P_{T|V}(t|v))} \notag \\
& = \sum_{t \in \{0,1\}} \sum_{\begin{subarray}{c}
        v: P_{T|V}(t|v) > 0 \\ v \in \mathcal{V}
      \end{subarray}} P_{TV}(t,v) \sqrt{ \frac{1 - P_{T|V}(t|v)}{P_{T|V}(t|v)} } \notag \\
%& = \mathbb{E}_{T,V}\left[ \sqrt{ \frac{1 - P_{T|V}(T|V)}{P_{T|V}(T|V)} } \right] \label{eqn:BhattEstimate0} \\
%& = \mathbb{E}_{T,V}\left[ \sqrt{ \tilde{L}^{-1}(V)\indicator{T = 0} + \tilde{L}(V)\indicator{T = 1} } \right]. %\label{eqn:BhattEstimate1}
& = \mathbb{E}_{T,V}\sqrt{\varphi(T,V)}. \notag
\end{align}
\end{proof}
\begin{lemma}[Stochastic Degradation~(cf.~\cite{koradaphd09})]\label{lemma:DegradationBhatt} Consider discrete random variables $V$, $Y_1$, and $Y_2$. Assume that $|\mathcal{V}| = 2$ and that discrete alphabets $\mathcal{Y}_1$ and $\mathcal{Y}_2$ have an \emph{arbitrary} size. Then
\begin{align}
& P_{Y_1|V}(y_1|v) \degBC P_{Y_2|V}(y_2|v) \Rightarrow Z(V|Y_2) \geq Z(V|Y_1). \label{eqn:DegradationBhatt}
\end{align}
\end{lemma}
\begin{proof} Beginning with the definition of the Bhattacharyya parameter leads to the following derivation:
\begin{align}
& Z(V|Y_2) \notag \\
& \triangleq 2 \sum_{y_2} \sqrt{P_{VY_2}(0,y_2)P_{VY_2}(1,y_2)} \notag \\
& = 2 \sum_{y_2} \sqrt{P_{V}(0)P_{V}(1)}\sqrt{P_{Y_2|V}(y_2|0)P_{Y_2|V}(y_2|1)} \notag \\
& = 2 \sqrt{P_{V}(0)P_{V}(1)} \sum_{y_2} \Biggl[\sqrt{ \sum_{y_1} P_{Y_1|V}(y_1|0) \tilde{P}_{Y_2|Y_1}(y_2|y_1) } \notag \\
&  ~~~~~~~~~~~~~~~~ \cdot \sqrt{ \sum_{y_1} P_{Y_1|V}(y_1|1) \tilde{P}_{Y_2|Y_1}(y_2|y_1) }\Biggl]. \notag
\end{align}
Then applying the Cauchy--Schwarz inequality yields
\begin{align}
& Z(V|Y_2) \notag \\
& \geq 2 \sqrt{P_{V}(0)P_{V}(1)} \sum_{y_2} \Biggl[ \sum_{y_1} \sqrt{ P_{Y_1|V}(y_1|0) \tilde{P}_{Y_2|Y_1}(y_2|y_1) } \notag \\
&  ~~~~~~~~~~~~~~~~ \cdot \sum_{y_1} \sqrt{ P_{Y_1|V}(y_1|1) \tilde{P}_{Y_2|Y_1}(y_2|y_1) }\Biggl] \notag \\
& = 2 \sqrt{P_{V}(0)P_{V}(1)} \sum_{y_2} \Biggl[ \sum_{y_1} \tilde{P}_{Y_2|Y_1}(y_2|y_1) \notag \\
& ~~~~~~~~~~~~~~~~~ \cdot \sqrt{ P_{Y_1|V}(y_1|0) P_{Y_1|V}(y_1|1) } \Biggl]. \notag
\end{align}
Interchanging the order of summations yields
\begin{align}
& Z(V|Y_2) \notag \\
& \geq 2 \sqrt{P_{V}(0)P_{V}(1)} \Biggl[ \sum_{y_1} \sqrt{ P_{Y_1|V}(y_1|0) P_{Y_1|V}(y_1|1) } \notag \\
& ~~~~~~~~~~~~~~~~~ \cdot  \sum_{y_2} \tilde{P}_{Y_2|Y_1}(y_2|y_1) \Biggl] \notag \\
& = Z(V|Y_1). \notag
\end{align}
\end{proof}

\begin{lemma}[Successive Stochastic Degradation~(cf.~\cite{koradaphd09})]\label{lemma:SuccessiveDegradationBhatt} Consider a binary random variable $V$, and discrete random variables $Y_1$ with alphabet $\mathcal{Y}_1$, and $Y_2$ with alphabet $\mathcal{Y}_2$. Assume that the joint distribution $P_{VY_1Y_2}$ obeys the constraint $P_{Y_1|V}(y_1|v) \degBC P_{Y_2|V}(y_2|v)$. Consider two $i.i.d.$ random copies $(V^{1}, Y_1^{1}, Y_2^{1})$ and $(V^{2}, Y_1^{2}, Y_2^{2})$ distributed according to $P_{VY_1Y_2}$. Define two binary random variables $U^{1} \triangleq V^{1} \oplus V^{2}$ and $U^{2} \triangleq V^{2}$. Then the following holds
\begin{align}
Z\left(U^{1}\bigl|Y_2^{1:2}\right) & \geq Z\left(U^{1}\bigl|Y_1^{1:2}\right), \label{eqn:SuccessiveDegradationBhatt1} \\
Z\left(U^{2}\bigl|U^{1}, Y_2^{1:2}\right) & \geq Z\left(U^{2}\bigl|U^{1}, Y_1^{1:2}\right). \label{eqn:SuccessiveDegradationBhatt2}
\end{align}
\end{lemma}
\begin{proof} Given the assumptions, the following stochastic degradation conditions hold:
\begin{align}
P_{Y_1^{1}|V^{1}}(y_1^{1}|v^{1}) & \degBC P_{Y_2^{1}|V^{1}}(y_2^{1}|v^{1}), \label{eqn:SuccessiveDAssumption1} \\
P_{Y_1^{2}|V^{2}}(y_1^{2}|v^{2}) & \degBC P_{Y_2^{2}|V^{2}}(y_2^{2}|v^{2}). \label{eqn:SuccessiveDAssumption2}
\end{align}
The goal is to derive new stochastic degradation conditions for the polarized conditional distributions. The binary random variables $U^{1}$ and $U^{2}$ are not necessarily independent Bernoulli($\frac{1}{2}$) variables. Taking this into account,
\begin{align}
& P_{Y_2^{1}Y_2^{2}|U^{1}}\bigl(y_2^{1},y_2^{2}\bigl| u^{1}\bigl) \notag \\
& ~~ = \frac{1}{P_{U^{1}}(u^{1})} \sum_{u^{2} \in \{0,1\}} P_{V^{1}Y_2^{1}}\bigl(u^{1} \oplus u^{2}, y_2^{1}\bigl)P_{V^{2}Y_2^{2}}\bigl(u^{2}, y_2^{2}\bigl) \notag \\
& ~~ = \frac{1}{P_{U^{1}}(u^{1})} \sum_{u^{2} \in \{0,1\}}\Biggl[ P_{Y_2^{1}|V^{1}}\bigl(y_2^{1}\bigl| u^{1} \oplus u^{2}\bigl)P_{V^{1}}\bigl(u^{1} \oplus u^{2}\bigl) \notag \\
& ~~~~~~~~~~~~~~~~~~~~~~~~~ \cdot P_{Y_2^{2}|V^{2}}\bigl(y_2^{2}\bigl|u^{2}\bigl)P_{V^{2}}(u^{2})\Biggl]. \notag
\end{align}
Applying the property due to the assumption in~\eqref{eqn:SuccessiveDAssumption1},
\begin{align}
& P_{Y_2^{1}Y_2^{2}|U^{1}}\bigl(y_2^{1},y_2^{2}\bigl| u^{1}\bigl) \notag \\
& ~~ = \frac{1}{P_{U^{1}}(u^{1})} \sum_{u^{2} \in \{0,1\}}\Biggl[ P_{V^{1}}\bigl(u^{1} \oplus u^{2}\bigl)P_{V^{2}}(u^{2}) \notag \\
& ~~~~ \cdot \sum_{a \in \mathcal{Y}_1} P_{Y_1^{1}|V^{1}}\bigl(a \bigl| u^{1} \oplus u^{2}\bigl) \tilde{P}_{Y_2^{1}|Y_1^{1}}\bigl(y_2^{1}\bigl| a) \notag \\
& ~~~~ \cdot \sum_{b \in \mathcal{Y}_1} P_{Y_1^{2}|V^{2}}\bigl(b \bigl| u^{2}\bigl)\tilde{P}_{Y_2^{2}|Y_1^{2}}\bigl(y_2^{2}\bigl| b\bigl) \Biggl]. \notag
\end{align}
Interchanging the order of summations and grouping the terms representing $P_{Y_1^{1}Y_1^{2}|U^{1}}\bigl(y_1^{1},y_1^{2}\bigl| u^{1}\bigl)$ yields the following
\begin{align}
& P_{Y_2^{1}Y_2^{2}|U^{1}}\bigl(y_2^{1},y_2^{2}\bigl| u^{1}\bigl) \notag \\
& = \sum_{a \in \mathcal{Y}_1, b \in \mathcal{Y}_1} P_{Y_1^{1}Y_1^{2}|U^{1}}\bigl(a, b \bigl| u^{1}\bigl) \tilde{P}_{Y_2^{1}|Y_1^{1}}\bigl(y_2^{1}\bigl| a) \tilde{P}_{Y_2^{2}|Y_1^{2}}\bigl(y_2^{2}\bigl| b\bigl). \notag
\end{align}
The above derivation proves that
\begin{align}
P_{Y_1^{1}Y_1^{2}|U^{1}}\bigl(y_1^{1},y_1^{2}\bigl| u^{1}\bigl) & \degBC P_{Y_2^{1}Y_2^{2}|U^{1}}\bigl(y_2^{1},y_2^{2}\bigl| u^{1}\bigl). \notag
\end{align}
Combined with Lemma~\ref{lemma:DegradationBhatt}, this concludes the proof for the ordering of the Bhattacharyya parameters given in~\eqref{eqn:SuccessiveDegradationBhatt1}.

In a similar way, it is possible to show that
\begin{align}
& P_{Y_2^{1}Y_2^{2}U^{1}|U^{2}}\bigl(y_2^{1},y_2^{2},u^{1}\bigl| u^{2}\bigl) \notag \\
& ~~ = \frac{1}{P_{U^{2}}(u^{2})} P_{V^{1}Y_2^{1}}\bigl(u^{1} \oplus u^{2}, y_2^{1}\bigl)P_{V^{2}Y_2^{2}}\bigl(u^{2}, y_2^{2}\bigl) \notag \\
& ~~ = \frac{1}{P_{U^{2}}(u^{2})} \Biggl[ P_{Y_2^{1}|V^{1}}\bigl(y_2^{1}\bigl| u^{1} \oplus u^{2}\bigl)P_{V^{1}}\bigl(u^{1} \oplus u^{2}\bigl) \notag \\
& ~~~~~~~~~~~~~~~~~~~~~~~~~ \cdot P_{Y_2^{2}|V^{2}}\bigl(y_2^{2}\bigl|u^{2}\bigl)P_{V^{2}}(u^{2})\Biggl]. \notag
\end{align}
Applying the property due to the assumption in~\eqref{eqn:SuccessiveDAssumption2},
\begin{align}
& P_{Y_2^{1}Y_2^{2}U^{1}|U^{2}}\bigl(y_2^{1},y_2^{2},u^{1}\bigl| u^{2}\bigl) \notag \\
& ~~ = \frac{1}{P_{U^{2}}(u^{2})} \Biggl[ P_{V^{1}}\bigl(u^{1} \oplus u^{2}\bigl)P_{V^{2}}(u^{2}) \notag \\
& ~~~~ \cdot \sum_{a \in \mathcal{Y}_1} P_{Y_1^{1}|V^{1}}\bigl(a \bigl| u^{1} \oplus u^{2}\bigl) \tilde{P}_{Y_2^{1}|Y_1^{1}}\bigl(y_2^{1}\bigl| a) \notag \\
& ~~~~ \cdot \sum_{b \in \mathcal{Y}_1} P_{Y_1^{2}|V^{2}}\bigl(b \bigl| u^{2}\bigl)\tilde{P}_{Y_2^{2}|Y_1^{2}}\bigl(y_2^{2}\bigl| b\bigl) \Biggl]. \notag
\end{align}
Interchanging the order of the terms and grouping the terms representing $P_{Y_1^{1}Y_1^{2}U^{1}|U^{2}}\bigl(y_1^{1},y_1^{2}, u^{1}\bigl| u^{2}\bigl)$ yields the following
\begin{align}
& P_{Y_2^{1}Y_2^{2}U^{1}|U^{2}}\bigl(y_2^{1},y_2^{2},u^{1}\bigl| u^{2}\bigl) \notag \\
& = \sum_{a \in \mathcal{Y}_1, b \in \mathcal{Y}_1} \Biggl[ P_{Y_1^{1}Y_1^{2}U^{1}|U^{2}}\bigl(a, b, u^{1} \bigl| u^{2}\bigl) \notag \\
& ~~~~~~~~~~~~~ \tilde{P}_{Y_2^{1}|Y_1^{1}}\bigl(y_2^{1}\bigl| a) \tilde{P}_{Y_2^{2}|Y_1^{2}}\bigl(y_2^{2}\bigl| b\bigl)\Biggl], \notag \\
& = \sum_{a \in \mathcal{Y}_1, b \in \mathcal{Y}_1, c \in \{0,1\}} \Biggl[ P_{Y_1^{1}Y_1^{2}U^{1}|U^{2}}\bigl(a, b, c \bigl| u^{2}\bigl) \notag \\
& ~~~~~~~~~~~~~ \tilde{P}_{Y_2^{1}|Y_1^{1}}\bigl(y_2^{1}\bigl| a) \tilde{P}_{Y_2^{2}|Y_1^{2}}\bigl(y_2^{2}\bigl| b\bigl)\indicator{u^{1} = c} \Biggl]. \notag
\end{align}
The above derivation proves that
\begin{align}
P_{Y_1^{1}Y_1^{2}U^{1}|U^{2}}\bigl(y_1^{1},y_1^{2},u^{1}\bigl| u^{2}\bigl) & \degBC P_{Y_2^{1}Y_2^{2}U^{1}|U^{2}}\bigl(y_2^{1},y_2^{2},u^{1}\bigl| u^{2}\bigl). \notag
\end{align}
Combined with Lemma~\ref{lemma:DegradationBhatt}, this concludes the proof for the ordering of the Bhattacharyya parameters given in~\eqref{eqn:SuccessiveDegradationBhatt2}.
\end{proof}

\begin{lemma}[Pinsker's Inequality]\label{lemma:Pinsker} Consider two discrete probability measures $P(y)$ and $Q(y)$ for $y \in \mathcal{Y}$. The following inequality holds for a constant $\kappa \triangleq 2\ln 2$.
\begin{align}
\sum_{y \in \mathcal{Y}} \Bigl| P(y) - Q(y) \Bigl| \leq \sqrt{\kappa D\left(P(y) \bigl\| Q(y)\right) }. \notag
\end{align}
\end{lemma}
\begin{lemma}[Ar\i kan~\cite{arikan10}]\label{lemma:ArikanLemma} Consider two discrete random variables $X \in \{0,1\}$ and $Y \in \mathcal{Y}$. The Bhattacharyya parameter and conditional entropy are related as follows.
\begin{align}
Z(X|Y)^{2} & \leq H(X|Y) \notag \\
H(X|Y) & \leq \log_2 (1 + Z(X|Y)) \notag
\end{align}
\end{lemma}
% YAMAMOTO MAY BE INCORRECT ON THESE PARTICULAR LEMMAS
%\begin{lemma}[Honda and Yamamoto~\cite{yamamoto12}]\label{lemma:YamamotoLemma} Consider two discrete random variables %$X \in \{0,1\}$ and $Y \in \mathcal{Y}$. The Bhattacharyya parameter and conditional entropy are related as follows.
%\begin{align}
%H(X|Y) & \geq 1 - \frac{1 - Z(X|Y)^{2}}{2\ln 2} \notag \\
%H(X|Y) & \leq \log_2(1 + (2P_X(0) - 1)^{2} + Z(X|Y)) \notag
%\end{align}
%\end{lemma}
\begin{lemma}[Bhattacharyya vs. Entropy Parameters]\label{lemma:ClosenessOfHandZOne} Consider two discrete random variables $X \in \{0,1\}$ and $Y \in \mathcal{Y}$. For any $0 < \delta < \frac{1}{2}$,
\begin{align}
Z(X|Y) \geq 1 - \delta & \Rightarrow H(X|Y) \geq 1 - 2\delta. \notag \\
Z(X|Y) \leq \delta & \Rightarrow H(X|Y) \leq \log_2(1 + \delta). \notag
\end{align}
\end{lemma}
\begin{proof} Due to Lemma~\ref{lemma:ArikanLemma}, $H(X|Y) \geq Z(X|Y)^{2} \geq (1 - \delta)^{2} \geq 1 - 2\delta + \delta^{2} \geq 1 - 2\delta$. It follows that if $Z(X|Y) \geq 1 - \delta$ and $\delta \rightarrow 0$, then $H(X|Y) \rightarrow 1$ as well. Similarly, due to Lemma~\ref{lemma:ArikanLemma}, taking constant $\kappa = \frac{1}{\log_e 2}$ and using the series expansion of $\log_e(1 + \delta)$, if $Z(X|Y) \leq \delta$ then $H(X|Y) \leq \log_2(1 + \delta) = \kappa\left(\sum_{k=1}^{\infty} (-1)^{k+1}\frac{\delta^{k}}{k}\right) \leq \kappa \delta$. It follows that if $Z(X|Y) \leq \delta$ and $\delta \rightarrow 0$, then $H(X|Y) \rightarrow 0$ as well.
\end{proof}

%%\begin{lemma}[Rate of Convergence~\cite{arikantelatar09}]\label{lemma:MartingaleRateOfConvergence} Let $\{B_{\ell}\}_{\ell \geq 1}$ be an $i.i.d.$ random process where $B_{\ell} \in \{-, +\}$. In addition, let $\{\Phi_{\ell}: \ell \geq 0\}$ be a positive random process satisfying
%%\begin{align}
%%\Phi_{\ell} & \leq \begin{cases} 2 \Phi_{\ell-1}, & \mbox{if}~ B_{\ell} = - \\ \Phi_{\ell - 1}^{2}, & \mbox{if}~B_{\ell} = +. \end{cases} \notag
%%\end{align}
%%Suppose that $\Phi_{\ell}$ converges almost surely to a $\{0,1\}$-valued random variable $\Phi_{\infty}$ with $\mathbb{P}\{\Phi_{\infty} = 0\} = p_{\infty}$. Then for any $\beta < \frac{1}{2}$,
%%\begin{align}
%%\lim_{\ell \rightarrow \infty} \mathbb{P}\{\Phi_{\ell} \leq 2^{-2^{\ell \beta}}\} = p_{\infty}. \notag
%%\end{align}
%%\end{lemma}

\section{Proof Of Lemma~\ref{lemma:TVBound}}\label{sec:AppendixTVBound}

The total variation bound of Lemma~\ref{lemma:TVBound} is decomposed in a simple way due to the chain rule for Kullback-Leibler distance between discrete probability measures. The joint probability measures $P$ and $Q$ were defined in~\eqref{eqn:UJOINT} and~\eqref{eqn:QUJOINT} respectively. According to definition, if $P\bigl( \{u_{i}^{1:n}\}_{i \in [m]}\bigl) > 0$ then $Q\bigl( \{u_{i}^{1:n}\}_{i \in [m]}\bigl) > 0$. Therefore the Kullback-Leibler distance $D(P \| Q)$ is well-defined and upper bounded as follows.
\begin{align}
& D\Bigl( P\bigl( \{u_{i}^{1:n}\}_{i \in [m]}\bigl) \Bigl\| Q\bigl( \{u_{i}^{1:n}\}_{i \in [m]}\bigl) \Bigl) & \notag \\
& =\sum_{i=1}^{m}\sum_{j=1}^{n} \Biggl[ D\Bigl( P\left(u_i(j) \Bigl| u_i^{1:j-1}, \{u_{k}^{1:n}\}_{k \in [1:i-1]} \right) \Bigl\| \notag \\
& ~~~~~~~~~~~~~~~~Q\left(u_i(j) \Bigl| u_i^{1:j-1}, \{u_{k}^{1:n}\}_{k \in [1:i-1]} \right) \Bigl)\Biggl] \label{eqn:FirstStep} \\
& =\sum_{i=1}^{m}\sum_{j \in \mathcal{M}_i^{(n)}} \Biggl[ D\Bigl( P\left(u_i(j) \Bigl| u_i^{1:j-1}, \{u_{k}^{1:n}\}_{k \in [1:i-1]} \right) \Bigl\|  \notag \\
& ~~~~~~~~~~~~~~~~~~~~Q\left(u_i(j) \Bigl| u_i^{1:j-1}, \{u_{k}^{1:n}\}_{k \in [1:i-1]} \right) \Bigl) \Biggl] \label{eqn:SecondStep} \\
& =\sum_{i=1}^{m}\sum_{j \in \mathcal{M}_i^{(n)}} 1 - H\left(U_i(j) \Bigl| U_i^{1:j-1}, \{U_{k}^{1:n}\}_{k \in [1:i-1]} \right) \label{eqn:ThirdStep} \\
& =\sum_{i=1}^{m}\sum_{j \in \mathcal{M}_i^{(n)}} 1 - H\left(U_i(j) \Bigl| U_i^{1:j-1}, \{Y_{k}^{1:n}\}_{k \in [1:i-1]} \right) \label{eqn:ThirdStepAgain} \\
& \leq \sum_{i=1}^{m} 2 \delta_n \left| \mathcal{M}_i^{(n)} \right|. \label{eqn:FourthStep}
\end{align}
The equality in~\eqref{eqn:FirstStep} is due to the chain rule for Kullback-Leibler distance. The equality in~\eqref{eqn:SecondStep} is valid because for indices $j \notin \mathcal{M}_{i}^{(n)}$, $P\left(u_i(j) \bigl| u_i^{1:j-1}, \{u_{k}^{1:n}\}_{k \in [1:i-1]} \right) = Q\left(u_i(j) \bigl| u_i^{1:j-1}, \{u_{k}^{1:n}\}_{k \in [1:i-1]} \right)$. The equality in~\eqref{eqn:ThirdStep} is valid due to Lemma~\ref{lemma:ConditionalEntropyAndKLDistance} and the fact that $Q\left(u_i(j) \bigl| u_i^{1:j-1}, \{u_{k}^{1:n}\}_{k \in [1:i-1]} \right) = \frac{1}{2}$ for indices $j \in \mathcal{M}_i^{(n)}$. The equality in~\eqref{eqn:ThirdStepAgain} follows due to the one-to-one correspondence between variables $\{U_{k}^{1:n}\}_{k \in [1:i-1]}$ and $\{Y_{k}^{1:n}\}_{k \in [1:i-1]}$. The last inequality~\eqref{eqn:FourthStep} follows from Lemma~\ref{lemma:ClosenessOfHandZOne} due to the fact that $Z\left(U_i(j) \bigl| U_i^{1:j-1}, \{Y_k^{1:n}\}_{k \in [1:i-1]}\right) \geq 1 - \delta_n$ for indices $j \in \mathcal{M}_i^{(n)}$.

To finish the proof of Lemma~\ref{lemma:TVBound},
\begin{align}
& \sum_{\{u_{k}^{1:n}\}_{k \in [m]}} \Bigl| P\bigl( \{u_{k}^{1:n}\}_{k \in [m]}\bigl) - Q\bigl(\{u_{k}^{1:n}\}_{k \in [m]}\bigl) \Bigl| \notag \\
& ~~~~~~~~~~\leq \sqrt{\kappa D\left( P\left( \{u_{k}^{1:n}\}_{k \in [m]}\right) \Bigl\| Q\left( \{u_{k}^{1:n}\}_{k \in [m]} \right) \right)} ~~~~~~~ \label{eqn:PinskerStep} \\
& ~~~~~~~~~~\leq \sqrt{\kappa \sum_{i=1}^{m} 2 \delta_n \left| \mathcal{M}_i^{(n)} \right|} ~~~~~~~ \label{eqn:KLBoundStep} \\
& ~~~~~~~~~~\leq \sqrt{(2\kappa)(m\cdot n)(2^{-n^{\beta^{\prime}}})}.~~~~~~~ \notag
\end{align}
The inequality in~\eqref{eqn:PinskerStep} is due to Pinsker's inequality given in Lemma~\ref{lemma:Pinsker}. The inequality in~\eqref{eqn:KLBoundStep} was proven in~\eqref{eqn:FourthStep}. Finally for $\beta^{\prime} \in (\beta, \frac{1}{2})$, $\sqrt{(2\kappa)(m\cdot n)(2^{-n^{\beta^{\prime}}})} < 2^{-n^{\beta}}$ for sufficiently large $n$. Hence the total variation distance is bounded by $\mathcal{O}(2^{-n^{\beta}})$ for any $0 < \beta < \frac{1}{2}$.

\section{Superposition Coding}\label{sec:AppendixSuperposition}

The total variation bound of Lemma~\ref{lemma:TVBoundSuperposition} is decomposed in a simple way due to the chain rule for Kullback-Leibler distance between discrete probability measures. The joint probability measures $P$ and $Q$ were defined in~\eqref{eqn:U1U2DecomposedDistributionSuperposition} and~\eqref{eqn:QUJOINTSuperposition} respectively. According to definition, if $P_{U_1^{n}U_2^{n}}\bigl(u_1^{n}, u_2^{n}\bigl) > 0$ then $Q\bigl(u_1^{n}, u_2^{n}\bigl) > 0$. Therefore the Kullback-Leibler distance $D(P \| Q)$ is well-defined. Applying the chain rule,
\begin{align}
& D\Bigl( P_{U_1^{n}U_2^{n}}\bigl(u_1^{n}, u_2^{n}\bigl) \Bigl\| Q\bigl(u_1^{n}, u_2^{n}\bigl) \Bigl) & \notag \\
& =\sum_{j=1}^{n} D\Bigl( P\left(u_1(j) \Bigl| u_1^{1:j-1}\right) \Bigl\| Q\left(u_1(j) \Bigl| u_1^{1:j-1}\right) \Bigl) \notag \\
& + \sum_{j=1}^{n} D\Bigl( P\left(u_2(j) \Bigl| u_2^{1:j-1}, u_1^{n}\right) \Bigl\| Q\left(u_2(j) \Bigl| u_2^{1:j-1}, u_1^{n}\right) \Bigl) \notag \\
& =\sum_{j \in \mathcal{M}_1^{(n)}} D\Bigl( P\left(u_1(j) \Bigl| u_1^{1:j-1}\right) \Bigl\| Q\left(u_1(j) \Bigl| u_1^{1:j-1}\right) \Bigl) \notag \\
& + \sum_{j \in \mathcal{M}_2^{(n)}} D\Bigl( P\left(u_2(j) \Bigl| u_2^{1:j-1}, u_1^{n}\right) \Bigl\| Q\left(u_2(j) \Bigl| u_2^{1:j-1}, u_1^{n}\right)\Bigl). \notag
\end{align}
Applying Lemma~\ref{lemma:ConditionalEntropyAndKLDistance}, the one-to-one relation between $U_1^{n}$ and $V^{n}$, and Lemma~\ref{lemma:ClosenessOfHandZOne} leads to the following result.
\begin{align}
& D\Bigl( P_{U_1^{n}U_2^{n}}\bigl(u_1^{n}, u_2^{n}\bigl) \Bigl\| Q\bigl(u_1^{n}, u_2^{n}\bigl) \Bigl) & \notag \\
& = \sum_{j \in \mathcal{M}_1^{(n)}} \Bigl[ 1 - H\left(U_1(j) \Bigl| U_1^{1:j-1} \right) \Biggl] + \notag \\
& ~~~~~~~~ \sum_{j \in \mathcal{M}_2^{(n)}} \Biggl[ 1 - H\left(U_2(j) \Bigl| U_2^{1:j-1} U_1^{n} \right) \Biggl] \notag \\
& = \sum_{j \in \mathcal{M}_1^{(n)}} \Bigl[ 1 - H\left(U_1(j) \Bigl| U_1^{1:j-1} \right) \Biggl] + \notag \\
& ~~~~~~~~ \sum_{j \in \mathcal{M}_2^{(n)}} \Biggl[ 1 - H\left(U_2(j) \Bigl| U_2^{1:j-1} V^{n} \right) \Biggl] \notag \\
& \leq 2 \delta_n \Biggl[ \left| \mathcal{M}_1^{(n)} \right| + \left| \mathcal{M}_2^{(n)} \right| \Biggl]. \notag
\end{align}
Using identical arguments as applied in the proof of Lemma~\ref{lemma:TVBound}, the total variation distance between $P$ and $Q$ is bounded as $\mathcal{O}(2^{-n^{\beta}})$.
%Over indices $j \in \mathcal{M}_1^{(n)}$, $H\left(U_1(j) \Bigl| U_1^{1:j-1} \right) \geq 1 - \delta_n$ and similarly, over %indices $j \in \mathcal{M}_2^{(n)}$, $H\left(U_2(j) \Bigl| U_2^{1:j-1} V^{n} \right) \geq 1 - \delta_n$. Therefore,
%\begin{align}
% & D\Bigl( P_{U_1^{n}U_2^{n}}\bigl(u_1^{n}, u_2^{n}\bigl) \Bigl\| Q\bigl(u_1^{n}, u_2^{n}\bigl) \Bigl) \leq 2 \delta_n %\left| \mathcal{M}_1^{(n)} \right| + 2 \delta_n \left| \mathcal{M}_2^{(n)} \right|.
%\end{align}

To prove Lemma~\ref{theorem:ErrorProbSuperposition}, the expectation of the average probability of error of the polarization-based superposition code is written as
\begin{align}
& \mathbb{E}_{\{\Psi_1^{(j)}, \Psi_2^{(j)}\}} \Bigl[ P_e^{(n)}[\{\Psi_1^{(j)}, \Psi_2^{(j)}\}] \Bigl] = \notag \\
& ~~ \sum_{\{u_1^{n}, u_2^{n}, y_1^{n}, y_2^{n}\} \in \mathcal{T}} \Biggl[ P_{Y_1^{n}Y_2^{n}\bigl|U_1^{n}U_2^{n}}\bigl(y_1^{n}, y_2^{n}\bigl|u_1^{n}, u_2^{n}\bigl) \notag \\
& ~~ \cdot \frac{1}{2^{nR_2}} \prod_{j \in [n]: j \notin \mathcal{M}_2^{(n)}} \mathbb{P}\left\{\Psi_2^{(j)}\left(u_2^{1:j-1}\right) = u_2(j)\right\} \notag \\
& ~~ \cdot \frac{1}{2^{nR_1}} \prod_{j \in [n]: j \notin \mathcal{M}_1^{(n)}} \mathbb{P}\left\{\Psi_1^{(j)}\left(u_1^{1:j-1}, u_2^{n}\matbold{G}_n \right) = u_1(j)\right\}\Biggl]. \notag
\end{align}
From the definitions of the random boolean functions $\Psi_1^{(j)}$ in~\eqref{eqn:RANDMapSpecificSuperposition1} and $\Psi_2^{(j)}$ in~\eqref{eqn:RANDMapSpecificSuperposition2}, it follows that
\begin{align}
& \mathbb{P}\left\{\Psi_1^{(j)}\left(u_1^{1:j-1}, u_2^{n}\matbold{G}_n \right) = u_1(j)\right\} \notag \\
& ~~ = \mathbb{P}\left\{U_1(j) = u_1(j) \bigl| U_1^{1:j-1} = u_1^{1:j-1}, V^{n} = u_2^{n}\matbold{G}_n \right\} \notag \\
& ~~ = \mathbb{P}\left\{U_1(j) = u_1(j) \bigl| U_1^{1:j-1} = u_1^{1:j-1}, U_2^{n} = u_2^{n} \right\}, \notag \\
& \mathbb{P}\left\{\Psi_2^{(j)}\left(u_2^{1:j-1}\right) = u_2(j)\right\} \notag \\
& ~~ = \mathbb{P}\left\{U_2(j) = u_2(j) \bigl| U_2^{1:j-1} = u_2^{1:j-1}\right\}. \notag
\end{align}
The expression for the expected average probability of error is then simplified by substituting the definition for $Q(u_1^{n}, u_2^{n})$ provided in~\eqref{eqn:QUJOINTSuperposition} as follows,
\begin{align}
& \mathbb{E}_{\{\Psi_1^{(j)}, \Psi_2^{(j)}\}} \Bigl[ P_e^{(n)}[\{\Psi_1^{(j)}, \Psi_2^{(j)}\}] \Bigl] = \notag \\
& \sum_{\{u_1^{n}, u_2^{n}, y_1^{n}, y_2^{n}\} \in \mathcal{T}} \Biggl[ P_{Y_1^{n}Y_2^{n}\bigl|U_1^{n}U_2^{n}}\bigl(y_1^{n}, y_2^{n}\bigl|u_1^{n}, u_2^{n}\bigl) Q(u_1^{n}, u_2^{n}) \Biggl]. \notag
\end{align}
The next step in the proof is to split the error term $\mathbb{E}_{\{\Psi_1^{(j)}, \Psi_2^{(j)}\}} \Bigl[ P_e^{(n)}[\{\Psi_1^{(j)}, \Psi_2^{(j)}\}] \Bigl]$ into \emph{two} main parts, one part due to the error caused by polar decoding functions, and the other part due to the total variation distance between probability measures.
\begin{align}
& \mathbb{E}_{\{\Psi_1^{(j)}, \Psi_2^{(j)}\}} \Bigl[ P_e^{(n)}[\{\Psi_1^{(j)}, \Psi_2^{(j)}\}] \Bigl] \notag \\
& = \sum_{\{u_1^{n}, u_2^{n}, y_1^{n}, y_2^{n}\} \in \mathcal{T}} \Biggl[ P_{Y_1^{n}Y_2^{n}\bigl|U_1^{n}U_2^{n}}\bigl(y_1^{n}, y_2^{n}\bigl|u_1^{n}, u_2^{n}\bigl) \notag \\
& ~~ \cdot \biggl(Q\bigl(u_1^{n}, u_2^{n}\bigl) - P_{U_1^{n}U_2^{n}}\bigl(u_1^{n},u_2^{n}\bigl) + P_{U_1^{n}U_2^{n}}\bigl(u_1^{n},u_2^{n}\bigl) \biggl) \Biggl] \notag \\
& \leq \Biggl[\sum_{\{u_1^{n}, u_2^{n}, y_1^{n}, y_2^{n}\} \in \mathcal{T}} P_{U_1^{n}U_2^{n}Y_1^{n}Y_2^{n}}\bigl(u_1^{n}, u_2^{n}, y_1^{n}, y_2^{n}\bigl)\Biggl] \notag \\
& ~~~~~~~~ + \Biggl[ \sum_{\begin{subarray}{c} u_1^{n} \in \{0,1\}^{n} \\ u_2^{n} \in \{0,1\}^{n} \end{subarray}} \Bigl| P_{U_1^{n}U_2^{n}}\bigl( u_1^{n}, u_2^{n} \bigl) - Q\bigl( u_1^{n}, u_2^{n} \bigl) \Bigl| \Biggl]. \label{eqn:TwoErrorPartsSuperposition}
\end{align}
Lemma~\ref{lemma:TVBoundSuperposition} established that the error term due to the total variation distance is upper bounded as $\mathcal{O}(2^{-n^{\beta}})$. Therefore, it remains to upper bound the error term due to the polar decoding functions. Towards this end, note first that $\mathcal{T} = \mathcal{T}_{1v} \cup \mathcal{T}_{1} \cup \mathcal{T}_{2}$, $\mathcal{T}_{1v} = \cup_{j} \mathcal{T}_{1v}^{j}$ for $j \in \mathcal{M}_{2}^{(n)} \subseteq \mathcal{M}_{1v}^{(n)}$, $\mathcal{T}_{1} = \cup_{j} \mathcal{T}_{1}^{j}$ for $j \in \mathcal{M}_{1}^{(n)}$, and $\mathcal{T}_{2} = \cup_{j} \mathcal{T}_{2}^{j}$ for $j \in \mathcal{M}_{2}^{(n)}$. It is convenient to bound each type of error bit by bit successively at both decoder $\mathcal{D}_1$ and $\mathcal{D}_2$ as follows.
\begin{align}
\mathcal{E}_{1v}^{j} & \triangleq \sum_{\{u_1^{n}, u_2^{n}, y_1^{n}, y_2^{n}\} \in \mathcal{T}_{1v}^{j}} P_{U_1^{n}U_2^{n}Y_1^{n}Y_2^{n}}\bigl(u_1^{n}, u_2^{n}, y_1^{n}, y_2^{n}\bigl) \notag \\
%& = \sum_{\begin{subarray}{c} (u_2^{n}, y_1^{n}) \in \{0,1\}^{n} \times \{0,1\}^{n} \\ (u_1^{n}, y_2^{n}) \in \{0,1\}^{n} %\times \{0,1\}^{n} \end{subarray}} P_{U_1^{n}U_2^{n}Y_1^{n}Y_2^{n}}\bigl(u_1^{n}, u_2^{n}, y_1^{n}, y_2^{n}\bigl) \notag %\\
%& ~~~~ \cdot \mathbbm{1}\Biggl[P_{U_2^{j}\bigl|U_2^{1:j-1}Y_1^{n}}\bigl(u_2(j)\bigl|u_2^{1:j-1},y_1^{n}\bigl) \leq \notag %\\
%& ~~~~~~~~~~ P_{U_2^{j}\bigl|U_2^{1:j-1}Y_1^{n}}\bigl(u_2(j) \oplus 1 \bigl| u_2^{1:j-1}, y_1^{n}\bigl)\Biggl], \notag \\
& = \sum_{\begin{subarray}{c} (u_2^{1:j}, y_1^{n}) \in \{0,1\}^{j} \times \mathcal{Y}_1^{n} \end{subarray}} P_{U_2^{1:j}Y_1^{n}}\bigl(u_2^{1:j}, y_1^{n}\bigl) \notag \\
& ~~~~~~~~ \cdot \mathbbm{1}\Biggl[P_{U_2^{j}\bigl|U_2^{1:j-1}Y_1^{n}}\bigl(u_2(j)\bigl|u_2^{1:j-1},y_1^{n}\bigl) \leq \notag \\
& ~~~~~~~~~~~~~~ P_{U_2^{j}\bigl|U_2^{1:j-1}Y_1^{n}}\bigl(u_2(j) \oplus 1 \bigl| u_2^{1:j-1}, y_1^{n}\bigl)\Biggl]. \notag
\end{align}
In this form, it is possible to upper bound the error term $\mathcal{E}_{1v}^{j}$ with the corresponding Bhattacharyya parameter as follows,
\begin{align}
\mathcal{E}_{1v}^{j} & = \sum_{\begin{subarray}{c} u_2^{1:j} \in \{0,1\}^{j} \\ y_1^{n} \in \mathcal{Y}_1^{n} \end{subarray}} P\bigl(u_2^{1:j-1}, y_1^{n}\bigl) P\bigl(u_2^{j}\bigl| u_2^{1:j-1}, y_1^{n}\bigl)\notag \\
& ~~~~ \cdot \mathbbm{1}\Biggl[P_{U_2^{j}\bigl|U_2^{1:j-1}Y_1^{n}}\bigl(u_2(j)\bigl|u_2^{1:j-1},y_1^{n}\bigl) \leq \notag \\
& ~~~~~~~~~~ P_{U_2^{j}\bigl|U_2^{1:j-1}Y_1^{n}}\bigl(u_2(j) \oplus 1 \bigl| u_2^{1:j-1}, y_1^{n}\bigl)\Biggl], \notag \\
& \leq \sum_{\begin{subarray}{c} u_2^{1:j} \in \{0,1\}^{j} \\ y_1^{n} \in \mathcal{Y}_1^{n} \end{subarray}} P\bigl(u_2^{1:j-1}, y_1^{n}\bigl) P\bigl(u_2^{j}\bigl| u_2^{1:j-1}, y_1^{n}\bigl) \notag \\
& ~~~~ \cdot \sqrt{ \frac{ P_{U_2^{j}\bigl|U_2^{1:j-1}Y_1^{n}}\bigl(u_2(j) \oplus 1 \bigl| u_2^{1:j-1}, y_1^{n}\bigl) } { P_{U_2^{j}\bigl|U_2^{1:j-1}Y_1^{n}}\bigl(u_2(j) \bigl| u_2^{1:j-1}, y_1^{n}\bigl) } } \notag \\
& = Z\bigl( U_2^{j} \bigl| U_2^{1:j-1}, Y_1^{n}\bigl). \notag
\end{align}
Using identical arguments, the following upper bounds apply for the individual bit-by-bit error terms caused by successive decoding at both $\mathcal{D}_1$ and $\mathcal{D}_2$.
\begin{align}
\mathcal{E}_{1v}^{j} & \leq Z\bigl( U_2^{j} \bigl| U_2^{1:j-1}, Y_1^{n}\bigl), \label{eqn:Z1vSuperposition} \\
\mathcal{E}_{1}^{j} & \leq Z\bigl( U_1^{j} \bigl| U_1^{1:j-1}, V^{n}, Y_1^{n}\bigl), \label{eqn:Z1Superposition} \\
\mathcal{E}_{2}^{j} & \leq Z\bigl( U_2^{j} \bigl| Y_2^{n}\bigl). \label{eqn:Z2Superposition}
\end{align}
Therefore, the total error due to decoding at the receivers is upper bounded as
\begin{align}
\mathcal{E} & \triangleq \sum_{\{u_1^{n}, u_2^{n}, y_1^{n}, y_2^{n}\} \in \mathcal{T}} P_{U_1^{n}U_2^{n}Y_1^{n}Y_2^{n}}\bigl(u_1^{n}, u_2^{n}, y_1^{n}, y_2^{n}\bigl) \notag \\
& \leq \sum_{j \in \mathcal{M}_2^{(n)} \subseteq \mathcal{M}_{1v}^{(n)}} Z\bigl( U_2^{j} \bigl| U_2^{1:j-1}, Y_1^{n}\bigl) \notag \\
& ~~~~~~~~~~ + \sum_{j \in \mathcal{M}_1^{(n)}} Z\bigl( U_1^{j} \bigl| U_1^{1:j-1}, V^{n}, Y_1^{n}\bigl) \notag \\
& ~~~~~~~~~~ + \sum_{j \in \mathcal{M}_2^{(n)}} Z\bigl( U_2^{j} \bigl| Y_2^{n}\bigl) \notag \\
& \leq \delta_n \Biggl[ \left| \mathcal{M}_{1v}^{(n)} \right| + \left| \mathcal{M}_1^{(n)} \right| + \left| \mathcal{M}_2^{(n)} \right| \Biggl] \notag \\
& \leq 3n \delta_n. \notag
\end{align}
This concludes the proof demonstrating that the expected average probability of error is upper bounded as $\mathcal{O}(2^{-n^{\beta}})$.

\section{Marton Coding}\label{sec:AppendixMarton}

To prove Lemma~\ref{theorem:ErrorProbMarton}, the expectation of the average probability of error of the polarization-based Marton code is written as
\begin{align}
& \mathbb{E}_{\{\Psi_1^{(j)}, \Psi_2^{(j)}, \Gamma\}} \Bigl[ P_e^{(n)}[\{\Psi_1^{(j)}, \Psi_2^{(j)}, \Gamma\}] \Bigl] = \notag \\
& ~~ \sum_{\{u_1^{n}, u_2^{n}, y_1^{n}, y_2^{n}\} \in \mathcal{T}} \Biggl[ P^{\phi}_{Y_1^{n}Y_2^{n}\bigl|U_1^{n}U_2^{n}}\bigl(y_1^{n}, y_2^{n}\bigl|u_1^{n}, u_2^{n}\bigl) \notag \\
& ~~ \cdot \frac{1}{2^{nR_1}} \prod_{j \in [n]: j \notin \mathcal{M}_1^{(n)}} \mathbb{P}\Bigl\{\Psi_1^{(j)}\left(u_1^{1:j-1}\right) = u_1(j)\Bigl\} \notag \\
& ~~ \cdot \frac{1}{2^{nR_2}} \prod_{j \in \mathcal{H}^{(n)}_{V_2|V_1} \backslash \mathcal{M}_2^{(n)}} \mathbb{P}\Bigl\{\Gamma(j) = u_2(j)\Bigl\} \notag \\
& ~~ \cdot \prod_{j \in [n]: j \notin \mathcal{H}^{(n)}_{V_2|V_1}} \mathbb{P}\Bigl\{\Psi_2^{(j)}\left(u_2^{1:j-1}, u_1^{n}\matbold{G}_n \right) = u_2(j)\Bigl\}\Biggl]. \notag
\end{align}
The expression is then simplified by substituting the definition of $Q(u_1^{n}, u_2^{n})$ provided in~\eqref{eqn:QUJOINTMarton}, and then splitting the error term into two parts:
\begin{align}
& \mathbb{E}_{\{\Psi_1^{(j)}, \Psi_2^{(j)}, \Gamma\}} \Bigl[ P_e^{(n)}[\{\Psi_1^{(j)}, \Psi_2^{(j)}, \Gamma\}] \Bigl] = \notag \\
& \sum_{\{u_1^{n}, u_2^{n}, y_1^{n}, y_2^{n}\} \in \mathcal{T}} \Biggl[ P^{\phi}_{Y_1^{n}Y_2^{n}\bigl|U_1^{n}U_2^{n}}\bigl(y_1^{n}, y_2^{n}\bigl|u_1^{n}, u_2^{n}\bigl) Q(u_1^{n}, u_2^{n}) \Biggl], \notag \\
& \leq \Biggl[\sum_{\{u_1^{n}, u_2^{n}, y_1^{n}, y_2^{n}\} \in \mathcal{T}} P_{U_1^{n}U_2^{n}Y_1^{n}Y_2^{n}}\bigl(u_1^{n}, u_2^{n}, y_1^{n}, y_2^{n}\bigl)\Biggl] \notag \\
& ~~~~~~~~ + \Biggl[ \sum_{\begin{subarray}{c} u_1^{n} \in \{0,1\}^{n} \\ u_2^{n} \in \{0,1\}^{n} \end{subarray}} \Bigl| P_{U_1^{n}U_2^{n}}\bigl( u_1^{n}, u_2^{n} \bigl) - Q\bigl( u_1^{n}, u_2^{n} \bigl) \Bigl| \Biggl]. \notag
\end{align}
The error term pertaining to the total variation distance was already upper bounded as in Lemma~\ref{lemma:TVBoundMarton}. The error due to successive cancelation decoding at the receivers is upper bounded as follows.
\begin{align}
\mathcal{E} & \triangleq \sum_{\{u_1^{n}, u_2^{n}, y_1^{n}, y_2^{n}\} \in \mathcal{T}} P_{U_1^{n}U_2^{n}Y_1^{n}Y_2^{n}}\bigl(u_1^{n}, u_2^{n}, y_1^{n}, y_2^{n}\bigl) \notag \\
& \leq \sum_{j \in \mathcal{M}_1^{(n)}} Z\bigl( U_1^{j} \bigl| U_1^{1:j-1}, Y_1^{n}\bigl) + \sum_{j \in \mathcal{L}_{V_2|Y_2}^{(n)}} Z\bigl( U_2^{j} \bigl| U_2^{1:j-1}, Y_2^{n}\bigl), \notag \\
& \leq \delta_n \Biggl[ \left| \mathcal{M}_1^{(n)} \right| + \left| \mathcal{L}_{V_2|Y_2}^{(n)} \right| \Biggl] \notag \\
& \leq 2n \delta_n. \notag
\end{align}
This concludes the proof demonstrating that the expectation of the average probability of block error is upper bounded as $\mathcal{O}(2^{-n^{\beta}})$.

\section{Proof Of Lemma~\ref{lemma:SpecialClassesDMBCsProperties}}\label{sec:AppendixLemmaSpecialClassesDMBCsProperties}

The implication in~\eqref{eqn:LemmaFirstI} follows since $X-Y_1-Y_2$ means that $P_{Y_2|X}(y_2|x) = \sum_{y_1}P_{Y_1|X}(y_1|x)P_{Y_2|Y_1}(y_2|y_1)$. The implication in~\eqref{eqn:LemmaSecondI} follows by observing that
\begin{align}
& P_{Y_2|V}(y_2|v) \notag \\
& ~~~~= \sum_{y_1 \in \mathcal{Y}_1} P_{Y_1Y_2|V}(y_1, y_2|v) \notag \\
& ~~~~= \sum_{x \in \mathcal{X}} \sum_{y_1 \in \mathcal{Y}_1} P_{X|V}(x|v) P_{Y_1Y_2|X}(y_1, y_2|x) \notag \\
& ~~~~= \sum_{x \in \mathcal{X}} P_{X|V}(x|v) \sum_{y_1 \in \mathcal{Y}_1} P_{Y_1Y_2|X}(y_1, y_2|x) \notag \\
& ~~~~= \sum_{x \in \mathcal{X}} P_{X|V}(x|v) P_{Y_2|X}(y_2|x) \notag \\
& ~~~~= \sum_{x \in \mathcal{X}} P_{X|V}(x|v) \sum_{y_1 \in \mathcal{Y}_1} P_{Y_1|X}(y_1|x) \tilde{P}_{Y_2|Y_1}(y_2|y_1) \label{eqn:LemmaDegradedStep} \\
& ~~~~= \sum_{y_1 \in \mathcal{Y}_1} \sum_{x \in \mathcal{X}} P_{X|V}(x|v)P_{Y_1|X}(y_1|x) \tilde{P}_{Y_2|Y_1}(y_2|y_1) \notag \\
& ~~~~= \sum_{y_1 \in \mathcal{Y}_1} P_{Y_1|V}(y_1|v) \tilde{P}_{Y_2|Y_1}(y_2|y_1). \notag
\end{align}
In step~\eqref{eqn:LemmaDegradedStep}, the assumed stochastic degraded condition $P_{Y_1|X}(y_1|x) \degBC P_{Y_2|X}(y_2|x)$ ensures the existence of the distribution $\tilde{P}_{Y_2|Y_1}(y_2|y_1)$. The converse to~\eqref{eqn:LemmaSecondI} follows since it is possible to select $P_{X|V}(x|v) = \indicator{x = v}$ where the alphabet $\mathcal{V} = \mathcal{X}$. In this case, for any $v \in \mathcal{X}$,
\begin{align}
P_{Y_2|V}(y_2|v) & = \sum_{x \in \mathcal{X}} P_{X|V}(x|v)P_{Y_2|X}(y_2|x) \notag \\
& = \sum_{x \in \mathcal{X}} \indicator{x = v} P_{Y_2|X}(y_2|x) \notag \\
& = P_{Y_2|X}(y_2|v). \notag
\end{align}
Similarly, $P_{Y_1|V}(y_1|v) = P_{Y_1|X}(y_1|v)$ for any $v \in \mathcal{X}$. Due to the assumed stochastic degradedness condition $P_{Y_2|V}(y_2|v) = \sum_{y_1} P_{Y_1|V}(y_1|v) \tilde{P}_{Y_2|Y_1}(y_2|y_1)$, for any $v \in \mathcal{X}$,
\begin{align}
P_{Y_2|X}(y_2|v) & = P_{Y_2|V}(y_2|v) \notag \\
& = \sum_{y_1} P_{Y_1|V}(y_1|v) \tilde{P}_{Y_2|Y_1}(y_2|y_1) \notag \\
& = \sum_{y_1} P_{Y_1|X}(y_1|v) \tilde{P}_{Y_2|Y_1}(y_2|y_1). \notag
\end{align}
Therefore the stochastic degradedness property $P_{Y_1|X}(y_1|x) \degBC P_{Y_2|X}(y_2|x)$ must hold as well. The statement of~\eqref{eqn:LemmaSecondI} means that Class $I$ and Class $II$ are equivalent as shown in Figure~\ref{fig:BroadcastChannelClasses}. The implication in~\eqref{eqn:LemmaThirdI} follows because assuming the stochastic degradedness property $P_{Y_1|V}(y_1|v) \degBC P_{Y_2|V}(y_2|v)$ holds for all $P_{X|V}(x|v)$, there exists a $\tilde{Y}_1$ such that $V-\tilde{Y}_1-Y_2$ form a Markov chain and $P_{\tilde{Y}_1|V}(\tilde{y}_1|v) = P_{Y_1|V}(\tilde{y}_1|v)$ for all $P_{X|V}(x|v)$. By the data processing inequality, $I(V;\tilde{Y}_1) \geq I(V;Y_2)$. If $P_{\tilde{Y}_1|V}(\tilde{y}_1|v) = P_{Y_1|V}(\tilde{y}_1|v)$, then $P_{V\tilde{Y}_1}(v,\tilde{y}_1) = P_{V Y_1}(v,\tilde{y}_1)$ for all $P_{V}(v)$. It follows that for all $P_{VX}(v,x)$, the mutual information $I(V;\tilde{Y}_1) = I(V;Y_1)$. The implication in~\eqref{eqn:LemmaFourthI} follows by setting $P_{VX}(v,x) = \indicator{v = x}P_{X}(x)$ and letting $\mathcal{V} = \mathcal{X}$. Then for any $v \in \mathcal{X}$,
\begin{align}
P_{VY_1}(v, y_1) & = \sum_{x \in \mathcal{X}} P_{VX}(v,x) P_{Y_1|X}(y_1|x) \notag \\
& = \sum_{x \in \mathcal{X}} \indicator{v = x} P_{X}(x) P_{Y_1|X}(y_1|x) \notag \\
& = P_{X}(v)P_{Y_1|X}(y_1|v) \notag \\
& = P_{XY_1}(v, y_1). \notag
\end{align}
Similarly for any $v \in \mathcal{X}$, $P_{VY_2}(v, y_2) = P_{XY_2}(v,y_2)$. Therefore for the particular choice of $P_{VX}(v,x) = \indicator{v = x}P_{X}(x)$, $I(V;Y_1) = I(X;Y_1)$ and $I(V;Y_2) = I(X;Y_2)$. The converse statements for~\eqref{eqn:LemmaFirstI},~\eqref{eqn:LemmaThirdI}, and~\eqref{eqn:LemmaFourthI} do \emph{not} hold due to a counterexample involving a DM-BC comprised of a binary erasure channel \textsc{BEC}($\epsilon$) and a binary symmetric channel \textsc{BSC}($p$) as described in Example~\ref{ex:DMBCWithBECAndBSC}.
%The implication in~\eqref{eqn:LemmaFourthI} is well-known in the sense that all less noisy broadcast channels have the more capable %property. The converse statements do not hold via a counter-example-- a two-user DM-BC constructed using a $\textsc{BSC}_{p}$ and %$\textsc{BEC}_{\epsilon}$~\cite[Example 5.4]{elgamalkim2010}.
%\$P_{Y_1|V}(y_1|v) = P_{X|V}(x|v)P_{Y_1|X}(y_1|x)$ and that $P_{Y_2|V}(y_2|v) = %P_{X|V}(x|v)P_{Y_1|X}(y_1|x)\tilde{P}_{Y_2|Y_1}(y_2|y_1)$ due to the given condition that $P_{Y_1|X}(y_1|x) \degBC P_{Y_2|X}(y_2|x)$.

%\vspace{-0.025in}
%\section*{Acknowledgment} %\footnotesize
%The authors thank S. B. Korada for the helpful discussions regarding the journal paper~\cite{korada10}. \vspace{-0.02in}

%\section{Conclusion}
%The conclusion goes here.

% or


\appendices

%\section{}

%\section*{Acknowledgment}


\ifCLASSOPTIONcaptionsoff
  \newpage
\fi



%\bibliographystyle{IEEEtran}
%\begin{thebibliography}{1}


%\end{thebibliography}

%\small
\bibliographystyle{ieeetr}
\bibliography{naveenBIB}

%\begin{biography}[{\includegraphics[width=1in,height=1.25in,clip,keepaspectratio]{mshell}}]{Michael Shell}

%\end{IEEEbiography}


%\newpage



%\vfill


\end{document}